\documentclass[onecolumn,11pt]{IEEEtran}

\usepackage{xcolor}

\usepackage{amsfonts}                                                          
\usepackage{epsfig}
\usepackage{amsthm}
\usepackage{graphicx}        
\usepackage{latexsym}
\usepackage{amssymb}
\usepackage{amsmath}
\usepackage{parskip}
\usepackage{multirow}
\usepackage{cite}
\usepackage{mathtools}
\usepackage{epstopdf}
\usepackage{etoolbox}
\usepackage{array}
\usepackage{mathptmx}
\usepackage[colorlinks=true]{hyperref}
\usepackage[bottom]{footmisc}
\usepackage{tikz}

\usetikzlibrary{decorations.pathreplacing}

\usepackage{verbatim}
\usepackage{calrsfs}
\usepackage{booktabs}

\DeclareMathAlphabet{\pazocal}{OMS}{zplm}{m}{n}

\newcommand{\mb}{\mathbb}

\let\bbordermatrix\bordermatrix
\patchcmd{\bbordermatrix}{8.75}{4.75}{}{}
\patchcmd{\bbordermatrix}{\left(}{\left[}{}{}
\patchcmd{\bbordermatrix}{\right)}{\right]}{}{}

\newcommand{\sr}{\stackrel}

\newcommand{\rar}{\rightarrow}

\newcommand{\tri}{\sr{\triangle}{=}}

\newcommand{\be}{\begin{equation}}
\newcommand{\ee}{\end{equation}}
\newcommand{\bea}{\begin{eqnarray}}
\newcommand{\eea}{\end{eqnarray}}
\newcommand{\bes}{\begin{eqnarray*}}
\newcommand{\ees}{\end{eqnarray*}}
\newcommand{\bce}{\begin{center}}
\newcommand{\ece}{\end{center}}
\newcommand{\beae}{\begin{IEEEeqnarray}{rCl}}
\newcommand{\eeae}{\end{IEEEeqnarray}}

\def\VR{\kern-\arraycolsep\strut\vrule &\kern-\arraycolsep}
\def\vr{\kern-\arraycolsep & \kern-\arraycolsep}

\newcommand{\ben}{\begin{enumerate}}
\newcommand{\een}{\end{enumerate}}

\newcommand{\hso}{\hspace{.1in}}
\newcommand{\hst}{\hspace{.2in}}

\newtheorem{theorem}{Theorem}[section]

\newtheorem{remark}{Remark}[section]
\newtheorem{corollary}{Corollary}[section]
\newtheorem{notation}{Notation}[section]

\newtheorem{definition}{Definition}[section]
\newtheorem{lemma}{Lemma}[section]
\newtheorem{example}{Example}[section]

\newtheorem{proposition}{Proposition}[section]
\newtheorem{conclusion}{Conclusion}[section]

\begin{document}

\baselineskip=16pt


\title{New Formulas of  Feedback Capacity for AGN  Channels with Memory: A Time-Domain Sufficient Statistic Approach}






 \author{
   \IEEEauthorblockN{
     Charalambos D. Charalambous\IEEEauthorrefmark{1} 
     and 
     Christos Kourtellaris\IEEEauthorrefmark{1},
      Stelios Louka\IEEEauthorrefmark{1} 
     \\}
   \IEEEauthorblockA{
     \IEEEauthorrefmark{1}Department of Electrical and Computer Engineering\\University of Cyprus 
     \\
     Email: chadcha@ucy.ac.cy,  kourtellaris.christos@ucy.ac.cy, louka.stelios@ucy.ac.cy}}

\maketitle

\begin{abstract}
In the recent paper  \cite{charalambous-kourtellaris-loykaIEEEITC2019} it is shown, via an application example, that the  Cover and Pombra \cite{cover-pombra1989} ``characterization of the $n-$block or transmission'' feedback capacity  formula,  of additive Gaussian noise (AGN) channels,  is  the subject of much confusion in the literature, with redundant incorrect results.    The main objective of this  paper is to derive new results on the Cover and Pombra  characterization of the $n-$block feedback capacity  formula, which clarify   the main points of confusion and remove any further ambiguity on  result currently available in literature. The first part of  the paper applies  time-domain methods,  to derive for a first time,   equivalent sequential characterizations of the Cover and Pombra  characterization of  feedback capacity of   AGN channels driven by nonstationary and nonergodic Gaussian noise. The optimal channel input processes of the  new equivalent sequential characterizations are  expressed as   functionals of  a {\it sufficient statistic and a Gaussian orthogonal innovations process}. From the new representations  follows that the Cover and Pombra   $n-$block  capacity formula is  expressed as a functional of  two generalized matrix difference Riccati equations (DRE) of filtering theory of Gaussian systems, contrary to results that appeared in the literature. In the second part of the paper   the existence of the asymptotic limit of the $n-$block  feedback capacity formula  is shown to be equivalent to the  convergence properties of  solutions of the two generalized DREs. Further,  necessary and or sufficient conditions are identified for existence of the  asymptotic limits,  for stable and unstable Gaussian noise, when the optimal input distributions are time-invariant, but not necessarily stationary.
 
The paper contains an  in depth analysis, with examples,   of the specific  technical issues, which are overlooked in past literature  
 \cite{kim2010,liu-han2019,gattami2019,ihara2019,li-elia2019}, that  studied    the AGN channel of  \cite{cover-pombra1989}, for stationary noises.

\end{abstract}

%

\section{Introduction, Motivation, Main Results,   Current State  of Knowledge}
\label{sect:problem}

\subsection{The Problem, Motivation, and Main Results}
\label{sect:motivation}
We consider the additive Gaussian noise (AGN) channel defined  by \cite{cover-pombra1989}
\begin{align}
Y_t=X_t+V_t,  \hst t=1, \ldots, n, \hst \frac{1}{n}  {\bf E} \Big\{\sum_{t=1}^{n} (X_t)^2\Big\} \leq \kappa, \hso \kappa \in [0,\infty) \label{g_cp_1} 
\end{align}
where\\
$X^n \tri \{X_1, X_2, \ldots, X_n\}$ is  the sequence of channel input random variables (RVs) $X_t :  \Omega \rar {\mathbb R}$,   \\
$Y^n \tri \{Y_1, Y_2, \ldots, Y_n \}$ is  the sequence of channel output RVs $Y_t :  \Omega \rar {\mathbb R}$,\\
   $V^n\tri \{ V_1, \ldots, V_n\}$ is the  sequence of  jointly Gaussian distributed  RVs $V_t :  \Omega \rar {\mathbb R}$, with  distribution ${\bf P}_{V^n}(dv^n)$, not necessarily stationary or ergodic.  \\
 We wish to  examine the feedback capacity of the AGN channel (\ref{g_cp_1}) for  two  distinct formulations of code definition and noise model, described  below  under Case I) and Case II) formulations, which are subject to confusion in the literature.

{\bf Case I) Formulation.} This formulation respects the following two conditions.  \\
I.1). The  feedback code  does not assume  knowledge of   the initial state of the noise at   the encoder and the decoder (see Definition~\ref{def_code}),  and \\
I.2)  the noise sequence $V^n$ is represented by a partially observable\footnote{Partially observable means that knowledge of $V^{t-1}$ and initial state do not specify the state $S^t, t=1, \ldots, n$.} state space realization, with state  sequence $S^n$ (see Definition~\ref{def_nr_2}). 

For a formulation that respects I.1) and I.2), Cover and Pombra   characterized the  ``$n-$finite   transmission'' feedback capacity   \cite[eqn(10),  eqn(11)]{cover-pombra1989}, using  the information measure\footnote{$C_n^{fb}(\kappa)$ is identified using the converse coding theorem \cite{cover-pombra1989}.},    
\begin{align}
{C}_{n}^{fb}(\kappa) \tri  \sup_{{\bf P}_{X_t|X^{t-1},Y^{t-1}}, t=1, \ldots, n: \: \frac{1}{n}  {\bf E} \big\{\sum_{t=1}^{n} \big(X_t\big)^2\big\} \leq \kappa    }\sum_{t=1}^n H(Y_t|Y^{t-1})-H(V^n) \label{ftfic_is_g_def_in}
\end{align}
provided the supremum exists,   where $H(\cdot)$ denotes differential entropy.\\
Although, not mentioned in \cite{cover-pombra1989},  if the    feedback code assumes knowledge of  the  initial state of the noise or the channel, $S_1=s$,  at the encoder and the decoder (see Definition~\ref{def_rem:cp_1}), it follows directly from \cite[eqn(10), eqn(11)]{cover-pombra1989}, that (\ref{ftfic_is_g_def_in}) is replaced by the information measure    
\begin{align}
{C}_{n}^{fb}(\kappa,s) \tri  \sup_{{\bf P}_{X_t|X^{t-1},Y^{t-1},S}, t=1, \ldots, n: \: \frac{1}{n}  {\bf E} \big\{\sum_{t=1}^{n} \big(X_t\big)^2\Big|S_1=s\big\} \leq \kappa    }\sum_{t=1}^n H(Y_t|Y^{t-1},s)-H(V^n|s). \label{ftfic_is_g_def_in_IS}
\end{align}
Thus, a formulation that respects I.1) and I.2) is the most general.

{\bf  Case II) Formulation.} This formulation  relaxes  conditions I.1) and I.2) to the following two conditions. \\
II.1) The  feedback code  assumes knowledge of  the  initial state of the noise or the channel, $S_1=s$,  at the encoder and the decoder (see Definition~\ref{def_rem:cp_1}), and \\
II.2) the noise sequence $V^n$ is represented  by a fully observable state space realization\footnote{Fully observable means knowledge of $V^{t-1}$ and initial state specify the state $S^t, t=1, \ldots, n$.},  with state sequence $S^n$ such that the noise $V^{t-1}$  (including  the initial state) uniquely defines the noise state sequence $S^t$ and vice-versa for $t=1, \ldots, n$.  

For a formulation that respects II.1) and II.2),  Yang, Kavcic, and Tatikonda  \cite{yang-kavcic-tatikonda2007}, characterized the $n-$finite  transmission feedback capacity  \cite[Section~II, in particular Section~II.C,  I)-III)]{yang-kavcic-tatikonda2007}),  using   the information measure, 
\begin{align}
{C}_{n}^{fb,S}(\kappa,s)\tri  \sup_{{\bf P}_{X_t|S^{t},Y^{t-1},S}, t=1, \ldots, n: \: \frac{1}{n}  {\bf E} \big\{\sum_{t=1}^{n} \big(X_t\big)^2\Big|S_1=s\big\} \leq \kappa    }\sum_{t=1}^n H(Y_t|Y^{t-1},s)-H(V^n|s). \label{ftfic_is_g_is_def_in}
\end{align}

Compared to $C_n^{fb}(\kappa)$ and $C_n^{fb}(\kappa,s)$,  the definition of ${C}_{n}^{fb,S}(\kappa,s)$ is fundamentally different, because the input distributions are different and the information rates are different.
 

{\bf Motivation and Fundamental Differences of Case I) and Case II) Formulations.}\\
At this point we pause  to  disuss two technical issues of Cases I) and II) formulations, which are not clarified in \cite{kim2010,liu-han2019,gattami2019,ihara2019,li-elia2019} and lead to {\it fundamental confusions and incorrect interpretation of the results therein}. \\
Consider the   autoregressive moving average stable noise denoted by  ARMA$(a,c), a \in [-1,1], c\in (-1, 1), c\neq a$, studied by many authors, \cite{yang-kavcic-tatikonda2007,kim2010,liu-han2019,gattami2019,ihara2019,li-elia2019},  $V_t = c V_{t-1} + W_t -a W_{t-1}, \forall t \in {\mathbb Z}_+\tri \{1, 2, \ldots\},$
$V_0\in  N(0, K_{V_0}),  K_{V_0}\geq 0, W_0 \in N(0,  K_{W_0}),  K_{W_0}\geq 0,  W_t \in N(0,  K_{W}),    K_{W}>0, \forall t \in {\mathbb Z}_+$, $\{W_0, W_1, \ldots, W_n\}$ indep. seq. and indep. of $V_0$. 
Define the state variable of the noise by 
\begin{align}
S_t \tri \frac{c V_{t-1} -a W_{t-1}}{c-a}, \hso \forall t \in {\mathbb Z}_+    \label{ex_1_4_in}
\end{align}
Then the state space realization of $V^n$ is 
\begin{align}
&S_{t+1}=c S_t +W_t, \hst V_t= \Big(c-a\Big) S_t + W_t, \hso \forall t \in {\mathbb Z}_+, \label{ex_1_6_in}\\
&K_{S_{1}}= \frac{\big(c\big)^2 K_{V_0} +\big(a\big)^2 K_{W_0}}{\Big(c-a\Big)^2}, \hso K_{V_0}\geq 0,\hso K_{W_0}\geq 0 \hst \mbox{both given}.\label{ex_1_7_n}
\end{align}
First, to make  the transition from  Case I) formulation  of the information measure $C_n^{fb}(\kappa)$, i.e., (\ref{ftfic_is_g_def_in}),  with corresponding  channel input distributions ${\bf P}_{X_t|X^{t-1},Y^{t-1}}, t=1, \ldots,n$,  to Case II) formulation of the information measure $C_n^{fb,S}(\kappa,s)$, i.e., (\ref{ftfic_is_g_is_def_in}), with corresponding channel input distributions  ${\bf P}_{X_t|S^{t},Y^{t-1},S_1}, t=1, \ldots,n$,   the conditions stated in (\ref{inp_2}), (\ref{inp_3}) are necessary (follows  from the converse coding theorem).   
\begin{align}
{\bf P}_{X_t|X^{t-1},Y^{t-1}}=&{\bf P}_{X_t|V^{t-1},Y^{t-1}} \hst \mbox{always holds by channel definition $Y_k=X_k+V_k, k=1, \ldots, n$} \label{inp_1}
\\
=&{\bf P}_{X_t|V^{t-1},Y^{t-1},S_1} \hst \mbox{if  the initial state $S_1=s$ is known to the feedback code} \label{inp_2}\\
=&{\bf P}_{X_t|S^t,Y^{t-1},S_1} \hst \mbox{if $(V^{t-1}, S_1=s)$ uniquely defines $S^t$ and vice-versa}. \label{inp_3}
\end{align} 
Thus, a necessary condition for  (\ref{inp_3}) to hold is: $S_1=(V_0,W_0)=(v_0,w_0)=s$ is known to the encoder. It follows that in \cite[Theorem~6.1, see also Lemma~6.1 and comments above it, eqn(71)]{kim2010},  Conditions II.1), II.2) are assumed; hence, these results are not developed for the Cover and Pombra \cite{cover-pombra1989} formulation.     Additional elaboration is found in Remark~\ref{rem_kim}.\\
Second,  the  analysis of  the asymptotic per unit time  limits of (\ref{ftfic_is_g_def_in})-(\ref{ftfic_is_g_is_def_in}), and their variants (when the supremum over distributions  and limit over $n \longrightarrow \infty$ are interchanged), require certain technical necessary and/or sufficient conditions for the  limits to exist, and for the  rates to be independent of the initial data, $S_1=s$ (see \cite{charalambous-kourtellaris-loykaIEEEITC2019}), even if the noise process $V^n$ is stationary.  This is shown in  \cite{charalambous-kourtellaris-loykaIEEEITC2019}.

 
To conclude, it follows from   \cite{charalambous-kourtellaris-loykaIEEEITC2019},   that the analysis in \cite{kim2010,liu-han2019,gattami2019,ihara2019,li-elia2019}, at many parts is confusing and often contain incorrect statements; some of these are attributed to the fact that the analysis does not correspond to  the Cover and Pombra code definition and noise model, because it presupposes Conditions II.1) and II.2) hold, while others are related to the existence of asymptotic limits.



{\bf Main Results.} The  main results of this paper  are briefly stated below.\\ 
1) In the first part of the paper we derive new equivalent  sequential characterizations  of  the Cover and Pombra ``$n-$block or transmission'' feedback capacity formula \cite[eqn(11)]{cover-pombra1989}, ${C}_{n}^{fb}(\kappa)$,   which have not appeared elsewhere in the literature.  In particular, we derive  equivalent  realizations to   the optimal channel input process $X^n$ \cite[eqn(11)]{cover-pombra1989}, which are     linear  functionals of a {\it finite-dimensional sufficient statistic and an orthogonal innovations process}. From these new realizations, follows  the  sequential characterizations of the ``$n-$block or transmission'' feedback capacity formula \cite[eqn(11)]{cover-pombra1989},  henceforth  called the ``$n-$finite transmission feedback information ($n-$FTFI) capacity'',  which are expressed as  functionals of   {\it two generalized matrix difference Riccati equations (DRE) of filtering theory of Gaussian systems}.

2) In the second part of the paper we analyze the  asymptotic per unit time limit of  the sequential characterizations of the   $n-$FTFI capacity, denoted by ${C}^{fb,o}(\kappa)$, when the supremum and limit over $n\longrightarrow \infty$ are interchanged. We identify necessary and/or sufficient conditions for the asymptotic limit to exist, and for   the optimal input process $X_t, t=1, \ldots, $ to be  asymptotically stationary, in terms  of the convergence properties of two  generalized matrix difference Riccati equations (DREs) to their corresponding two generalized matrix algebraic Riccati equations (AREs). Use in made of  the so-called detectability and stabilizability conditions of generalized Kalman-filters of Gaussian processes \cite{caines1988,kailath-sayed-hassibi}.

3) From  1) and 2) we derive  analogous results for   ${C}_{n}^{fb}(\kappa,s)$
and its per unit time asymptotic limit denoted by, ${C}^{fb,o}(\kappa,s)$,  as degenerate versions of ${C}_{n}^{fb}(\kappa)$
and ${C}^{fb,o}(\kappa)$. Further, we show that for certain noise models, and under certain conditions,  it holds that ${C}^{fb,o}(\kappa,s)={C}^{fb}(\kappa)$, i.e., these values do not depend on the initial state or initial distributions. 

4) From  1) and 2) we derive  analogous results for   Case II) formulation, i.e., ${C}_{n}^{fb,S}(\kappa,s)$
and its per unit time asymptotic limit denoted by, ${C}^{fb,S,o}(\kappa,s)$, and we  show these are    fundamentally  different from  Case I) formulation, ${C}_{n}^{fb}(\kappa)$ and ${C}^{fb,o}(\kappa)$, and also  ${C}_{n}^{fb}(\kappa,s)$ and ${C}^{fb,o}(\kappa,s)$. In particular, we show that  the characterizations of $n-$FTFI capacity, ${C}_{n}^{fb,S}(\kappa,s)$,  for   Case II) formulation  follow directly from   Case I) formulation, as a special case (an independent derivation is also presented). Moreover, ${C}_{n}^{fb,S}(\kappa,s)$ is a functional of  one generalized DRE, while  ${C}_{n}^{fb}(\kappa), {C}_{n}^{fb}(\kappa,s)$, are functionals of two generalized DREs.

5) We provide proofs that our results listed under 1)-4) are fundamentally different from the current believe of  researchers, such as,  \cite{kim2010,liu-han2017,liu-han2019,gattami2019,ihara2019,li-elia2019}.


\subsection{The Code Definitions and Noise Models}
{\bf Case I) Feedback Code and Noise Definitions. }
 For Case I) formulation we consider the   code of Definition~\ref{def_code} (due to \cite{cover-pombra1989}).   \\

\begin{definition} Time-varying feedback code  \cite{cover-pombra1989}\\
\label{def_code}
 A noiseless time-varying feedback code
 for the  AGN Channel (\ref{g_cp_1}), is denoted  by  $(2^{n R},n)$,   $n=1,2, \ldots$, and  consists of the following  elements and assumptions. \\
(i) The  uniformly distributed messages $W : \Omega \rar  {\cal M}_n \tri  \left\{1, 2,\ldots, 2^{nR}\right\}$. \\
(ii) The time-varying encoder strategies, often called codewords  of block length $n$, defined by\footnote{The superscript $e(\cdot) $ on ${\bf E}^e$ indicates  that the distribution depends on the strategy $e(\cdot)\in {\cal E}_{[0,n]}(\kappa)$.} 
\begin{align}
{\cal E}_{[0,n]}(\kappa) \triangleq  \Big\{ X_1=e_1(W),X_2=e_2(W,X_1,Y_1) \ldots, X_n=e_n(W,X^{n-1}, Y^{n-1}): \; \frac{1}{n}   {\bf E}^e    \Big( \sum_{t=1}^n  (X_t)^2\Big) \leq \kappa \Big\}.  
\end{align}
(iii) The average  error  probability of  the decoder  functions  $y^n \longmapsto d_{n}(y^n)\in  {\cal M}_n$, defined by
\begin{align}
{\bf P}_{error}^{(n)} = {\mathbb P}\Big\{d_{n}(Y^n) \neq W\Big\}= \frac{1}{2^{nR}} \sum_{W=1}^{2^{nR}} {\mathbb  P}\Big\{d_n(Y^n) \neq W\Big\}. \label{g_cp_4}
\end{align}
(iv) The channel input sequence ``$X^n \tri \{X_1, \ldots,X_n\}$ is causally related\footnote{A notion found in    \cite{cover-pombra1989}, page 39, above Lemma~5.} to $V^n$'', which is equivalent to the following decomposition of the joint probability distribution of $(X^n, V^n)$: 
\begin{align}
{\bf P}_{X^n, V^n}
=&{\bf P}_{V_n|V^{n-1}, X^{n}} \; {\bf P}_{X_n|X^{n-1}, V^{n-1}} \; \ldots \; {\bf P}_{V_2|V_1, X^2}  {\bf P}_{X_2|X_1, V_1} {\bf P}_{V_1|X_1} {\bf P}_{X_1}\\
=&{\bf P}_{V^n} \prod_{t=1}^n  {\bf P}_{X_t|X^{t-1}, V^{t-1}}, \hst \mbox{that is,} \hso  {\bf P}_{V_t|V^{t-1},X^t}={\bf P}_{V_t|V^{t-1}}. \label{g_cp_3}
\end{align}
That is, $X^t \leftrightarrow V^{t-1} \leftrightarrow V_t$ is a Markov chain, for  $t=1, \ldots, n$. As usual, the messages $W$ are independent of the channel noise $V^n$.   \\
A rate $R$ is called an achievable rate with feedback coding,  if there exists a sequence of codes $(2^{nR},n), n=1,2, \ldots$, 
such that ${\bf P}_{error}^{(n)} \longrightarrow 0$ as $n \longrightarrow \infty$. The feedback capacity  $C^{fb}(\kappa)$  is defined as the supremum of all achievable rates $R$. 
\end{definition}
   
We consider a noise  model which is consistent with \cite{cover-pombra1989}, i.e., $V^n$ is jointly Gaussian distributed,    ${\bf P}_{V^n}=\times_{t=1}^n {\bf P}_{V_t|V^{t-1}}$, and induced by the  partially observable state space (PO-SS) realization of Definition~\ref{def_nr_2}.   \\

\begin{definition} 
 \label{def_nr_2}
A time-varying PO-SS realization of the Gaussian noise $V^n \in N(0, K_{V^n})$ is defined by  
\begin{align}
&S_{t+1}=A_{t} S_{t}+ B_{t} W_t, \hso t=1, \ldots, n-1\label{real_1a}\\
&V_t= C_t S_{t} +  N_t W_t,  \hso t=1, \ldots, n, \label{real_1_ab}\\
 & S_1\in N(\mu_{S_1},K_{S_1}), \hso K_{S_1} \succeq 0, \\
&W_t\in N(0,K_{W_t}),\hso K_{W_t} \succeq 0, \hso  t=1 \ldots, n \hso \mbox{an indep. Gaussian process}, \hso W^t \hso \mbox{indep. of} \hso S_1, \\
& S_t : \Omega \rar {\mathbb R}^{n_s}, \hso W_t : \Omega \rar {\mathbb R}^{n_w}, \hso V_t : \Omega \rar {\mathbb R}^{n_v},  \hso R_t\tri N_t K_{W_t} N_t^T \succ 0,\hso  t=1,\ldots, n \label{cp_e_ar2_s1_a_new}
\end{align}
where  $n_v=1$,  $n_s, n_w$ are arbitrary positive integers, and   $(A_{t}, B_{t}, C_t, N_t,\mu_{S_1}, K_{S_1}, K_{W_t})$ are nonrandom for all $t$, and $n_s, n_w$ are finite positive integers. \\
A time-invariant PO-SS realization of the Gaussian noise $V^n \in N(0, K_{V^n})$ is defined by  (\ref{real_1a})-(\ref{cp_e_ar2_s1_a_new}), with $(A_{t}, B_{t}, C_t, N_t, K_{W_t})=(A, B, C, N,K_{W}), \forall t$. 
\end{definition}

For Case I) formulation we use  the terminology ``partially observable'', which  is standard  in filtering theory \cite{caines1988}, because the noise $V^n$  induces a distribution   ${\bf P}_{V^n}=\times_{t=1}^n {\bf P}_{V_t|V^{t-1}}$,  and ${\bf P}_{V_t|V^{t-1}}$ cannot be expressed as a function of the state of the noise, i.e., $V^{t-1}$ does not uniquely define  $S^{t}$.   The  PO-SS realization  is often adopted in many practical problems of engineering and science,  to realize jointly Gaussian processes $V^n$. 


We should emphasize that for  Case I) formulation  to be consistent  with the  Cover and Pombra \cite{cover-pombra1989} formulation, see for example,    the code definition in  \cite[page~37]{cover-pombra1989}, the  characterization of the $n-$finite transmission feeback  capacity   \cite[eqn(10), eqn(11)]{cover-pombra1989},  and the coding theorems \cite[Theorem~1]{cover-pombra1989}, then  
 both   the code of Definition~\ref{def_code} and the   PO-SS realization of Definition~ \ref{def_nr_2}, must  respect the following two conditions:

{\bf (A1)} {\it The  initial state $S_1$ of the noise is not known at the encoder and the decoder, and}\\
{\bf (A2)}  {\it at each $t$, the representation of the noise $V^{t-1}$ by the PO-SS realization  of Definition~\ref{def_nr_2},  does  not uniquely determine  the state of the noise $S^t$ and vice-versa, i.e., it is a partially observable realization.}

{\bf Case II) Formulation of Feedback Code and Noise Definitions.} For Case II) formulation we pressupose: 

{\bf Condition 1.} {\it  The initial state of the noise or  channel $S_1=s$ is known to the encoder and  decoder, and } \\
{\bf Condition 2.} {\it  given a fixed initial state  $S_1=s$, known to the encoder and the decoder, at each $t$, the channel noise $V^{t-1}$ uniquely defines the state of the noise $S^t$ and vice-versa.} 

Thus, for Case II) formulation  the   code  is that of Definition~\ref{def_rem:cp_1}, below (hence different from Definition~\ref{def_code}).    \\

\begin{definition} A code with initial state known  at the  encoder and the decoder\\
\label{def_rem:cp_1}
A variant of the   code of Definition~\ref{def_code}, is a  feedback code
 with the   initial state of the noise or channel $S_1= s$, known to the encoder and decoder strategies,  denoted  by  $(s, 2^{n R},n)$,   $n=1,2, \ldots$. \\
 The code $(s, 2^{n R},n)$,   $n=1,2, \ldots$  is  defined as in Definition~\ref{def_nr_2}, with (ii), (iii), (iv) replaced by  
\begin{align}
{\cal E}_{[0,n]}^{s}(\kappa) \triangleq & \Big\{ X_1=e_1(W,S_1),X_2=e_2(W,S_1,X_1,Y_1) \ldots, X_n=e_n(W,S_1,X^{n-1}, Y^{n-1}): \nonumber \\
&\frac{1}{n+1}   {\bf E}^e    \Big\{ \sum_{i=1}^n  (X_t)^2\Big|S_1=s  \Big\}\leq \kappa \Big\},  \hst y^n \longmapsto  d_{n}^{s}(y^n,v_{-\infty}^o)\in  {\cal M}_n, \label{code_ykt_1} \\
{\bf P}_{X^n, V^n|S_1}
=&{\bf P}_{V^n|S_1} \prod_{t=1}^n  {\bf P}_{X_t|X^{t-1}, V^{t-1},S_1}, \hst \mbox{that is,} \hso  {\bf P}_{V_t|V^{t-1},X^t,S_1}={\bf P}_{V_t|V^{t-1},S_1}. \label{code_ykt_3}
\end{align}
The initial state may include $S_1\tri (V_{-\infty}^0, Y_{-\infty}^0)$, etc.
\end{definition} 

For Case II) formulation it is obvious  (from the converse to the coding theorem), that the optimal  channel input  conditional distribution is  expressed as a function of the state of the noise, $S^n$, due to (\ref{inp_2}), (\ref{inp_3}).

%

\subsection{Approach  of this Paper} 
  Our approach and analysis of information measures (\ref{ftfic_is_g_def_in})-(\ref{ftfic_is_g_is_def_in}), and their per unit time limits,  is based on the following two step procedure.

{\bf Step \# 1.} We apply a linear transformation to the Cover and Pombra optimal channel input process \cite[eqn(11)]{cover-pombra1989}  (see  (\ref{cp_11})-(\ref{cp_9}) which are reproduced from \cite{cover-pombra1989} for the convenience of the reader), to equivalently represent  it by a linear functional of the past channel noise sequence, the past channel output sequence,  and an orthogonal  Gaussian   process, i.e., an      innovations process. That is, $X^n$  is uniquely represented, since it is expressed in terms of the orthogonal process.  

{\bf Step \# 2.} We express the optimal input process   by  a functional of a {\it sufficient statistic}, which satisfies a Markov recursion,  and an {\it orthogonal innovations process}. It then follows   that the Cover and Pombra characterization of the ``$n-$block'' formula \cite[eqn(10)]{cover-pombra1989} (see (\ref{cp_11}) and (\ref{cp_12})) is equivalently represented  by a sequential characterization. The problem of feedback capacity is then expressed as the maximization of the per unit time limit of a sum of (differential)  entropies  of the innovations processes of $Y^n$, and $V^n$,  over  two sequences of time-varying strategies  of the channel input process,  of the following entropies (analog of entropies  in the right hand side of (\ref{ftfic_is_g_def_in})).  
\begin{align}
H(Y^n)-H(V^n)=& \sum_{t=1}^n \Big\{ H(Y_t|Y^{t-1})-H(V_t|V^{t-1})\Big\}  \label{inn_a_intr}   \\
=& \sum_{t=1}^n\Big\{ H(Y_t- {\bf E}\big\{Y_t\Big|Y^{t-1}\big\}|Y^{t-1})-H(V_t- {\bf E}\big\{V_t\Big|V^{t-1}\big\}|V^{t-1})\Big\}\\
\sr{(a)}{=}&\sum_{t=1}^n\Big\{ H(Y_t- {\bf E}\big\{Y_t\Big|Y^{t-1}\big\})-H(V_t- {\bf E}\big\{V_t\Big|V^{t-1}\big\})\Big\} \\
=&\sum_{t=1}^n\Big\{ H(I_t)-H(\hat{I}_t)\Big\}, \hso I_t\tri Y_t- {\bf E}\big\{Y_t\Big|Y^{t-1}\big\}, \; \hat{I}_t \tri V_t- {\bf E}\big\{V_t\Big|V^{t-1}\big\} \label{inn_intr}
\end{align}
where $(a)$ is due to orthogonality of innovations process $I_t$ and $Y^{t-1}$ and of innovations process  $\hat{I}_t$ and $V^{t-1}$. \\The asymptotic analysis of $\lim_{n \longrightarrow \infty} \frac{1}{n}{C}_{n}^{fb}(\kappa)$ (or with limit and supremum interchanged)  is then addressed from the asymptotic properties of the  entropy rates and the  average power, 
\begin{align}
\lim_{n \longrightarrow \infty} \frac{1}{n} \sum_{t=1}^n\Big\{ H(I_t)-H(\hat{I}_t)\Big\}, \hst \lim_{n \longrightarrow \infty} \frac{1}{n}  {\bf E} \Big\{\sum_{t=1}^{n} \big(X_t\big)^2\Big\}
\end{align}
over the channel input distributions, and where the  covariance of the  innovations process of $Y^n$  is a functional  of the solutions of two generalized matrix DREs. We identify necessary and/or sufficient conditions for  existence of the limits, irrespectively of whether the noise $V^n$ is nonstationary, unstable, or stationary. Further, we show the characterizations of feedback capacity for Case I) formulation and Case II) are fundamentally different.


  \subsection{The Cover and Pombra Characterizations of Capacity and Related Literature}
\label{sect:cp}

Cover and Pombra  applied the converse coding theorem and the  maximum entropy principle of Gaussian distributions  to  characterize  the $n-$FTFI capacity  \cite[eqn(10)]{cover-pombra1989} by\footnote{We use $H(X)$ to denote differential entropy of a continuous-valued RV $X$, hence we indirectly assume  the  probability density functions exist.} 
\begin{align}
 C_n^{fb}(\kappa) 
  \tri & \max_{  \big({\bf B}^n, K_{{\bf \overline{Z}}^n}\big): \hso \frac{1}{n} tr  \Big\{ {\bf E} \big({\bf X}^n ({\bf X}^n)^T\big)\Big\}\leq \kappa}  H(Y^n) - H(V^n )  \label{cp_11}\\
=&  \max_{  \big({\bf B}^n, K_{{\bf \overline{Z}}^n}\big):\hso \frac{1}{n} tr\Big\{{\bf B}^n \; K_{\bf {V^n}} \; ({\bf B}^n)^T+ K_{\bf {\overline{Z}}^n}\Big\}  \leq \kappa } \frac{1}{2} \log \frac{|\big({\bf B}^n+I_{n\times n}\big) K_{{\bf V}^n}\big({\bf B}^{n}+I_{n\times n}\big)^T+ K_{{\bf \overline{Z}}^n}|}{|K_{{\bf V}^n}|} \label{cp_12}
  \end{align}
where the distribution ${\bf P}_{Y^n}$ is induced by a    jointly Gaussian channel input process $X^n$ \cite[eqn(11)]{cover-pombra1989}:
\begin{align}
&X_t=\sum_{j=1}^{t-1}{B}_{t,j}V_j +\overline{Z}_t, \hso t=1, \ldots, n, \label{cp_6}   \\
&  
{\bf X}^n ={\bf B}^{n} {\bf V}^n + {\bf \overline{Z}}^n, \hso {\bf Y}^n=\Big({\bf B}^n +I_{n\times n}\Big) {\bf V}^n + {\bf \overline{Z}}^n, \label{cp_10}\\
& {\bf \overline{Z}}^n \hso \mbox{is jointly Gaussian, $N(0, K_{\bf {\overline{Z}}^n})$}, \hso {\bf \overline{Z}}^n \hso \mbox{is independent of} \hso {\bf V}^n, \label{cp_8} \\
&{\bf X}^n\tri \left[ \begin{array}{cccc} X_1 &X_2 &\ldots &X_n\end{array}\right]^T \hso \mbox{and similarly for the rest}, \hso \mbox{${\bf B}^n$ is a lower diagonal  matrix}, \\
&\frac{1}{n}  {\bf E} \Big\{\sum_{t=1}^{n} (X_t)^2\Big\} =\frac{1}{n}tr \; \Big\{ {\bf E}\Big({\bf X}^n ({\bf X}^{n})^T\Big) \Big\} \leq \kappa. \label{cp_9} 
\end{align} 
The notation $N(0, K_{\bf {\overline{Z}}^n})$ means the random variable ${\bf \overline{Z}}^n$ is jointly Gaussian with mean ${\bf E}\{{\bf \overline{Z}}^n\}=0$ and covariance matrix $K_{{\bf \overline{Z}}^n}={\bf E}\{ {\bf \overline{Z}}^n ({\bf \overline{Z}}^n)^T  \}$, and $I_{n\times n}$ denotes an $n$ by $n$ identity matrix.\\
  The  feedback capacity, $C^{fb}(\kappa)$, is characterized    by  the per unit time limit of the $n-$FTFI capacity  \cite{cover-pombra1989}.
 \begin{align}
 C^{fb}(\kappa) \tri \lim_{n \longrightarrow \infty} \frac{1}{n} C_n^{fb}(\kappa). \label{CP_F}
 \end{align}
 The  direct and converse coding theorems, are   stated in  \cite[Theorem~1]{cover-pombra1989}. \\ 
Over the years,  considerable efforts have been devoted  to compute $C_n^{fb}(\kappa)$ and $C^{fb}(\kappa)$, \cite{yang-kavcic-tatikonda2007,kim2010,liu-han2017,liu-han2019,gattami2019,ihara2019}, often  under simplified assumptions on  the channel noise. In addition,    bounds are described in  \cite{chen-yanaki1999,chen-yanaki2000}, while numerical methods are developed in \cite{ordentlich1994}, mostly for time-invariant AGN channel, driven by stationary noise. We should mention that most papers considered a variant of (\ref{CP_F}), by interchanging the per unit time limit and the maximization operations, 
under the assumption: the joint process $(X^n, Y^n), n=1, 2, \ldots$ is either {\it jointly stationary or asymptotically stationary} (see \cite{kim2010,liu-han2017,liu-han2019,gattami2019}), and  {\it the  joint  distribution of the joint process  $(X^n, Y^n), n=1, 2, \ldots$ is time-invariant}.

Yang, Kavcic and Tatikonda \cite{yang-kavcic-tatikonda2007} and Kim \cite{kim2010} analyzed the feedback capacity of the AGN channel (\ref{g_cp_1}) driven by a stationary  noise, described  the power spectral density   (PSD) functions $S_V(e^{j\theta}), \theta \in [-\pi,\pi]$:
\begin{align}
S_V(e^{j\theta}) \tri& K_W\frac{\Big( 1-\sum_{k=1}^L a(k) e^{jk \theta}\Big)\Big(1- \sum_{k=1}^L a(k) e^{-j k \theta}\Big)}{\Big( 1-\sum_{k=1}^L c(k) e^{jk \theta}\Big)\Big(1- \sum_{k=1}^L c(k) e^{-jk\theta}\Big)}, 
\; |c(k)|<1,\; |a(k)|< 1, \; c(k) \neq a(k), \forall k. \label{PSD_G}
\end{align}
More specifically, the analysis by Yang, Kavcic and Tatikonda,  presupposed Case II) formulation (see \cite[Section~II, in particular Section~II.C,  I)-III), Theorem~1, Section~III]{yang-kavcic-tatikonda2007}), and   a specific state space realization of the noise PSD (\ref{PSD_G}), such  that the following holds:

 {\it The initial state of  the noise, $S_1=s$,  is known to the encoder and the decoder, and the initial state and noise $(s,V^{t-1})$   uniquely define the noise state $S^t$, and vice versa, for all $t$.}

 Kim also analyzed the feedback capacity of the AGN channel (\ref{g_cp_1}) driven by a stationary  noise described by the PSD (\ref{PSD_G}), and by a state space realization of the noise $V^n$ (see \cite[Section~VI]{kim2010}). A major  point of confusion, which should be read with caution is that,     Kim's characterization of feedback capacity in time-domain \cite[Theorem~6.1]{kim2010}, does not state the conditions based on which this characterization is derived. The reader, however,  can verify  from \cite[Lemma~6.1 and comments above it]{kim2010}, that the characterization of feedback capacity  \cite[Theorem~6.1]{kim2010}, presupposed a Case II) formulation, precisely as Yang, Kavcic and Tatikonda \cite{yang-kavcic-tatikonda2007}.
We  reconfirm this point at various parts of  this paper (see for example,  Section~\ref{sect:p_l}).  

A recent investigation of AGN channels driven by autoregressive unit memory stable and unstable noise with feedback is \cite{kourtellaris-charalambous-loyka:2020b}.  An investigation of 
nonfeedback capacity of stable and unstable noise   is \cite{kourtellaris-charalambous-loyka:2020a}.   
The connection of ergodic theory and feedback capacity of unstable channels is discussed in 
\cite{kourtellaris-charalambousIT2015_Part_1,charalambous-kourtellaris-loykaIT2015_Part_2}. 

We structure the paper as follows. 
In Section~\ref{sect:AGN} we derive the new sequential  characterizations of the $n-$FTFI  capacity for the   Cover and Pombra formulation of feedback capacity of the  AGN channel (\ref{g_cp_1}), i.e. for Case I) formulation, $C_n^{fb}(\kappa)$. We  also derive analogous sequential  characterizations for $C_n^{fb}(\kappa,s)$, and for  Case II) formulation, $C_n^{fb,S}(\kappa,s)$,  i.e., when   Conditions 1 and 2 hold to illustrate  their fundamental differences. 
%
%
In Section~\ref{sect:as_an}  we present  the asymptotic analysis of feedback capacity for Case I) formulation. 
In Section~\ref{sect_POSS} we treat the Case II) formulation. 
The paper contains several  examples, and comparisons to   existing literature.



 \section{Sequential Characterizations of $n-$FTFI Capacity for Case I) Formulation}
  \label{sect:AGN}
In this section we derive equivalent sequential  characterizations, for 

i) $C_n^{fb}(\kappa)$ defined by  (\ref{ftfic_is_g_def_in}) of Case I) formulation, i.e., for the Cover and Pombra $n-$FTFI capacity characterization  (\ref{cp_12}),

ii)   ${C}_{n}^{fb}(\kappa,s)$ defined by  (\ref{ftfic_is_g_def_in_IS}), as a degenerated case of $C_n^{fb}(\kappa)$, and     

iii) ${C}_{n}^{fb,S}(\kappa,s)$ defined by (\ref{ftfic_is_g_is_def_in}) of Case II) formulation,  as a degenerated case of $C_n^{fb}(\kappa)$.

Their  asymptotic per unit time limit are addressed in  Section~\ref{sect:as_an}.  

We organize the presentation of the material  as follows: \\
1) Section~\ref{sect:notation}. Here we  introduce our notation.

2) Section~\ref{sect:preliminary}. The main result is 
Theorem~\ref{thm_FTFI}, which gives  an equivalent   sequential characterization of the Cover and Pombra characterization of the $n-$FTFI capacity, $C_n^{fb}(\kappa)$, i.e., of  (\ref{cp_11}), (\ref{cp_12}). Our derivation proceeds as follows. We apply    a linear transformation to the Cover and Pombra Gaussian optimal channel input $X^n$ (\ref{cp_6}), to represent $X_t$, by  a linear function of $(V^{t-1}, Y^{t-1})$ or equivalently $(X^{t-1}, Y^{t-1})$ and an  orthogonal  Gaussian  innovations process $Z_t$, which is   independent of $(Z^{t-1},X^{t-1},V^{t-1}, Y^{t-1})$ for  $t=1, \ldots, n$. \\
Subsequently, we apply Theorem~\ref{thm_FTFI} to  the time-varying PO-SS$(a_t,c_t, b_t^1, b_t^2, d_t^1, d_t^2)$ noise (see  Example~\ref{ex_1_poss}),  to the nonstationary autoregressive moving average, ARMA$(a,c), a \in (-\infty, \infty), c\in (-\infty, \infty)$ noise, and to the stationary ARMA$(a,c), a \in (-1,1), c\in (-1,1)$ noise (see  
Example~\ref{ex_1_1_n}), which is found in many references, such as, \cite{kim2010}. It will become apparent  that our characterizations of $n-$FTFI capacity are fundamentally different from past literature.

%
%
%
%
%


3) Section~\ref{sect_POSS_SS}. The main result is  Theorem~\ref{thm_SS}, which gives a simplified characterization of  the sequential characterization of   the $n-$FTFI capacity, $C_n^{fb}(\kappa)$, given in  Theorem~\ref{thm_FTFI} (i.e., the equivalent of (\ref{cp_12})), 
 for time-varying  AGN channel (\ref{g_cp_1}) driven by the  PO-SS realization of Definition~\ref{def_nr_2}, for  the code of Definition~\ref{def_code}. The $n-$FTFI capacity of Theorem~\ref{thm_SS} is expressed in terms of solutions to two DREs.
Our derivation   is based on identifying a {\it finite-dimensional  sufficient statistic}  to express $X_t$ as a functional of the sufficient statistic,  instead of  $(V^{t-1}, Y^{t-1})$ or  $(X^{t-1}, Y^{t-1})$, and an orthogonal Gaussian innovations process. 

4) Section~\ref{appl_ex}. The main results is  Corollary~\ref{cor_ex_1},  which is  an application  of Theorem~\ref{thm_SS} (i.e., the sufficient statistic representation),  to   the  ARMA$(a,c), a \in (-\infty,\infty), c\in (-\infty,\infty)$ noise of Example~\ref{ex_1_1_n}. This example  shows that the $n-$FTFI capacity is expressed in terms of solutions to two DREs.  From Corollary~\ref{cor_ex_1}, the following will become apparent: \\
(i) Neither the  time-domain characterization  \cite[Theorem~6.1]{kim2010} (see \cite[Theorem~5.3]{kim2010}) nor   the frequency domain characterization \cite[Theorem~4.1]{kim2010},  correspond to the Cover and Pombra characterization of feedback capacity (when the limit and maximization operations are interchanged) of the nonstationary and  stationary ARMA$(a,c)$ noise of  Example~\ref{ex_1_1_n}.   \\
 (ii) The characterizations given in  \cite[Theorem~6.1 and Theorem~4.1]{kim2010} is incorrect, reconfirming the recent analysis in \cite{charalambous-kourtellaris-loykaIEEEITC2019}, of the autoregressive unit memory noise.


5) Section~\ref{sect:case_II}. The main results is Corollary~\ref{cor_issr}, which gives the  $n-$FTFI Capacity for Case II) formulation, as a degenerate case of Case I) formulation, i.e.,  of Theorem~\ref{thm_SS}.

6) Section~\ref{sect:p_l}.  The main result is Proposition~\ref{pro_1}, which further clarifies the following: (i)   the formulation of \cite{yang-kavcic-tatikonda2007} and the formulation that let to \cite[Theorem~6.1]{kim2010}, are based on Case II) formulation, and (ii) some of the  oversights in  \cite{kim2010,liu-han2019,gattami2019,ihara2019,li-elia2019}. 



\subsection{Notation}
\label{sect:notation}
 Throughout the paper, we use the following notation.\\
${\mathbb Z}\tri \{\ldots, -1,0,1,\ldots\}, {\mathbb Z}_+ \tri \{1,\ldots\},    {\mathbb Z}_+^n \tri \{1,2, \ldots, n\}$, where $n$ is a finite positive integer. \\
${\mathbb R}\tri (-\infty, \infty)$,  and ${\mathbb R}^m$ is the vector space of tuples
of the real numbers  for an integer  $n\in {\mathbb Z}_+$.\\
${\mathbb C} \tri \{a+j b: (a,b) \in {\mathbb R} \times {\mathbb R}\}$ is the  space of complex numbers. \\
${\mathbb R}^{n \times m}$ is the set of $n$ by $m$ matrices with entries from the set of real numbers for integers   $(n,m)\in {\mathbb Z}_+\times {\mathbb Z}_+$. \\
${\mathbb D}_o \tri \big\{c \in {\mathbb C}: |c| <1\big\}$ is the open unit disc of the space of complex number ${\mathbb C}$. \\
${\mathbb S}_{+}^{n\times n}, n\in {\mathbb Z}_+$ (resp. ${\mathbb S}_{++}^{n\times n}$)
denotes the set of positive semidefinite (resp. positive definite) symmetric matrices with
elements in the real numbers and of size $n\times n$. Thus, $A \in {\mathbb S}_{+}^{n\times n}$
if for all $w \in {\mathbb R}^n$,  $w^T A w \geq 0 $. 
Positive semidefiniteness is denoted by $A \succeq  0$ and (strict) positive definiteness
by $A \succ  0$. $I_{n\times n} \in {\mathbb S}_{++}^{n\times n}, n\in {\mathbb Z}_{+}$ denones the identity matrix, $tr\big(A\big)$ denotes the trace of any matrix $A \in {\mathbb R}^{n\times n}, n\in {\mathbb Z}_+$.  \\
$spec(A) \subset {\mathbb C}$ is  the Spectrum of a matrix $A \in {\mathbb R}^{q \times q}, q \in {\mathbb Z}_+$ (the set of all its eigenvalues).
A  matrix $A \in {\mathbb R}^{q \times q}$ is called exponentially stable if all its eigenvalues are within the open unit disc, that is,  $spec(A) \subset {\mathbb D}_o$.\\
$\Big(\Omega, {\cal F}, {\mathbb P}\Big)$ denotes a probability space. 
Given a  random variable $X: \Omega \rar {\mathbb R}^{n_x}, n_x\in {\mathbb Z}_+^n$, its  induced    distribution on ${\mathbb R}^{n_x}$ is denoted by  ${\bf P}_{X}$. \\
${\bf P}_{X}\in N(\mu_{X}, K_{X}), K_{X}\succeq 0$ denotes a Gaussian distributed RV $X$, with   mean value  $\mu_{X}$ and covariance matrix $K_{X}=cov(X,X)\succeq 0$, defined by  
\begin{align}
\mu_{X} \tri {\bf E}\{X\}, \hst K_X = cov(X,X) \tri {\bf E}\Big\{\Big(X-{\bf E}\Big\{X\Big\}\Big) \Big(X-{\bf E}\Big\{X\Big\}\Big)^T \Big\}.
\end{align}  
Given another  Gaussian random variables $Y: \Omega \rar {\mathbb R}^{n_y},  n_y\in {\mathbb Z}_+^n$, which is jointly Gaussian distributed with $X$, i.e., the  joint distribution is ${\bf P}_{X,Y}$,   the  conditional covariance of $X$ given $Y$ is  defined by
\begin{align}
K_{X|Y} = cov(X,X\Big|Y) \tri & {\bf E}\Big\{\Big(X-{\bf E}\Big\{X\Big|Y\Big\}\Big) \Big(X-{\bf E}\Big\{X\Big|Y\Big\}\Big)^T\Big|Y \Big\}\\
=&{\bf E}\Big\{\Big(X-{\bf E}\Big\{X\Big|Y\Big\}\Big) \Big(X-{\bf E}\Big\{X\Big|Y\Big\}\Big)^T \Big\}
\end{align}
where the last equality is due to a property of  jointly Gaussian distributed RVs.\\
Given three arbitrary RVs $(X, Y, Z)$ with induced distribution ${\bf P}_{X, Y, Z}$, the RVs $(X, Z)$ are called conditionally independent given the RV $Y$ if  ${\bf P}_{Z|X,Y}={\bf P}_{Z|Y}$. This conditional independence is often denoted by,   $X \leftrightarrow Y \leftrightarrow Z$ is a Markov chain.


\subsection{Preliminary Characterizations of $n-$FTFI Capacity of AGN Channels Driven by Correlated Noise}
\label{sect:preliminary}
We start with preliminary calculations, for the feedback  code of Definition~\ref{def_code},  which we   use to prove Theorem~\ref{thm_FTFI}. These calculations are  introduced for the sake of clarity and to establish our notation.\\
For the feedback  code of Definition~\ref{def_code}, by the channel definition (\ref{g_cp_1}), i.e., (\ref{g_cp_3}), the conditional distribution of $Y_t$ given $Y^{t-1}=y^{t-1}, X^t=x^t$, is   
\begin{align}
{\mb P}\big\{Y_t \in dy \Big| Y^{t-1}=y^{t-1}, X^t=x^t\big\} =&{\mb P}\big\{Y_t \in dy \Big| Y^{t-1}=y^{t-1}, X^t=x^t, V^{t-1}=v^{t-1}\big\}, \hst \mbox{by (\ref{g_cp_1})} \\
=&{\bf P}_{V_t|V^{t-1}}\Big(v_t: x_t+v_t \in dy\Big),   \ \ t=2, \ldots, n,\hst  \mbox{by (\ref{g_cp_3})}   \\
=&{\bf P}_{Y_t|X_t, V^{t-1}}\\
\equiv &{\bf P}_t(dy |x_t,v^{t-1}), \label{NCM-A.D_CD_C_CD_n}\\
{\mb P}\big\{Y_1 \in dy \Big| Y^{0}=y^{0}, X^1=x^1\big\}=&{\bf P}_{Y_1|X_1}\equiv {\bf P}_1(dy |x_1).\label{NCM-A.D_CD_C_CD2_n}
\end{align}
 We introduce  the  set of channel input distributions with feedback, which are consistent with the code of Definition~\ref{def_code}, not necessarily generated by the messages $W$, as follows:  
\begin{align}
{\cal P}_{[0,n]}(\kappa) \tri \Big\{{P}_t(dx_t| x^{t-1}, y^{t-1})\tri {\bf P}_{X_t|X^{t-1},Y^{t-1}}, t=1,\ldots,n: 
\frac{1}{n} {\bf E}^{ P}\Big( \sum_{t=1}^n \big(X_t\big)^2 \Big) \leq \kappa   \Big\} . \label{conv_11} 
\end{align} 
By Definition~\ref{def_code}, we have  ${\cal E}_{[0,n]}(\kappa) \subseteq {\cal P}_{[0,n]}(\kappa)$. Moreover, by the channel definition, any  pair of  the  sequence triple $(V^t,X^t, Y^t)$ uniquely defines  the remaining sequence. Thus,  the identity holds:
\begin{align}
&\overline{\cal P}_{[0,n]}(\kappa) \tri \Big\{ \overline{P}_t(dx_t|v^{t-1}, y^{t-1}), t=1,\ldots,n: \frac{1}{n+1} {\bf E}^{ \overline{P}}\Big( \sum_{t=1}^n \big(X_t\big)^2\Big) \leq \kappa    \Big\}={\cal P}_{[0,n]}(\kappa). \label{conv_1}
\end{align}
We also  emphasize that, by Definition~\ref{def_code}, for   a given feedback encoder strategy $e(\cdot) \in {\cal E}_{_{[0,n]}}(\kappa)$, i.e., $x_1=e_1(w), x_2=e_2(w, x_1, y_1), \ldots, x_n=e_n(w, x^{n-1},y^{n-1})$  the conditional distributions of $Y_t$ given $(Y^{t-1}, W)=(y^{t-1},w)$ depend on the strategies, $e(\cdot)$ as follows: 
\begin{align}
{\bf P}_{Y_t|W,Y^{t-1}}^e(dy_t|,y^{t-1},w) \sr{(a)}{=}&{\bf P}_t(dy_t|\{e_j(w, x^{j-1}, y^{j-1}):j=1,\ldots, t\}, y^{t-1},w)  \label{kernel_1}  \\
\sr{(b)}{=}&{\bf P}_t(dy_t|\{e_j(w, x^{j-1}, y^{j-1}):j=1,\ldots, t\},y^{t-1}, v^{t-1}, w) \\
\sr{(c)}{=}&{\bf P}_t(dy_t|\{e_j(w, x^{j-1}, y^{j-1}):j=1,\ldots, t\}, v^{t-1}, w)\\
\sr{(d)}{=}&{\bf P}_t(dy_t|\{e_j(w, x^{j-1}, y^{j-1}):j=1,\ldots, t\},v^{t-1})\\
\sr{(e)}{=}&{\bf P}_t(dy_t|e_t(w, x^{t-1}, y^{t-1}), v^{t-1}) \label{kernel_2}
\end{align}
$(a)$ is due to  knowledge of the distribution of the strategies $e_j(\cdot), j=1, \ldots, t$,  the code definition, and the recursive substitution, $x_1=e_1(w), x_2=e_2(w,x_1, y_1), \ldots, e_t(w,x^{t-1},y^{t-1})$, where $x^{t-1}$ is specified by the knowledge of  the strategies, $e_j(\cdot), j=1, \ldots, t-1$ and the knowledge of $(y^{t-2}, w) $,\\
$(b)$ is due to knowing $x_j=e_j(w, x^{j-1}, y^{j-1}), y_j, j=1, \ldots, t-1$ specifies $v_j=y_j-x_j, j=1, \ldots, t-1$, \\
$(c)$ is due to the fact that, any pair of the  triple  $(x^t, y^t, v^t)$ specifies the  remaining  sequence, i.e., knowing $(x^{t-1}, v^{t-1})$ specifies $y^{t-1}$, and hence $y^{t-1}$ is redundant, \\
$(d)$ is due to the conditional independence ${\bf P}_{V_t|V^{t-1}, X^t, W}=  {\bf P}_{V_t|V^{t-1}, X^t}$, \\
$(e)$ is due to   (\ref{g_cp_3}), i.e., ${\bf P}_{V_t|V^{t-1}, X^t}={\bf P}_{V_t|V^{t-1}}$, and the channel definition. 

By the channel definition $Y_t=X_t+V_t, t=1, \ldots, n$, then  each $e(\cdot)\in {\cal E}_{[0,n]}(\kappa)$ is also expressed as
\begin{align}
 x_1=&e_1(w)=\overline{e}_1(w),\hso  x_2=e_2(w,x_1,y_1)=\tilde{e}_2(w,x_1,v_1, y_1)\sr{(a)}{=}\overline{e}_2(w,v_1, y_1), \hso   \ldots,\nonumber \\
 & x_n=e_n(w, x^{n-1}, y^{n-1})=\tilde{e}_n(w, x^{n-1},v^{n-1}, y^{n-1})\sr{(a)}{=}  \overline{e}_n(w,v^{n-1}, y^{n-1}),\hso w\in {\cal M}^{(n)}. \label{enc_1}
 \end{align}
where $(a)$ is due to the  channel definition, i.e.,    the presence of $x^{t-1}$ in $\tilde{e}_t(\cdot,v^{t-1},\cdot)$ can be removed, since it is redundant,  and  specified by $(v^{t-1}, y^{t-1})$. Consequently,  we have the identity 
\begin{align}
\overline{\cal E}_{[0,n]}(\kappa) \triangleq  \Big\{ x_1=\overline{e}_1(w),x_2=\overline{e}_2(w,v_1,y_1) \ldots, x_n=\overline{e}_n(w,v^{n-1}, y^{n-1}): \frac{1}{n}   {\bf E}^{\overline{e}}    \Big( \sum_{i=1}^n  (X_t)^2 \Big) \leq \kappa  \Big\}= {\cal E}_{[0,n]}(\kappa). \label{eqiv_s_1_n}
\end{align}

\begin{notation}
\label{not_1}
  For the feedback  code of Definition~\ref{def_rem:cp_1}, with initial state  $S_1=s$,   known to the encoder and the decoder, the above  sets  ${\cal P}_{[0,n]}(\kappa), \overline{\cal P}_{[0,n]}(\kappa), {\cal E}_{[0,n]}, \overline{\cal E}_{[0,n]}$ are replaced by ${\cal P}_{[0,n]}^s(\kappa), \overline{\cal P}_{[0,n]}^s(\kappa), {\cal E}_{[0,n]}^s, \overline{\cal E}_{[0,n]}^s$, to indicate the distributions and codes  are $ \overline{P}_t(dx_t|v^{t-1}, y^{t-1},s), t=1, \ldots,  x_1=\overline{e}_1(w,s),x_2=\overline{e}_2(w,v_1,y_1,s) \ldots, x_n=\overline{e}_n(w,v^{n-1}, y^{n-1},s)$, etc., and these depend on $s$.
\end{notation}

In the next theorem we present  our preliminary  equivalent   sequential characterization of the Cover and Pombra characterization $C_n^{fb}(\kappa)$, i.e., of  (\ref{cp_11}), 
 under encoder strategies ${\cal E}_{[0,n]}(\kappa)=\overline{\cal E}_{[0,n]}(\kappa)$,  and
   channel input distributions  ${\cal P}_{[0,n]}(\kappa)= \overline{\cal P}_{[0,n]}(\kappa)$. Unlike  the  Cover and Pombra \cite{cover-pombra1989} realization of $X^n$, given by  (\ref{cp_6}), at each time $t$, $X_t$ is driven by an orthogonal Gaussian process $Z_t$. \\

\begin{theorem}  Information structures of maximizing distributions for AGN Channels \\
\label{thm_FTFI}
Consider the AGN channel (\ref{g_cp_1}), i.e., with noise distribution ${\bf P}_{V^n}$, and the  code of Definition~\ref{def_code}. Then the following hold.\\
(a) The inequality holds, 
\begin{align}
\sup_{\overline{\cal E}_{[0,n]}(\kappa) } \sum_{t=1}^n H^{\overline{e}}(Y_t|Y^{t-1})\leq  \sup_{ \overline{\cal P}_{[0,n]}(\kappa)  }\sum_{t=1}^n H^{\overline{P}}(Y_t|  Y^{t-1})
\label{dp_1_new}
\end{align}
where the conditional (differential) entropy $H^{\overline{e}}(Y_t|Y^{t-1})$ is evaluated with respect to 
the probability distribution ${\bf P}_t^{\overline{e}}(dy_t|y^{t-1})$, defined by  
\begin{align}
 {\bf P}_t^{\overline{e}}(dy_t|y^{t-1}) = \int {\bf P}_t(dy_t|\overline{e}_t(w,v^{t-1},y^{t-1}), v^{t-1})\; {\bf P}_t^{\overline{e}}(dw, dv^{t-1}| y^{t-1}), 
  \hso t=0, \ldots, n. \label{eq_a_1_thm}
 \end{align}
and  $H^{\overline{P}}(Y_t|Y^{t-1})$ is evaluated with respect to 
the probability distribution ${\bf P}_t^{\overline{P}}(dy_t|y^{t-1})$, defined by  
\begin{align}
 {\bf P}_t^{\overline{P}}(dy_t|y^{t-1}) = \int {\bf P}_t(dy_t|x_t,v^{t-1})\; {\bf P}_t^{\overline{P}}(dx_t|v^{t-1}, y^{t-1}) \; {\bf P}_t^{\overline{P}}(dv^{t-1}|y^{t-1}), \hso t=0, \ldots, n.  \label{eq_a_2_thm}
 \end{align}
(b) The optimal channel input distribution
$\{\overline{P}(dx_t|v^{t-1},y^{t-1}), t=1, \ldots, n\}\in 
\overline{\cal P}_{[0,n]}(\kappa)$, which maximizes   $\sum_{t=1}^n H^{\overline{P}}(Y_t|  Y^{t-1})$ of part (a), i.e., the right hand side of (\ref{dp_1_new}),   is induced by an input process $X^n$, which is  conditionally Gaussian, 
with linear conditional mean, nonrandom conditional covariance, given by
\begin{align}
&{\bf E}^{\overline{P}}\Big\{X_t\Big|V^{t-1},Y^{t-1}\Big\}=\left\{ \begin{array}{lll}  \Gamma_t^1 {\bf V}^{t-1} + \Gamma_t^2 {\bf Y}^{t-1}, & \mbox{for} & t=2, \ldots, n\\
0, & \mbox{for} & t=1,
\end{array} \right.  \label{mean_1}  \\
&K_{X_t|V^{t-1}, Y^{t-1}}\tri cov\big(X_t, X_t\Big|V^{t-1}, Y^{t-1}\big)= K_{Z_t} \succeq 0, \hso t=1, \ldots, n\label{var_1}
\end{align}
and such that the average constraint holds and (\ref{g_cp_3}) is respected. \\
(c) The optimal channel input distribution
$\{\overline{P}(dx_t|v^{t-1},y^{t-1}), t=1, \ldots, n\}\in 
\overline{\cal P}_{[0,n]}(\kappa)$  of part (b), 
 is induced by a jointly  Gaussian process $X^n$, with a realization given by 
\begin{align}
&X_t = \sum_{j=1}^{t-1} \Gamma_{t,j}^1 { V}_{j} + \sum_{j=1}^{t-1}\Gamma_{t,j}^2 Y_j + Z_t,    \hso X_1=Z_1,   \hso t=2, \ldots,n,  \label{Q_1_3_s1}  \\
& \hst  = \Gamma_t^1 {\bf V}^{t-1} +\Gamma_t^2 {\bf Y}^{t-1} +Z_t,   \label{Q_1_3_s1_a}    \\ 
&Z_t\in  N(0, K_{Z_t}), \hso t=1, \ldots, n \hso \mbox{a  Gaussian sequence,}\label{Q_1_5_s1} \\
&Z_t \hso  \mbox{independent of}  \hso  (V^{t-1},X^{t-1},Y^{t-1}),  \hso t=1, \ldots, n,\label{Q_1_6_s1}\\
&Z^n \hso  \mbox{independent of}  \hso  V^{n},\\
&\frac{1}{n}  {\bf E} \Big\{\sum_{t=1}^{n} \big(X_t\big)^2\Big\} \leq \kappa,     \label{cp_e_ar2_s1}\\
&(\Gamma_t^1, \Gamma_t^2,  K_{Z_t})\in (-\infty, \infty) \times (-\infty, \infty) \times  [0,\infty) \hst \mbox{nonrandom}. \label{cp_e_ar2_s11}
\end{align}  
(d) An equivalent characterization of the $n-$FTFI capacity ${C}_n^{fb}(\kappa)$, defined  by (\ref{cp_11}), (\ref{cp_12}),  is given by 
\begin{align}
{C}_{n}^{fb}(\kappa)= & \sup_{\frac{1}{n}  {\bf E} \big\{\sum_{t=1}^{n} \big(X_t\big)^2\big\} \leq \kappa    }\sum_{t=1}^n H^{\overline{P}}(Y_t|Y^{t-1})-H(V^n) \label{ftfic_is_g}\\
=& \sup_{\frac{1}{n}  {\bf E} \big\{\sum_{t=1}^{n} \big(X_t\big)^2\big\} \leq \kappa    }\sum_{t=1}^n \Big\{ H^{\overline{P}}(I_t)-H(\hat{I}_t)\Big\} \label{ftfic_is_g_in}
\end{align}
where $I_t, \hat{I}_t$ are innovations processes defined, 
\begin{align}
I_t \tri Y_t- {\bf E}\big\{Y_t\Big|Y^{t-1}\big\}, \hst  \hat{I}_t \tri V_t- {\bf E}\big\{V_t\Big|V^{t-1}\big\} \label{inn_intr_thm1}
\end{align}
parts (b), (c) hold, and where 
the supremum is over all $(\Gamma_t^1, \Gamma_t^2, K_{Z_t}), t=1, \ldots, n$ of the realization of part (c), that induces the distribution $\overline{P}_t(dx_t|v^{t-1},y^{t-1}), t=1, \ldots, n$.  
\end{theorem}
\begin{proof} See Appendix~\ref{sect:app_thm_FTFI}.
\end{proof}

\begin{remark} For the  code of Definition~\ref{def_rem:cp_1} that assumes knowledge of the initial state $S_1=s$, it is easy to verify that $C_n^{fb}(\kappa,s)$ is directly obtained from  Theorem~\ref{thm_FTFI}, as a degenerate case (an independent derivation is easily produced  following the derivation of Corollary~\ref{cor_kim}, with slight variations). 
\end{remark}

By  utilizing Theorem~\ref{thm_FTFI},   we can derive  the converse coding theorems stated below  for the feedback codes of  Definition~\ref{def_code} and Definition~\ref{def_rem:cp_1}.\\

\begin{theorem} Converse coding theorems for  codes of Definition~\ref{def_code} and Definition~\ref{def_rem:cp_1}\\
\label{thm_ftfic}
Consider the AGN channel (\ref{g_cp_1}). \\ 
(a) Any achievable rate $R$ for the  code of Definition~\ref{def_code} satisfies
\begin{align}
R \leq & {C}^{fb}(\kappa) \tri  \lim_{n \longrightarrow \infty}\frac{1}{n} {C}_{n}^{fb}(\kappa),\\
{C}_{n}^{fb}(\kappa)= & \sup_{\overline{P}_t(dx_t|v^{t-1},y^{t-1}), t=1, \ldots, n: \: \frac{1}{n}  {\bf E} \big\{\sum_{t=1}^{n} \big(X_t\big)^2\big\} \leq \kappa    }\sum_{t=1}^n H^{\overline{P}}(Y_t|Y^{t-1})-H(V^n) \label{ftfic_is_g_def}
\end{align}
provided the supremum exists and the limit exists, where the right hand side of (\ref{ftfic_is_g_def}) is given in  Theorem~\ref{thm_FTFI}.(d). \\
(b) Any achievable rate $R$ for the  code of Definition~\ref{def_rem:cp_1} (with initial state $S_1=s$) satisfies 
\begin{align}
R \leq & {C}^{fb}(\kappa,s) \tri  \lim_{n \longrightarrow \infty}\frac{1}{n} {C}_{n}^{fb}(\kappa,s),\\
{C}_{n}^{fb}(\kappa,s)= & \sup_{\overline{P}_t(dx_t|v^{t-1},y^{t-1},s), t=1, \ldots, n: \: \frac{1}{n}  {\bf E}_s \big\{\sum_{t=1}^{n} \big(X_t\big)^2\Big|S_1\big\} \leq \kappa    }\sum_{t=1}^n H^{\overline{P}}(Y_t|Y^{t-1},s)-H(V^n|s). \label{ftfic_is_g_is_def}
\end{align}
where ${\bf E}_s\{\cdot\}$ means the expectation is for a fixed $S_1=s$, provided the supremum exists and the limit exists, and   where the right hand side of (\ref{ftfic_is_g_is_def}) is obtained from   Theorem~\ref{thm_FTFI}.(d), by replacing all conditional distributions, entropies, etc, for fixed initial state $S_1=s$ (see Notation~\ref{not_1}). 
\end{theorem}
\begin{proof} Follows from standard arguments, using Fano's  inequality (see also \cite{cover-pombra1989}) and  Theorem~\ref{thm_FTFI}.
\end{proof}

In the next remark we clarify the equivalence of 
Theorem~\ref{thm_ftfic}.(d) to  Cover and Pombra \cite{cover-pombra1989}. \ \

\begin{remark} Relation of Theorem~\ref{thm_ftfic} and Cover and Pombra \cite{cover-pombra1989}\\
\label{rem:con_cp_a} 
(a) From the  realization of $X^n$ given by (\ref{Q_1_3_s1}), we can recover  the  Cover and Pombra \cite{cover-pombra1989} realization  (\ref{cp_6}), by recursive substitution of ${ Y}^{t-1}$ into the right hand side of (\ref{Q_1_3_s1}), as follows. 
\begin{align}
X_t 
=&\sum_{j=1}^{t-1}\Gamma_{t,j}^1 V_j + \sum_{j=1}^{t-1}\Gamma_{t,j}^2 Y_j +Z_t\\
=&\sum_{j=1}^{t-1}\Gamma_{t,j}^1 V_j + \sum_{j=1}^{t-2}\Gamma_{t,j}^2 Y_j + \Gamma_{t,t-1}^2 \Big(X_{t-1} + Z_{t-1}\Big) +Z_t \\
=& \sum_{j=1}^{t-1} B_{t,j}V_{j} + \overline{Z}_t, \hst  \mbox{by recursive substitution of  $X_1, \ldots, X_{t-1}, Y_1, \ldots, Y_{t-2}$} \label{alt-cov}
\end{align}
for some $\overline{Z}_t\in (0, K_{\overline{Z}_t})$ which is jointly correlated, and some nonrandom $B_{t,j}$, as given by (\ref{cp_6}) and (\ref{cp_10}).\\ 
(b) 
Unlike the Cover and Pombra \cite{cover-pombra1989} realization of $X^n$, i.e.,   (\ref{cp_6}),  the realization of $X^n$ given by    (\ref{Q_1_3_s1}) or in vector form by  (\ref{Q_1_3_s1_a}), is such  that,  at each time $t$, $X_t$ depends on $(V^{t-1},  Y^{t-1}, Z_t)$ or in vector form on  $({\bf V}^{t-1}, {\bf Y}^{t-1}, Z_t)$,   where $Z^t$ is an innovations or orthogonal process, i.e., (\ref{Q_1_6_s1}) holds.\\ (c) In subsequent parts of the paper, we 
 derive an equivalent sequential characterization of the Cover and Pombra  $n-$FTFI capacity (\ref{cp_12}), which is  simplified further, by the use of a sufficient statistic (that satisfies a Markov recursion).
\end{remark}

To characterize $C_n^{fb}(\kappa)$ using    Theorem~\ref{thm_FTFI}.(d)  we need to compute the (differential) entropy $H(V^n)$ of $V^n$. The following lemma is useful in this respect.\\

\begin{lemma} Entropy $H(V^n)$ calculation from generalized Kalman-filter of  the PO-SS noise realization.\\
\label{lemma_POSS}
Consider the PO-SS realization of $V^n$ of  Definition~\ref{def_nr_2}.  Define the conditional covariance and conditional mean  of $S_t$ given $V^{t-1}$ by 
\begin{align}
\Sigma_{t} \tri & cov\big(S_t, S_t\Big|V^{t-1}) = {\bf E}\Big\{\Big(S_t- \hat{S}_{t}\Big)\Big(S_t-\hat{S}_{t} \Big)^T\Big|V^{t-1}\Big\}, \hso \hat{S}_{t} \tri {\bf E}\Big\{S_t\Big|V^{t-1}\Big\},  \hso  t=2, \ldots, n, \label{cond_GK_1} \\
\Sigma_1 \tri & cov\big(S_1, S_1) =K_{S_1}, \hso \hat{S}_1 \tri \mu_{S_1}. \label{cond_GK_2}
\end{align} 
Then the following hold.\\
(a) The  conditional distribution of $V_t$ conditioned on $V^{t-1}$ is Gaussian, i.e., 
\begin{align}
{\bf P}_{V_t|V^{t-1}} \in N(\mu_{V_t|V^{t-1}}, K_{V_t|V^{t-1}}), \hst t=1, \ldots, n \hso 
\end{align}
where $\mu_{V_t|V^{t-1}}\tri {\bf E}\big\{V_t\Big|V^{t-1}\big\}, K_{V_t|V^{t-1}}\tri cov\big(V_t, V_t\Big|V^{t-1})$.\\
(b) The conditional mean and covariance  $\mu_{V_t|V^{t-1}}, K_{V_t|V^{t-1}}$ are given by  the Generalized Kalman-filter recursions, as follows. \\
(i) The optimal mean-square error estimate $\hat{S}_t$ satisfies the generalized Kalman-filter recursion   
\begin{align}
&\hat{S}_{t+1}=A_{t} \hat{S}_{t}+ M_{t}(\Sigma_t)  \hat{I}_t, \hso \hat{S}_{1}=\mu_{S_1}, \label{kal_fil_noise} \\
& M_{t}(\Sigma_t) \tri \Big( A_{t}  \Sigma_{t} C_{t}^T+B_{t} K_{W_{t}}N_{t}^T\Big)\Big(N_{t} K_{W_{t}}N_{t}^T+ C_{t} \Sigma_{t} C_{t}^T \Big)^{-1}, \\
&\hat{I}_t \tri V_t -{\bf E}\Big\{V_t\Big|V^{t-1}\Big\}=  V_t - C_t \hat{S}_{t}= C_t\big(S_{t}- \hat{S}_{t}\big) + N_t W_t, \hso t=1, \ldots, n, \label{inn_po_1} \\
& \hat{I}_t \in N(0, K_{\hat{I}_t}), \hso t=1, \ldots, n \hso \mbox{is an orthogonal innovations process, i.e., $\hat{I}_t$ is independent of}\nonumber \\
&\mbox{ $\hat{I}_s$, for all $t \neq s$, and $\hat{I}_t$ is independent of  $V^{t-1}$},  \label{inn_po_2}\\
&K_{\hat{I}_t} \tri cov(\hat{I}_t, \hat{I}_t)= C_t \Sigma_t C_t^T +N_t K_{W_t} N_t^T. \label{cov_in_noise}
\end{align}
(ii) The error $E_t \tri S_t-\hat{S}_t$  satisfies the recursion
\begin{align}
E_{t+1} = &M_t^{CL}(\Sigma_t)E_t +  \Big(B_t- M(\Sigma_t) N_t\Big) W_t, \hso E_1= S_1-\hat{S}_1, \hso t=1,\ldots, n,  \label{dre_1_a}  \\
M_t^{CL}(\Sigma_t)\tri & A_t- M_t(\Sigma_t)C_t.
\end{align}
(iii) The  covariance of the error is such that  ${\bf E} \big\{E_t E_t^T\big\}=\Sigma_t$ and satisfies the generalized matrix DRE 
\begin{align}
\Sigma_{t+1}= &A_{t} \Sigma_{t}A_{t}^T  + B_{t}K_{W_{t}}B_{t}^T -\Big(A_{t}  \Sigma_{t}C_{t}^T+B_{t}K_{W_{t}}N_{t}^T  \Big) \Big(N_{t} K_{W_{t}} N_{t}^T+C_{t}  \Sigma_{t} C_{t}^T\Big)^{-1}\nonumber \\
& \hst . \Big( A_{t}  \Sigma_{t}C_{t}^T+ B_{t} K_{W_{t}}N_{t}^T  \Big)^T,  \hso t=1, \ldots, n, \hso \Sigma_{1}=K_{S_1}\succeq 0, \hso \Sigma_t \succeq 0. \label{dre_1}
\end{align}
(iv) The conditional mean and covariance $\mu_{V_t|V^{t-1}}, K_{V_t|V^{t-1}}$ are given by  
\begin{align}
\mu_{V_t|V^{t-1}}=&C_t \hat{S}_{t},  \hso t=1, \ldots, n,  \\
K_{V_t|V^{t-1}}=&K_{\hat{I}_t} =C_t \Sigma_{t} C_t^T + N_t K_{W_t} N_t^T,  \hso t=1, \ldots, n. \label{inn_noise}  
\end{align}
(v) The entropy of $V^n$, is given by 
\begin{align}
H(V^n)=\sum_{t=1}^n H(\hat{I}_t)=   \frac{1}{2}\sum_{t=1}^n \log\Big(2\pi e \Big[C_t \Sigma_{t} C_t^T + N_t K_{W_t} N_t^T\Big] \Big) \label{entr_noise}
\end{align}
\end{lemma}
\begin{proof} (a), (b).(i)-(iv). The generalized Kalman filter of the PO-SS realization of $V^n$ and accompanied statements can be found in many textbooks, i.e., \cite{caines1988}. However,  it is noted that $\hat{I}_t, t=2, \ldots, n$, $\hat{I}_1=V_1$ are all independent Gaussian. For example, to  show   (\ref{dre_1_a}) we write the recursion for  $E_t=S_t-\hat{S}_t$ using part (i) and the realization of $S_t$. (b).(v) By the chain rule of joint entropy then 
\begin{align}
H(V^n)=&H(V_1) + \sum_{t=2}^n H(V_t|V^{t-1})   \\
=&H(V_1) + \sum_{t=2}^n H(V_t-{\bf E}\Big\{V_t\Big|V^{t-1}\Big\}|V^{t-1})\\
=&H(V_1) + \sum_{t=2}^n H(\hat{I}_t), \hso \hso \mbox{by orthogonality of $\hat{I}_t\tri V_t-{\bf E}\Big\{V_t\Big|V^{t-1}\Big\}$ and $V^{t-1}$} \label{entropy_1}
\end{align}
 From (\ref{entropy_1}) and (\ref{inn_noise}), then follows (\ref{entr_noise}),  from the entropy formula of  Gaussian RVs.
\end{proof}

From Lemma~\ref{lemma_POSS} follows directly the next corollary of the entropy  $H(V^n|s)$, when $S_1=s$ is fixed.\\

\begin{corollary} Conditional entropy $H(V^n|s), S_1=s$  of  the PO-SS noise realization.\\
\label{cor_POSS}
Consider the PO-SS realization of $V^n$ of  Definition~\ref{def_nr_2}, for fixed $S_1=s$, and denote the state process generated by recursion (\ref{real_1a}), by\footnote{We often use the notation $S_t=S_t^s$ to emphasize that the $S_t$ process is generated for  $S_1=S_1^s=s$  fixed.},  $S_t=S_t^s, t=2, \ldots, n, S_1=S_1^s=s$. Replace the  conditional covariance and conditional mean  (\ref{cond_GK_1}) and (\ref{cond_GK_2}), by 
\begin{align}
\Sigma_{t}^{s} \tri & cov\big(S_t^s, S_t^s\Big|V^{t-1},S_1^s) = {\bf E}\Big\{\Big(S_t^s- \hat{S}_{t}^{s}\Big)\Big(S_t^s-\hat{S}_{t}^{s} \Big)^T\Big|V^{t-1},S_1^s\Big\}, \label{cond_GK_1_n} \\
\hat{S}_{t}^{s} \tri& {\bf E}\Big\{S_t^s\Big|V^{t-1}, S_1^s\Big\},  \hso  t=2, \ldots, n, \hso  \hso  S_1^s=s, \hso    \hat{S}_1^{s} \tri s, \hso 
\Sigma_1^{s} \tri  cov\big(S_1^s, S_1^s|S_1^s) =0. \label{cond_GK_2_n}
\end{align} 
Then all statements of Lemma~\ref{lemma_POSS} hold, with the changes, 
\begin{align}
\Sigma_t \longmapsto \Sigma_t^s, \hso \Sigma_1^s=0, \hso {\bf P}_{V_t|V^{t-1}} \longmapsto {\bf P}_{V_t|V^{t-1},S_1^s}, \hso \hat{S}_t \longmapsto \hat{S}_t^s,\hso \hat{S}_1^s=s, \hso \mbox{etc}, \; t=1,\ldots, n. 
\end{align}
In particular, the conditional entropy of $V^n$ conditioned on $S_1=S_1^s=s$,  is given by 
\begin{align}
H(V^n|s)=\frac{1}{2}\sum_{t=1}^n \log\Big(2\pi e \Big[C_t \Sigma_{t}^s C_t^T + N_t K_{W_t} N_t^T\Big] \Big) \label{entr_noise_cond}
\end{align}
where $\Sigma_{t}^s, t=2, \ldots, n$ satisfies the generalized DRE (\ref{dre_1}) with initial condition $\Sigma_1^s=0$. 
\end{corollary}
\begin{proof} Follows directly from Lemma~\ref{lemma_POSS} and  (\ref{cond_GK_1_n}),(\ref{cond_GK_2_n}) . 
\end{proof}

Next we introduce an example of a PO-SS realization of the noise that we often use in  the paper. \\

\begin{example} 
\label{ex_1_poss}
A time-varying PO-SS$(a_t,c_t, b_t^1, b_t^2, d_t^1, d_t^2)$ noise realization is  defined by 
\begin{align}
&S_{t+1} = a_{t} S_{t} + b_t^1 W_t^1+b_t^2 W_t^2, \hso t=1, 2, \ldots, n-1 \label{ex_1_1_poss}\\
&V_t =c_t S_t +d_t^1 W_t^1 + d_t^2 W_t^2,\hso t=1, \ldots, n,\label{ex_1_1_poss_a} \\
&S_1\in  N(\mu_{S_1}, K_{S_1}),\hso  K_{S_1}\geq 0,  \hso W_t^i \in N(0,  K_{W_t^i}), \hso K_{W_t^i}\geq 0,\hso  i=1,2,\hso t=1, \ldots, n, \\
& \mbox{$W^{1,n}$ and $W^{2,n}$  indep. seq. and indep. of $S_1$}, \label{ex_1_2_poss} \\
&a_t \in {\mathbb R}, \hso c_t \in  {\mathbb R}, \hso b_t^i \in  {\mathbb R}, \hso   d_t^i \in  {\mathbb R}, \hso i=1,2, \forall t \hso \mbox{are nonrandom},\\
&b_t \circ b_t \tri   \big(b_t^1\big)^2 K_{W_t^1} + \big(b_t^2\big)^2K_{W_t^2},   \hso b_t \circ d_t \tri   b_t^1 K_{W_t^1}d_t^1 +b_t^ 2K_{W_t^2}d_t^2, \\
&      d_t\circ d_t\tri  \big(d_t^1\big)^2 K_{W_t^1} + \big(d_t^2\big)^2K_{W_t^2}>0, \; \forall t.\label{ex_1_3_poss}
\end{align}
\end{example}

The next corollary is an application of Lemma~\ref{lemma_POSS}  to the time-varying PO-SS noise of Example~\ref{ex_1_poss}. \\

\begin{corollary} 
\label{entr_poss}
The entropy $H(V^n)$ of  the PO-SS$(a_t,c_t, b_t^1, b_t^2, d_t^1, d_t^2)$ noise of Example~\ref{ex_1_poss} is computed from  Lemma~\ref{lemma_POSS}  with the following changes:
\begin{align}
&C_t \longmapsto c_t, \hso A_t \longmapsto a_t, \hso  B_t K_{W_t} N_t^T \longmapsto b_t \circ d_t, \nonumber \\
& B_tK_{W_t}B_t^T \longmapsto b_t \circ b_t, \hso N_t K_{W_t} N_t^T \longmapsto d_t \circ d_t, \hso   t=1,\ldots, n.  \label{poss_ex_1}
\end{align}
\end{corollary}
\begin{proof}
This is easily verified.
\end{proof}

From  Corollary~\ref{entr_poss}  we have the following observations.\\

\begin{remark} Consider the PO-SS$(a_t,c_t, b_t^1, b_t^2, d_t^1, d_t^2)$
\label{rem_poss_1}
 noise of Example~\ref{ex_1_poss}. Then the following hold.\\
(a) Consider  the code of Definition~\ref{def_nr_2}. At each time $t$, the optimal channel input process $X^n$ is either realized 
 by (\ref{cp_6}),  
or equivalently 
by (\ref{Q_1_3_s1}), i.e.,  $X_t=\sum_{j=1}^{t-1} B_{t,j}V_{t,j} + \overline{Z}_t=  \sum_{j=1}^{t-1}\Gamma_{t,j}^1 V_j + \sum_{j=1}^{t-1}\Gamma_{t,j}^2 Y_j +Z_t$.
Moreover,   $X_t$ cannot be expressed in terms of the state $S^t$, because by  (\ref{ex_1_1_poss}) and (\ref{ex_1_1_poss_a}) the noise sequence  $V^{t-1}$  does not specify  $S^t$, for $t=1, \ldots, n$. 
\\
(b) Consider the code of Definition~\ref{def_rem:cp_1}, i.e.,  with  a fixed initial state $S_1=S_1^s=s$. By Corollary~\ref{cor_POSS} using  (\ref{poss_ex_1}), then  $H(V^n|s)$ is computed from Lemma~\ref{lemma_POSS}, with $\Sigma_1=\Sigma_1^s=0$, and (\ref{entr_noise_cond}) reduces to   
\bea
H(V^n|s)=\frac{1}{2}\sum_{t=1}^n \log\Big(2\pi e \Big[\big(c_t\big)^2 \Sigma_{t}^s  + d_t\circ d_t \Big] \Big) \label{ent_fin}
\eea
where $\Sigma_t^s$ is the solution of (\ref{dre_1}) with $\Sigma_1=\Sigma_1^s=0$ (using (\ref{poss_ex_1})).
\end{remark}

We also apply our  results to various versions of the autoregressive moving average (ARMA) noise model, such as,  the double-side and single-sided, stationary version of the ARMA noise, previously analyzed  in  \cite{kim2010} and  in many other  papers, to illustrate fundamental differences of Case I) and Case II) formulations. \\

\begin{example} The time-invariant  ARMA$(a,c)$ noise\\
\label{ex_1_1_n}
(a) The time-invariant one-sided, stable or unstable,   autoregressive moving average (ARMA$(a,c), a \in (-\infty,\infty), c\in (-\infty, \infty)$) noise is  defined by 
\begin{align}
&V_t = c V_{t-1} + W_t -a W_{t-1}, \hso \forall t \in {\mathbb Z}_+\tri \{1, 2, \ldots\},  \label{ex_1_1}\\
&V_0\in  N(0, K_{V_0}),\hso  K_{V_0}\geq 0, \hso W_0 \in N(0,  K_{W_0}), \hso K_{W_0}\geq 0,  \hso W_t \in N(0,  K_{W}), \hso   K_{W}>0,\hso \forall t \in {\mathbb Z}_+, \\
& \mbox{$\{W_0, W_1, \ldots, W_n\}$ indep. seq. and indep. of $V_0$}, \label{ex_1_2} \\
&c \in (-\infty,\infty), \hso a \in (-\infty,\infty), \hso c\neq a .\label{ex_1_3}
\end{align}
To express the AR$(a,c)$ in state space form we  define the state variable of the noise by 
\begin{align}
S_t \tri \frac{c V_{t-1} -a W_{t-1}}{c-a}, \hso \forall t \in {\mathbb Z}_+    \label{ex_1_4}
\end{align}
Then the state space realization of $V^n$ is 
\begin{align}
&S_{t+1}=c S_t +W_t, \hso \forall t \in {\mathbb Z}_+,\label{ex_1_5} \\
&V_t= \Big(c-a\Big) S_t + W_t, \hso \forall t \in {\mathbb Z}_+, \label{ex_1_6}\\
&K_{S_{1}}= \frac{\big(c\big)^2 K_{V_0} +\big(a\big)^2 K_{W_0}}{\Big(c-a\Big)^2}, \hso K_{V_0}\geq 0,\hso K_{W_0}\geq 0 \hst \mbox{both given}.\label{ex_1_7}
\end{align}
We note that the AR$(a,c)$ is not necessarily  stationary or asymptotically stationary.\\ 
A special case of the AR$(a,c)$ is the AR$(c)$ noise (i.e.,  with $a=0$)  defined by 
\begin{align}
&V_t = c V_{t-1} + W_t, \hso t=1, 2, \ldots, \hso K_{V_0}\geq 0, \hso K_{W}>0. \label{ex_1_1_AR1}
\end{align}
(b) Double-Sided Wide-Sense Stationary  ARMA$(a,c), a \in [-1,1], c\in (-1,1)$ Noise. A   double-sided wide-sense stationary  ARMA$(a,c)$ noise is  defined by 
\begin{align}
V_t = cV_{t-1} + W_t -a W_{t-1}, \hso \forall t \in {\mathbb Z}\tri  \{\ldots, -1,0,1, \ldots\}, \hso |a|\leq 1, \: |c|<1. \label{ex_2_2_in}
\end{align}
where  $W_t , \forall t \in {\mathbb Z}$ is an independent and identically distributed Gaussian sequence, i.e., $W_t \in N(0, K_W)$, $\forall t$. The power spectral density (PSD) of the  wide-sense stationary noise  is  (this corresponds to  \cite[eqn(43) with $L=1$]{kim2010}) is given by    
\begin{align}
S_V(e^{j\theta}) \tri& K_W\frac{\Big( 1-a e^{i\theta}\Big)\Big(1-a e^{-i\theta}\Big)}{\Big( 1- c e^{i \theta}\Big)\Big(1- c e^{-i\theta}\Big)}, 
\hso |c|<1,\hso |a|\leq 1, \hso c\neq a, \hso K_W>0. \label{ex_2_1_in}
\end{align} 
We define the  state process by 
\begin{align}
S_t \tri \frac{c V_{t-1} -a W_{t-1}}{c-a},   \hso \forall t \in {\mathbb Z}  . \label{ex_2_3_in}
\end{align}
Then the stationary state space realization of $V_t, \forall t \in {\mathbb Z}$ is 
\begin{align}
&S_{t+1}=c S_t +W_t, \hso \forall t\in {\mathbb Z},\label{ex_2_4_in} \\
&V_t= \Big(c-a\Big) S_t + W_t, \hso \forall t\in {\mathbb Z} \label{ex_2_5_in} 
\end{align}
provided the initial covariances, $cov(S_t,S_t), cov(S_t,V_t), cov(V_t,V_t)$ are chosen appropriately to ensure stationarity (see Proposition~\ref{pr_ex_2}). \\
(c) One-sided Wide-Sense Stationary  ARMA$(a,c), a \in [-1,1], c \in (-1,1)$.     The one-sided wide-sense stationary  ARMA$(a,c)$ noise is defined as in part (b) with  $\forall t\in {\mathbb Z}\tri \{\ldots, -1,0,1, \ldots,\} $ replaced by    $\forall t \in {\mathbb Z}_+\tri \{1,2, \ldots,\}$ and (\ref{ex_2_3_in})-(\ref{ex_2_5_in}) hold,  $\forall t \in {\mathbb Z}_+$, provide the initial covariances are chosen appropriately (see Proposition~\ref{pr_ex_2}). 
\end{example}

For the AR$(a,c)$ noise, in  the next remark we clarify  differences of the feedback codes of  Definition~\ref{def_code} and Definition~\ref{def_rem:cp_1}, and of  Case I) formulation versus Case II) formulation (and discuss implication to results in   \cite{kim2010,liu-han2019,gattami2019,ihara2019,li-elia2019}).\\

\begin{remark}   ARMA$(a,c)$ noise of Example~\ref{ex_1_1_n}\\
\label{rem_kim}
(a) Consider any of the AR$(a,c)$ of Example~\ref{ex_1_1_n}. For  the code of Definition~\ref{def_nr_2} the channel input process  $X^n$ cannot be expressed in terms of the state $S^n$ (see also  Remark~\ref{rem_poss_1}.(a)). \\
(b) Consider the nonstationary AR$(a,c), a\in (-\infty,\infty), c\in (-\infty,\infty)$ of Example~\ref{ex_1_1_n}.(a).\\
(i) Assume the code of Definition~\ref{def_rem:cp_1}, with  initial state $V_0=v_0$ known to the encoder. By   (\ref{ex_1_4}),  
\bea
S_1=S_1^{v_0}= \frac{c v_{0} -a W_{0}}{c-a}, \hso V_0=v_0
\eea
 and hence knowledge of  $V_0=v_0$ at the encoder does not determine  $S_1^{v_0}$, because for this to hold the encoder requires knowledge  of $W_0$. It then follows that $H(V^n|v_0)$ is computed from Corollary~\ref{cor_POSS},
\bea 
 \Sigma_1=\Sigma_1^{v_0}=\frac{\big(a)^2K_{W_0}}{\big(c-a\big)^2}, \hso \mbox{and (\ref{entr_noise_cond}) reduces to }\hso   H(V^n|v_0)=\frac{1}{2}\sum_{t=1}^n \log\Big(2\pi e \Big[\big(c\big)^2 \Sigma_{t}^{v_0}  + K_{W} \Big] \Big)
\eea 
  where $\Sigma_t^{v_0}$ is the solution of (\ref{dre_1}) with initial data $\Sigma_1=K_{S_1}=\Sigma_1^{v_0},  K_{W_0}\geq 0$.\\
(ii) Assume the code of Definition~\ref{def_rem:cp_1}, with  initial state $S_1=s$ or  $(V_0, W_0)=(v_0,w_0)$, known to the encoder. Then  by Corollary~\ref{cor_POSS},
\bea
H(V^n|v_0,w_0)=\frac{1}{2}\sum_{t=1}^n \log\Big(2\pi e   K_{W}  \Big).
\eea
By   (\ref{ex_1_4}), $S_1 \tri \frac{c V_{0} -a W_{0}}{c-a}$, and  a    necessary condition for Conditions 1 of  Section~\ref{sect:motivation} to hold is:  both $(V_0,W_0)=(v_0,w_0)$ are known to the encoder and the decoder.  \\
(c) The statements of parts (a), (b) also hold for the double-sided and the one-sided wide-sense stationary   AR$(a,c), a\in [-1,1], c\in (-1,1)$ of Example~\ref{ex_1_1_n}.(b), (c).\\
(d) Case II) formulation discussed in Section~\ref{sect:motivation}, requires Conditions 1 and 2 to hold. For any of the AR$(a.c)$ noise models, then  Conditions 1 and 2 hold if and only if $S_1=s_1$ or  $(V_0, W_0)=(v_0,w_0)$ are known to the encoder. Clearly, the values of   $H(V^n)$  under Case I) formulation is fundamentally different from the value of  $H(V^n|s), S_1=s$ under Case II) formulation. Consequently,  in general,  $C_n^{fb}(\kappa)$ given by  (\ref{ftfic_is_g}) is fundamentally different from $C_n^{fb}(\kappa,s)$, i.e., that corresponds to a fixed initial state $S_1=s$, known to the encoder and the decoder, and  to the channel input distribution. \\
(e) From parts (a)-(d) follows the  characterization of feedback capacity  for the stationary ARMA$(a,c), a \in[-1,1], c\in (-1,1)$ given in   \cite[Theorem~6.1, $C_{FB}$]{kim2010} (which is derived based on \cite[Lemmas~6.1]{kim2010})  presupposed the encoder and the decoder assumed knowledge  of $S_1=S_1^s=s$ (similarly for  \cite{liu-han2019,gattami2019,li-elia2019}).
\end{remark}

 In the next proposition, we state conditions for the stable realizations of Example~\ref{ex_1_1_n}.(a), i.e., AR$(a,c), a\in [-1,1], c\in (-1,1)$ to be asymptotically stationary, and  for the realizations of Example~\ref{ex_1_1_n}.(b), (c) to be stationary. We should emphasize that for stationary noise,  we need to  determine  the  initial conditions of the  generalized Kalman-filter of Lemma~\ref{lemma_POSS} to correspond to the stationary noise.  \\

\begin{proposition} Asymptotically stationary and  stationary  ARMA$(a,c)$ noises of Example~\ref{ex_1_1_n} \\
\label{pr_ex_2}
(a) The realization of the double-sided ARMA$(a,c), a \in [-1,1], c\in (-1,1)$ noise of Example~\ref{ex_1_1_n}.(b) is stationary if the following conditions hold. 
\begin{align}
&d_{11} \tri cov\big(S_t,S_t\big)=K_{S_{t}}, \hso d_{12}\tri cov\big(S_t, V_t\big)=K_{S_t,V_t}, \hso d_{22}\tri cov(V_t,V_t)=K_{V_t},  \hso \mbox{are constant  $\forall t \in {\mathbb Z}$}. \label{ex_2_6}
\end{align}
where  the constants $(d_{11}, d_{12}, d_{22})$ are given by 
\begin{align}
&d_{11} = \frac{K_W}{1- c^2}, \hst d_{12}=\frac{\big(c-a\big)K_W}{1-c^2}, \hst d_{22}= \frac{\big(c-a\big)^2K_W}{1-c^2}+K_W. \label{ex_2_7}
\end{align}
Similarly the one-sided ARMA$(a,c), a \in [-1,1], c\in (-1,1)$ noise of Example~\ref{ex_1_1_n}.(c) is stationary if the above equations hold  $\forall t\in {\mathbb Z}_+\tri \{1,2,\ldots\}$.\\
(b) The realization of the ARMA$(a,c)$ noise of Example~\ref{ex_1_1_n}.(a) is asymptotically stationary if   $a\in [-1,1], c\in (-1,1)$. \\ 
(c) For the stationary  realization of part (a) the optimal conditional variance and conditional mean of $S_t$ from $(V_0, V_1,V_2, \ldots, V_{t-1})$, i.e.,  $\Sigma_{t} \tri cov\big(S_t, S_t\Big|V^{t-1},V_0), \hat{S}_{t} \tri {\bf E}\Big\{S_t\Big|V^{t-1}, V_0\Big\}$ are defined by the  generalized Kalman-filter given by 
\begin{align}
&\hat{S}_{t+1}=c \hat{S}_{t}+ \Big( \big(c\big)^2  \Sigma_{t}+K_{W}\Big)\Big(K_{W}+ \big(c-a\big)^2 \Sigma_{t} \Big)^{-1} \Big(V_t - \big(c-a\big) \hat{S}_t\Big), \hso t=1,2,  \ldots,  \label{arma_s1}\\
&\Sigma_{t+1}= \big(c\big)^2 \Sigma_{t}  + K_{W_{t}} -\Big( \big(c\big)^2  \Sigma_{t}+K_{W}  \Big)^2 \Big(K_{W} +\big(c\big)^2  \Sigma_{t} \Big)^{-1}\label{arma_s2}
\end{align}
initialized at the  initial data
\begin{align}
&\hat{S}_1 \tri {\bf E}\Big\{S_1 \Big|V_0\Big\}= \frac{cd_{12}+K_W}{d_{22}}V_0 , \label{in_data_1}\\
& \Sigma_1 \tri cov(S_1,S_1\Big|V_0)  =   d_{11} - \frac{\big(d_{12}\big)^2}{d_{22}}. \label{in_data_2}
\end{align}
(i) If the conditioning information is  $(V_{-N}, \ldots, V_0, V_1,V_2, \ldots, V_{t-1})$ then the generalized Kalman-filter (\ref{arma_s1}), (\ref{arma_s2}) still hold, and initialized at the initial data
\begin{align}
&\hat{S}_{-N} = \frac{cd_{12}+K_W}{d_{22}}V_{-N},  \label{in_data_3}   \\
& \Sigma_{-N}  =   d_{11} - \frac{\big(d_{12}\big)^2}{d_{22}}. \label{in_data_4}
\end{align}
(ii) If the inital data $V_0$ is not available then the generalized Kalman-filter is initialized at initial data $\hat{S}_1=0$, $\Sigma_1 = cov(S_1,S_1)=d_{11}$.
\end{proposition} 
\begin{proof} See Appendix~\ref{sect:app_pr_ex_2}.
\end{proof}

\ \

\begin{remark} Consider the stationary double-sided or one-sided ARMA$(a,c), a \in [-1,1], c\in (-1,1)$  of Example~\ref{ex_1_1_n}. 
From  Proposition~\ref{pr_ex_2}, and in particular the initial data $\hat{S}_1, \Sigma_1$ stated  in  (\ref{in_data_1}), (\ref{in_data_2}), it is  clear that even if the encoder and the decoder know the initial state  $V_0$, then $H(V^n|v_0) \neq \frac{1}{2}\sum_{t=1}^n \log\Big(2\pi e   K_{W}  \Big)$. In this case, the value of $C_n^{fb}(\kappa,v_0)$ defined by (\ref{ftfic_is_g_is_def}) is fundamentally different from the formulation in \cite{yang-kavcic-tatikonda2007} and \cite{kim2010}  that let to  the characterization of feedback capacity \cite[Theorem~6.1]{kim2010}. 
\end{remark}

In the next corollary we further clarify the difference between Case I) formulation and Case II) formulation, by stating the   analog of Theorem~\ref{thm_FTFI} for  
   the code of Definition~\ref{def_rem:cp_1}, i.e., when $S_1=S_1^s=s$ is fixed.\\

\begin{corollary} $n-$FTFI capacity for feedback code of Definition~\ref{def_rem:cp_1}\\
\label{cor_nftfic_s}
Consider the time-varying AGN channel defined by (\ref{g_cp_1}), driven by a noise with the PO-SS realization  of Definition~\ref{def_nr_2}, and the code of Definition~\ref{def_rem:cp_1}, with initial state $S_1=S_1^s=s$ fixed.\\
Then the following hold.\\
(a) The  $n-$FTFI capacity $C_n^{fb}(\kappa,s)$ is given by 
\begin{align}
C_n^{fb}(\kappa,s) \tri &  \sup_{ \frac{1}{n} {\bf E}_s\Big\{\sum_{t=1}^n \big(X_t\big)^2 \Big|S_1\Big\}\leq \kappa}  H^{\overline{P}}(Y^n|s)-H(V^n|s). \label{CP_F_new_nn}\\
X_t =& \Gamma^0 s+  \sum_{j=1}^{t-1} \Gamma_{t,j}^1 { V}_{j} + \sum_{j=1}^{t-1}\Gamma_{t,j}^2 Y_j + Z_t,    \hso t=1, \ldots,n. \label{Q_1_3_s1_is} 
 \end{align}
where the supremum is over all $(\Gamma^0, \Gamma_{t,j}^1, \Gamma_{t,j}^2, K_{Z_t}), j=1, \ldots, t-1, t=1, \ldots, n$ of the realization of $X^n$, that induces the distribution $\overline{P}_t(dx_t|v^{t-1},y^{t-1},s), t=1, \ldots, n$, and   all statements of  Theorem~\ref{thm_FTFI} and Lemma~\ref{lemma_POSS} hold, with the  conditional distribitions,  expectations,  and entropies replaced  by the corresponding expressions with fixed   $S_1=S_1^s=s$.\\
(b) A necessary condition  for Condition 2 of  Section~\ref{sect:motivation} to hold is \\
(i) $N_tW_t$ uniquely defines $C_{t+1}B_t W_t, \forall t$.\\ 
Moreover, if (i) holds then  the  entropy $H(V^n|s)$ of part (a)   is given by 
\bea
H(V^n|s)=\frac{1}{2}\sum_{t=1}^n \log\Big(2\pi e N_t K_{W_t} N_t^T \Big). \label{cond_ent}
\eea 
\end{corollary}
 \begin{proof} See Appendix~\ref{sect:app_cor_nftfic_s}.
 \end{proof}

In the next remark we illustrate  that $H(V^n|s)$ given by (\ref{cond_ent}) follows directly from  Lemma~\ref{lemma_POSS}, by fixing $S_1=S_1^s=s$, and assuming $N_tW_t$ uniquely defines $C_{t+1}B_t W_t, \forall t$. \\

\begin{remark} The $n-$FTFI capacity for code of Definition~\ref{def_code} versus code of Definition~\ref{def_rem:cp_1}.\\   
\label{rem_ss_is}
Consider the generalized Kalman-filter of  the PO-SS noise realization, of 
Lemma~\ref{lemma_POSS}, and assume the initial state of the noise $S_1$ is known, i.e., $S_1=S_1^s=s$ or $S_1=S_1^s=s=0$, and  $N_tW_t$ uniquely defines $C_{t+1}B_t W_t, \forall t$.  Then all statements of Lemma~\ref{lemma_POSS} hold, by replacing $(\Sigma_t, \hat{S}_t)$ by $(\Sigma_t^{s}, \hat{S}_t^{s})$ for $t=1,2, \ldots,$. Since $\Sigma_{t}^{s}$ satisfies the generalized DRE (\ref{dre_1}) with initial condition $\Sigma_1^{s}=0$, then it is easy to deduce that $\Sigma_t^{s}=0$, for $t=1,2,\ldots, n$ is a solution. Substituting $\Sigma_t^{s}=0, t=1,2,\ldots, n$ in  (\ref{entr_noise}) we obtain (\ref{cond_ent}), as expected.\\
On the other hand, for the code of Definition~\ref{def_code}, by Theorem~\ref{thm_FTFI}.(d) the right hand side of the $n-$FTFI capacity $C_n^{fb}(\kappa)$ involves $H(V^n)$, which is  computed  using  the  generalized Kalman-filter of  Lemma~\ref{lemma_POSS}.
\end{remark}



\subsection{A Sufficient Statistic Approach to the  Characterization of $n-$FTFI Capacity of AGN Channels Driven by PO-SS Noise Realizations}
\label{sect_POSS_SS}
The characterization of  the $n-$FTFI capacity via  (\ref{cp_12}) (which is equivalently given in   Theorem~\ref{thm_FTFI}.(d)), although compactly represented, is   not very practical,  because the input process $X^n$ is not expressed in terms of a {\it sufficient statistic} that summarizes the information of the channel input strategy \cite{kumar-varayia1986}.     \\
In this section,  we wish to identify a {\it sufficient statistic} for the input process $X_t$, given by (\ref{Q_1_3_s1}),  called the {\it state of the input}, which  summarizes  the  information contained in $(V^{t-1}, Y^{t-1})$. It will then become apparent that the characterization of the $n-$FTFI capacity for the Cover and Pombra formulation and code of Definition~\ref{def_code}, can be expressed as a functional of {\it two generalized matrix DREs}.\\ 
First, we invoke Theorem~\ref{thm_FTFI} and Lemma~\ref{lemma_POSS} to  show that for each time $t$,  $X_t$ is expressed as
 \begin{align}
&X_t=  \Lambda_{t}\Big(\hat{S}_t- {\bf E}\Big\{\hat{S}_t\Big|Y^{t-1}\Big\}\Big) + Z_t, \hso  t=1, \ldots, n,  \label{suf_s} \\
&\hat{S}_t \tri {\bf E}\Big\{S_t\Big|V^{t-1}\Big\}, \hst \widehat{\hat{S}}_t\tri {\bf E}\Big\{\hat{S}_t\Big|Y^{t-1}\Big\}
\end{align}
which  means, at each time $t$, the  state of the channel input process $X_t$ is $\Big(\hat{S}_t, \widehat{\hat{S}}_t\Big)$. We show that  $\widehat{\hat{S}}_t$ satisfies another  generalized Kalman-filter recursion. \\
Now,  we prepare to prove  (\ref{suf_s}) and the   main theorem.   We start with preliminary calculations.    
\begin{align}
{\mb P}\big\{Y_t \in dy \Big| Y^{t-1}, X^t\big\} =&{\bf P}_t(dy |X_t,V^{t-1}), \ \ t=2, \ldots, n, \hso \mbox{by channel definition} \\
=&{\bf P}_t(dy |X_t,V^{t-1},\hat{S}^{t}), \hso  \mbox{by $\hat{S}_t= {\bf E}\Big\{S_t\Big| V^{t-1}\Big\}$}\\
=&{\bf P}_t(dy |X_t,V^{t-1},\hat{S}_{t}, \hat{I}^{t-1}), \hso  \mbox{by (\ref{inn_po_1}), i.e., $V_t=C_t \hat{S}_t+\hat{I}_t$}\\
=&{\bf P}_t(dy |X_t,\hat{S}_{t}), \hso  \mbox{by $Y_t=X_t+V_t=X_t+C_t\hat{S}_t+\hat{I}_t$ and (\ref{inn_po_2})}. \label{PO_state_1}
\end{align}
At $t=1$ we also have 
${\mb P}\big\{Y_1 \in dy \Big| X_1\big\}={\bf P}_1(dy |X_1)$. By (\ref{PO_state_1}), it follows that the conditional distribution of $Y_t$ given $Y^{t-1}=y^{t-1}$ is 
 \begin{align}
 {\bf P}_t(dy_t|y^{t-1})=&\int {\bf P}_t(dy |x_t,\hat{s}_{t}){\bf P}_t(dx_t|\hat{s}_t,y^{t-1}) {\bf P}_t(d\hat{s}_t|y^{t-1}), \hso t=2, \ldots, n, \label{PO_tr_1} \\
{\bf P}_1(dy_1)=&\int {\bf P}_1(dy |x_t,\hat{s}_1){\bf P}_1(dx_1|\hat{s}_1) {\bf P}_1(d \hat{s}_1). \label{PO_tr_2}
 \end{align}
From the above distributions,  at each time $t$, the distribution of $X_t$ conditioned on $(V^{t-1}, Y^{t-1})$, given in Theorem~\ref{thm_FTFI}, is also expressed as a linear functional of  $(\hat{S}_t, Y^{t-1})$, for $t=1, \ldots, n$. \\
The next theorem further shows that for each $t$,  the dependence of $X_t$ on $Y^{t-1}$ is expressed in terms of  ${\bf E}\Big\{\hat{S}_t\Big|Y^{t-1}\Big\}$ for $t=1,\ldots, n$,  and this dependence gives rise to an equivalent sequential characterization of the Cover and Pombra $n-$FTFI capacity, $C_n^{fb}(\kappa)$. \\

\begin{theorem} Equivalent characterization of  $n-$FTFI Capacity ${C}_n^{fb}(\kappa)$ for PO-SS Noise realizations   \\ 
\label{thm_SS}
Consider the time-varying AGN channel defined by (\ref{g_cp_1}), driven by a noise with the PO-SS realization  of Definition~\ref{def_nr_2}, and the code of Definition~\ref{def_code}.  Consider also the generalized Kalman-filter of Lemma~\ref{lemma_POSS}. \\
Define the conditional covariance and conditional mean  of  $\hat{S}_t$ given $Y^{t-1}$, by 
\begin{align}
K_{t} \tri & cov\Big(\hat{S}_t,\hat{S}_t\Big|Y^{t-1}\Big)= {\bf E}\Big\{\Big(\hat{S}_t- \widehat{\hat{S}}_{t}\Big)\Big(\hat{S}_t-\widehat{\hat{S}}_{t} \Big)^T\Big\}, \hso \widehat{\hat{S}}_{t} \tri {\bf E}\Big\{\hat{S}_t\Big|Y^{t-1}\Big\},  \hso  t=2, \ldots, n, \label{ric_ga}     \\
\widehat{\hat{S}}_1 \tri &\mu_{S_1}, \hso K_1 \tri 0 \label{ric_gb}.
\end{align}
Then the following hold.  \\
(a) An equivalent characterization of the $n-$FTFI capacity ${C}_n^{fb}(\kappa)$,   defined  by (\ref{cp_6})-(\ref{cp_12}),  is 
\begin{align}
{C}_{n}^{fb}(\kappa)=  \sup_{{\cal P}_{[0,n]}^{\hat{S}}(\kappa)}\sum_{t=1}^n H(Y_t|Y^{t-1})-H(V^n) \label{ftfic_is}
\end{align}
where  $(X^n,Y^n)$ is jointly Gaussian, and 
\begin{align}
&H(V^n) \hso \mbox{is the entropy of $V^n$ given in Lemma~\ref{lemma_POSS}, i.e., (\ref{entr_noise})}, \label{cp_15_alt_n_1} \\
&\hat{I}^n \hso \mbox{is the innovations process of $V^n$ given in Lemma~\ref{lemma_POSS}}, \\
&Y_t= X_t +V_t, \hso t=1, \ldots, n, \label{cp_15_alt_n}\\
&V_t=C_t{\hat{S}}_{t} +\hat{I}_t, \label{cp_16_alt_n}\\
& {\bf P}_t(dy_t|y^{t-1})=\int {\bf P}_t(dy |x_t,\hat{s}_{t}){\bf P}_t(dx_t|\hat{s}_t,y^{t-1}) {\bf P}_t(d\hat{s}_t|y^{t-1}), \hso t=2, \ldots, n, \label{PO_tr_1} \\
&{\bf P}_1(dy_1)=\int {\bf P}_1(dy |x_t, \hat{s}_1){\bf P}_1(dx_1|\hat{s}_1) {\bf P}_1(d \hat{s}_1),  \label{PO_tr_2}\\
&{\bf P}_t(dy_t|y^{t-1}) \in N(\mu_{Y_t|Y^{t-1}}, K_{Y_t|Y^{t-1}}),\\
&\mu_{Y_t|Y_{t-1}} \hso \mbox{is linear in $Y^{t-1} \hso$ and $\hso K_{Y_t|Y^{t-1}}$ is nonrandom}, \label{PO_tr_2_a} \\
&{\bf P}_t(dx_t|\hat{s}_{t}, y^{t-1}) \in N(\mu_{X_t|\hat{S}_{t}, Y^{t-1}}, K_{X_t|\hat{S}_{t}, Y^{t-1}}),  \label{PO_tr_2_b}   \\
&\mu_{X_t|\hat{S}_{t}, Y^{t-1}} \hso \mbox{is linear in $(\hat{S}_{t}, Y^{t-1}) \hso$ and $\hso K_{X_t|\hat{S}_{t}, Y^{t-1}}$ is nonrandom}, \label{PO_tr_2_c}  \\
&{\cal P}_{[0,n]}^{\hat{S}}(\kappa) \tri \Big\{{\bf P}_t(dx_t|\hat{s}_{t}, y^{t-1}), t=1,\ldots,n: \frac{1}{n} {\bf E}\Big( \sum_{t=1}^n \big(X_t\big)^2\Big) \leq \kappa    \Big\}. \label{PO_tr_3} 
\end{align}
(b) The optimal jointly Gaussian process  $(X^n,Y^n)$ of part (a)  is  represented, as a function of a sufficient statistic,  by   
\begin{align}
&X_t=  \Lambda_{t}\Big(\hat{S}_{t}- \widehat{\hat{S}}_{t}\Big) + Z_t,\hso t=1, \ldots, n, \label{cp_13_alt_n} \\
&  Z_t \in N(0, K_{Z_t}) \hso \mbox{independent of} \hso (X^{t-1},V^{t-1}, \hat{S}^t, \widehat{\hat{S}^t}, \hat{I}^t, Y^{t-1}), \hso t=1, \ldots, n, \label{cp_8_al_n} \\
&  \hat{I}_t\in N(0, K_{\hat{I}_t})  \hso \mbox{independent of} \hso (X^{t-1}, V^{t-1},\hat{S}^t, Y^{t-1}, \widehat{\hat{S}^t}), \hso t=1, \ldots, n, \label{cp_8_al_nnn}\\
&Y_t= \Lambda_{t}\Big(\hat{S}_{t}- \widehat{\hat{S}}_{t}\Big) + Z_t +V_t, \hso t=1, \ldots, n, \label{cp_15_alt_n}\\
&\hso \: = \Lambda_{t}\Big(\hat{S}_{t}- \widehat{\hat{S}}_{t}\Big) +C_t{\hat{S}}_{t} +\hat{I}_t  + Z_t,  \label{cp_15_alt_n_1}  \\
&\frac{1}{n}  {\bf E} \Big\{\sum_{t=1}^{n} (X_t)^2\Big\}= \frac{1}{n} \sum_{t=1}^{n}\Big(\Lambda_t K_t \Lambda_t^T +K_{Z_t}\Big)  . \label{cp_9_al_n}
\end{align} 
where $\Lambda_t$ is nonrandom.   \\
The conditional mean and covariance, $\widehat{\hat{S}}_t$ and $K_t$, are given by  generalized Kalman-filter equations, as follows.\\
(i)   $\widehat{\hat{S}}_t$ satisfies the Kalman-filter recursion
\begin{align}
&\widehat{\hat{S}}_{t+1}=A_{t} \widehat{\hat{S}}_{t}+F_{t}(\Sigma_{t}, K_t)  I_t, \hso \widehat{\hat{S}}_{1}=\mu_{S_1},  \label{kf_m_1}  \\
&F_{t}(\Sigma_{t}, K_t) \tri \Big(A_{t}  K_{t}\big(\Lambda_t + C_t \big)^T+   M_{t}(\Sigma_{t}) K_{\hat{I}_{t}} \Big)\Big\{K_{\hat{I}_t}+   K_{Z_t} + \big(\Lambda_t + C_t \big) K_{t} \big(\Lambda_t + C_t \big)^T \Big\}^{-1}      \\
&I_t \tri Y_t -{\bf E} \Big\{Y_t\Big|Y^{t-1}\Big\}= Y_t-C_t\widehat{\hat{S}}_{t}=  \Big(\Lambda_t+C_t\Big) \Big(\hat{S}_t-\widehat{\hat{S}}_{t}\Big)+ \hat{I}_t+ Z_t, \hso t=1, \ldots, n, \label{kf_m_2} \\
&I_t \in  N(0, K_{I_t}), \hso t=1, \ldots, n \hso \mbox{is an orthogonal innovations process, i.e., $I_t$ is independent of}\nonumber \\
&\hst \hst \mbox{ $I_s$, for all $t \neq s$, and ${I}_t$ is independent of  $V^{t-1}$},  \label{inn_po_2_nn}\\
&K_{Y_t|Y^{t-1}}=K_{I_t} \tri cov\big(I_t,I_t\big)=  \Big(\Lambda_t +C_t\Big)K_t \Big(\Lambda_t +C_t\Big)^T + K_{\hat{I}_t} + K_{Z_t}, \label{inno_PO}  \\
&K_{\hat{I}_t}  \hso  \mbox{given by  (\ref{cov_in_noise})}. \label{kf_m_3} 
\end{align}
(ii) The error $\widehat{E}_t \tri \hat{S}_t- \widehat{\hat{S}}_{t}$ satisfies the recursion 
\begin{align}
\widehat{E}_{t+1} = & F_t^{CL}(\Sigma_t, K_t)\widehat{E}_t +\Big(M_t(\Sigma_t)- F_t(\Sigma_t, K_t)\Big) \hat{I}_t-  F_t(\Sigma_t. K_t) Z_t, \hso \widehat{E}_1= \hat{S}_1-\widehat{\hat{S}}_1=0, \hso t=1,\ldots, n,   \label{kf_m_3a}   \\
F_t^{CL}(\Sigma_t,& K_t)\tri A_t-  F_t(\Sigma_t, K_t)\Big(\Lambda_t+C_t\Big).
\end{align}
(iii)  $K_t={\bf E}\big\{\widehat{E}_t \widehat{E}_t^T\big\}$ satisfies the generalized DRE 
\begin{align}
&K_{t+1}= A_t K_{t}A_t^T  + M_t(\Sigma_{t})K_{\hat{I}_t}\big(M_t(\Sigma_{t})\big)^T -\Big(A_t  K_{t}\big(\Lambda_t + C_t \big)^T+ M_t(\Sigma_t)K_{\hat{I}_t}   \Big) \Big( K_{\hat{I}_t}+ K_{Z_t} \nonumber \\
&+ \big(\Lambda_t + C_t \big) K_{t} \big(\Lambda_t + C_t \big)^T     \Big)^{-1} \Big(  A_t  K_{t}\big(\Lambda_t + C_t \big)^T+ M_t(\Sigma_t)K_{\hat{I}_t}      \Big)^T, \hso K_t \succeq 0, \hso t=1, \ldots, n, \hso K_1=0. \label{kf_m_4_a}  
\end{align}
(c) An equivalent characterization of the $n-$FTFI capacity ${C}_n^{fb}(\kappa)$, defined  by (\ref{cp_11}), (\ref{cp_12}), using the sufficient statistics of part (b),   is  
\begin{align}
 {C}_n^{fb}(\kappa)
= & \sup_{  \big(\Lambda_{t}, K_{Z_t}\big), t=1, \ldots, n: \hso \frac{1}{n}  {\bf E}\big\{\sum_{t=1}^n \big(X_t\big)^2\big\}\leq \kappa } \frac{1}{2} 
\sum_{t=1}^n \log \frac{  K_{Y_t|Y^{t-1}}   }{K_{V_t|V^{t-1}}} \label{cp_12_alt_1}\\
= & \sup_{  \big(\Lambda_{t}, K_{Z_t}\big), t=1, \ldots, n: \hso \frac{1}{n}  {\bf E}\big\{\sum_{t=1}^n \big(X_t\big)^2\big\}\leq \kappa } \frac{1}{2} 
\sum_{t=1}^n \log \frac{  K_{I_t}   }{K_{\hat{I}_t}} \\
=& \sup_{  \big(\Lambda_{t}, K_{Z_t}\big), t=1, \ldots, n: \hso  \frac{1}{n}\sum_{t=1}^n\big( \Lambda_t K_t \Lambda_t^T+K_{Z_t} \big)    \leq \kappa } \frac{1}{2} 
\sum_{t=1}^n \log\Big( \frac{ \Big(\Lambda_t +C_t\Big)K_t \Big(\Lambda_t +C_t\Big)^T + K_{\hat{I}_t} + K_{Z_t}     }{K_{\hat{I}_t}}\Big). \label{cp_12_alt_1_new}
  \end{align}
\end{theorem}
\begin{proof} 
See Appendix~\ref{sect:app_thm_SS}. 
\end{proof}

\ \

\begin{remark} On the characterization of $n-$FTFI capacity  of Theorem~\ref{thm_SS}\\
The characterization of $n-$FTFI capacity  ${C}_n^{fb}(\kappa)$ given by (\ref{cp_12_alt_1_new}), involves the generalized matrix DRE $K_t$ which is also a functional of the generalized matrix DRE $\Sigma_t$ of the error covariance of the state $S^n$ from the noise output $V^n$.  This feature is not part of the analysis in  \cite{kim2010} and recent literature \cite{kim2010,liu-han2019,gattami2019,ihara2019,li-elia2019}. 
\end{remark}

From Theorem~\ref{thm_SS} follows directly, as degenerate case the next corollary.

\begin{corollary}Equivalent characterization of  $n-$FTFI Capacity ${C}_n^{fb}(\kappa,s)$ for PO-SS Noise realizations  
\label{thm_SS_in}
Consider the time-varying AGN channel defined by (\ref{g_cp_1}), driven by a noise with the PO-SS realization  of Definition~\ref{def_nr_2}, and the code of Definition~\ref{def_rem:cp_1}, with initial state $S_1=S_1^s=s$ fixed, and replace (\ref{ric_ga}), (\ref{ric_gb}) by
\begin{align}
K_{t}=K_t^s \tri & cov\Big(\hat{S}_t^s,\hat{S}_t^s\Big|Y^{t-1}, S_1=s\Big)= {\bf E}\Big\{\Big(\hat{S}_t^s- \widehat{\hat{S}_t^s}\Big)\Big(\hat{S}_t^s-\widehat{\hat{S}_t^s} \Big)^T\Big\},  \label{ric_ga_in}     \\
 \widehat{\hat{S}}_{t}=& \widehat{\hat{S}_t^s} \tri {\bf E}\Big\{\hat{S}_t^s\Big|Y^{t-1}, S_1=s\Big\},  \hso  t=2, \ldots, n, \hso  \widehat{\hat{S}_1} = \widehat{\hat{S}_1^s} \tri s, \hso K_1=K_1^s = 0 \label{ric_gb_in}.
\end{align}
Then the characterization of $n-$FTFI capacity, (\ref{ftfic_is_g_def_in_IS}), is 
\begin{align}
&{C}_{n}^{fb}(\kappa,s)=  \sup_{{\cal P}_{[0,n]}^{\widehat{\hat{S}^s}}(\kappa)}\sum_{t=1}^n H(Y_t|Y^{t-1},s)-H(V^n|s),  \label{ftfic_is_in}\\
&{\cal P}_{[0,n]}^{\widehat{\hat{S}^s}}(\kappa) \tri \Big\{{\bf P}_t(dx_t|\widehat{\hat{s}_{t}^s}, y^{t-1},s), t=1,\ldots,n: \frac{1}{n} {\bf E}\Big( \sum_{t=1}^n \big(X_t\big)^2\Big|S_1^s=s\Big) \leq \kappa    \Big\} \label{PO_tr_3_in} 
\end{align}
where $H(V^n|s)$ is given by Corollary~\ref{cor_POSS},  and the statements of Theorem~\ref{thm_SS} hold  with the above changes, i.e., (\ref{ric_ga_in}), (\ref{ric_gb_in}),  and all conditional entropies, distributions, expectations, etc, defined for fixed $S_1=S_1^s=s$,   
\end{corollary}
\begin{proof}
It is easily verified  from the derivation of Theorem~\ref{thm_SS},  by  fixing $S_1=S_1^s=s$.
\end{proof}

\begin{remark}  On the characterization of $n-$FTFI capacity  of Corollary~\ref{thm_SS_in}\\
The characterization of $n-$FTFI capacity  ${C}_n^{fb}(\kappa,s)$ given in  Corollary~\ref{thm_SS_in} (similar to  Theorem~\ref{thm_SS}) involves two generalized matrix DREs,   because it does not assume Conditions 1 and 2 hold.  This distinction is not part of the analysis in    \cite{kim2010,liu-han2019,gattami2019,ihara2019,li-elia2019}. 
\end{remark}

\subsection{Application Examples}
\label{appl_ex}
In this section we apply  Theorem~\ref{thm_SS} to specific examples. 

First, we consider the application example of the AGN channel driven by  the PO-SS$(a_t,c_t, b_t^1, b_t^2, d_t^1, d_t^2)$ noise.

\begin{corollary} 
The $n-$FTFI capacity $C_n^{fb}(\kappa)$  of the AGN channel driven  by the PO-SS$(a_t,c_t, b_t^1, b_t^2, d_t^1, d_t^2)$ noise is obtained from  Lemma~\ref{lemma_POSS} and Theorem~\ref{thm_SS}, by using  (\ref{poss_ex_1}). 
\end{corollary}
\begin{proof}
This is easily verified, as in Corollary~\ref{entr_poss}.
\end{proof}

In the next corollary we apply Theorem~\ref{thm_SS} to the  stable and unstable ARMA$(a,c)$ noise, to obtain the  characterization of $n-$FTFI capacity $C_n^{fb}(\kappa)$ and $C_n^{fb}(\kappa,s)$. It is then obvious that for the stable  ARMA$(a,c), a \in [-1,1], c\in (-1,1)$ noise, the characterization of $C_n^{fb}(\kappa)$   involves two generalized DREs, contrary to the analysis in   \cite{kim2010,liu-han2019,gattami2019,ihara2019,li-elia2019}, for the same noise model. \\

\begin{corollary} Characterization of  $n-$FTFI Capacity ${C}_n^{fb}(\kappa)$ for the ARMA$(a,c), a\in (-\infty, \infty), c\in (-\infty,\infty)$   \\ 
\label{cor_ex_1}
Consider the time-varying AGN channel defined by (\ref{g_cp_1}) and the code of Definition~\ref{def_code}.\\
(a) For the nonstationary ARMA$(a,c), a\in (-\infty,\infty), c\in (-\infty,\infty)$ noise of  Example~\ref{ex_1_1_n}.(a),  the characterization of the $n-$FTFI capacity, $C_n^{fb}(\kappa)$ is  
\begin{align}
 &{C}_n^{fb}(\kappa)
= \sup_{  \big(\Lambda_{t}, K_{Z_t}\big), t=1, \ldots, n: \hso  \frac{1}{n}    \sum_{t=1}^n\Big( \big(\Lambda_t\big)^2 K_t +K_{Z_t} \Big)    \leq \kappa } \frac{1}{2} 
\sum_{t=1}^n \log\Big( \frac{ \Big(\Lambda_t +c-a\Big)^2K_t  + K_{\hat{I}_t} + K_{Z_t}     }{K_{\hat{I}_t}}\Big) \label{cp_12_alt_1_new_new}
\end{align}
subject to the constraints
 \begin{align}
K_{t+1}= &\big(c\big)^2 K_{t}  +\big(M_t(\Sigma_{t})\big)^2K_{\hat{I}_t} -\Big(c  K_{t}\big(\Lambda_t + c-a \big)+ M_t(\Sigma_t)K_{\hat{I}_t}   \Big)^2 \nonumber \\
&\hst . \Big(  K_{\hat{I}_t}+ K_{Z_t} + \big(\Lambda_t + c-a \big)^2 K_{t}      \Big)^{-1}, \hso  K_1=0, \hso  t=1, \ldots, n, \label{kf_m_4}\\
K_{Z_t} \geq & 0, \hso K_t \geq0, \hso c\neq a, \hso K_W>0, \hso t=1, \ldots, n  
\end{align} 
and where 
\begin{align}
&M_{t}(\Sigma_t) \tri \Big( c \Sigma_{t} \big(c-a\big)+K_{W}\Big)\Big(K_{W}+ \big(c-a\big)^2 \Sigma_{t} \Big)^{-1}, \label{kf_m_444_a} \\
&K_{\hat{I}_t}= \big(c-a\big)^2 \Sigma_t +K_{W},  \hso t=1, \ldots, n,\label{kf_m_444_b} \\
&\Sigma_{t+1}= \big(c\big)^2 \Sigma_{t}  + K_{W} -\Big(c  \Sigma_{t}\big(c-a\big)+K_{W}  \Big)^2 \Big(K_{W} +\big(c-a\big)^2  \Sigma_{t} \Big)^{-1}, \hso t=1, \ldots, n,\label{kf_m_444} \\
&\Sigma_{1}=K_{S_1}=\frac{\big(c_0\big)^2 K_{S_0} +\big(a_0\big)^2 K_{W_0}}{\Big(c_0-a_0\Big)^2}.\label{kf_m_44}
\end{align}
The optimal  jointly Gaussian process  $(X^n,Y^n)$  is obtained from  Theorem~\ref{thm_SS}.(b), by invoking, 
\begin{align}
A_t \longmapsto c, \hso  C_t \longmapsto  c-a, \hso B_t \longmapsto 1, \hso  N_t \longmapsto 1, \hso t=1,2, \ldots, n. \label{trans} 
\end{align}
Special Case. If  $\Sigma_1=0$ or the initial state is fixed,  $S_1=S_1^s=s$, then  
\bea
\Sigma_t=\Sigma_t^s=0,\hso  K_{\hat{I}_t}=K_W, \hso  M_t(\Sigma_t)=M_t(\Sigma_t^s)=1, \hso  t=1, 2, \ldots
\label{deg_arma}
\eea
 and  $C_n^{fb}(\kappa)$ reduces to   
\begin{align}
 &{C}_n^{fb}(\kappa)=C_n^{fb}(\kappa,s)
= \sup_{  \big(\Lambda_{t}, K_{Z_t}\big), t=1, \ldots, n: \hso  \frac{1}{n}    \sum_{t=1}^n\Big( \big(\Lambda_t\big)^2 K_t^s +K_{Z_t} \Big)    \leq \kappa } \frac{1}{2} 
\sum_{t=1}^n \log\Big( \frac{ \Big(\Lambda_t +c-a\Big)^2K_t^s  + K_{W} + K_{Z_t}     }{K_{W}}\Big) \label{cp_12_alt_1_new_new_new}
\end{align}
subject to the constraints
 \begin{align}
K_{t+1}^s= &\big(c\big)^2 K_{t}^s  +K_{W} -\Big(c  K_{t}^s\big(\Lambda_t + c-a \big)+ K_{W}   \Big)^2 \nonumber \\
&. \Big( K_{Z_t} + \big(\Lambda_t + c-a \big)^2 K_{t}^s  +K_W    \Big)^{-1}, \hso  K_1^s=0, \hso 
 \hso K_t^s \geq0, \hso K_{Z_t} \geq  0,\hso t=1, \ldots, n.  \label{kf_m_4_a_n}
\end{align}  
(This special case is precisely the application example analyzed in \cite{yang-kavcic-tatikonda2007}).

(b) For the nonstationary AR$(c), c\in (-\infty,\infty)$ noise  of  Example~\ref{ex_1_1_n}.(c), the  characterization of the $n-$FTFI capacity $C_n^{fb}(\kappa)$ is obtained from   part (a) by setting  $a=0$, i.e., 
\begin{align}
 &{C}_n^{fb}(\kappa)
= \sup_{  \big(\Lambda_{t}, K_{Z_t}\big), t=1, \ldots, n: \hso  \frac{1}{n}    \sum_{t=1}^n\Big( \big(\Lambda_t\big)^2 K_t +K_{Z_t} \Big)    \leq \kappa } \frac{1}{2} 
\sum_{t=1}^n \log\Big( \frac{ \Big(\Lambda_t +c\Big)^2K_t  + \big(c\big)^2 \Sigma_t +K_{W}+ K_{Z_t}     }{\big(c\big)^2 \Sigma_t +K_{W}}\Big) \label{cp_12_alt_1_new_new_ar1}
\end{align}
subject to the constraints  $K_t, \Sigma_t$ are  the nonnegative solutions of the generalized RDEs:
\begin{align}
K_{t+1}=& \big(c\big)^2 K_{t}  +  \big(c\big)^2 \Sigma_{t}+K_{W} -\Big(c K_{t}\big(\Lambda_t + c\big)+  \big(c\big)^2 \Sigma_{t}+K_{W} \Big)^2 \nonumber \\
 &. \Big(  \big(c\big)^2 \Sigma_{t}+K_{W}+ K_{Z_t} + \big(\Lambda_t + c \big)^2 K_{t}      \Big)^{-1}, \hso  K_1=0, \hso  t=1, \ldots, n, \label{kf_m_4_ar1}\\
\Sigma_{t+1}=& \big(c\big)^2 \Sigma_{t}  + K_{W} -\Big( \big(c\big)^2  \Sigma_{t}+K_{W}  \Big)^2 \Big(K_{W} +\big(c\big)^2  \Sigma_{t} \Big)^{-1}, \hso \Sigma_1=K_{S_1}=K_{S_0}\geq 0, \hso t=1, \ldots, n. \label{grde_ar1}
\end{align}
(c) For the nonstationary AR$(c), c\in (-\infty,\infty)$ noise  of  Example~\ref{ex_1_1_n}.(c), with $\Sigma_1=0$ or  a fixed initial state  $S_1=S_1^s=s$, then (\ref{deg_arma}) holds, i.e.,     $\Sigma_t=\Sigma_t^s=0, K_{\hat{I}_t}=K_W,  M_t(\Sigma_t)=M_t(\Sigma_t^s)=1,  t=1, 2, \ldots$, and  ${C}_n^{fb}(\kappa)$ reduces to  
\begin{align}
{C}_n^{fb}(\kappa)=C_n^{fb}(\kappa,s)
= \sup_{  \big(\Lambda_{t}, K_{Z_t}\big), t=1, \ldots, n: \hso  \frac{1}{n}    \sum_{t=1}^n\Big( \big(\Lambda_t\big)^2 K_t^s +K_{Z_t} \Big)    \leq \kappa } \frac{1}{2} 
\sum_{t=1}^n \log\Big( \frac{ \Big(\Lambda_t +c\Big)^2K_t^s  + K_{W} + K_{Z_t}     }{K_{W}}\Big) \label{ar_1deg}
\end{align}
subject to the constraint
\begin{align}
&K_{t+1}^s= \big(c\big)^2 K_{t}^s  +K_{W} -\Big(c  K_{t}^s\big(\Lambda_t + c \big)+ K_{W}   \Big)^2 \Big(  K_{W}+ K_{Z_t} + \big(\Lambda_t + c \big)^2 K_{t}^s      \Big)^{-1}, \;  K_1^s=0, \;  t=1, \ldots, n. \label{ar_1deg_n}
\end{align}
\end{corollary}
\begin{proof} (a) The first part follows directly from 
Theorem~\ref{thm_SS}, by using (\ref{trans}).
The last part is obtained as follows. If  $\Sigma_1=0$ or $S_1=S_1^s=s$ is fixed,  then 
by   
(\ref{kf_m_444}) it follows, $\Sigma_t=\Sigma_t^s=0, \forall t=2,\ldots$, and by (\ref{kf_m_444_a}), (\ref{kf_m_444_b})
it follows,  $M_t(\Sigma_t)= M_t(\Sigma_t^s)=1, K_{\hat{I}_t}=(c-a)^2 \Sigma_t^s+K_W=K_W, \forall t =1, 2,$. Substituting into (\ref{cp_12_alt_1_new_new}), (\ref{kf_m_4}) we obtain (\ref{cp_12_alt_1_new_new_new}), (\ref{kf_m_4_a}). 
 (b) From part (a), letting $a=0$, then 
\begin{align}
&M_{t}(\Sigma_t)= \Big( \big(c\big)^2  \Sigma_{t}+K_{W}\Big)\Big(K_{W}+ \big(c\big)^2 \Sigma_{t} \Big)^{-1},\hso K_{\hat{I}_t}= \big(c\big)^2 \Sigma_t +K_{W}, \hso t=1, \ldots, n. \label{spe_c_1}
\end{align}
By substitution into the equations of part (a) we obtain  (\ref{kf_m_4_ar1}), (\ref{grde_ar1}). (c) This is a special case of parts (a), (b).
\end{proof}

\  \

\begin{remark} By Corollary~\ref{cor_ex_1}.(a) it is obvious that, if $\Sigma_1=0$, i.e., $K_{S_0}=K_{W_0}=0$, which means $S_1=S_1^s=s$ is fixed, and hence $(V_0, W_0) =(v_0, w_0)$ is fixed (and known to the encoder and the decoder), see (\ref{ex_1_4}), then $\Sigma_1=\Sigma_1^s=0$, and $C_n^{fb}(\kappa)=C_n^{fb}(\kappa,s)$, which depends on the initial state $S_1=S_1^s=s$. To ensure for large enough $n$ the rate $\frac{1}{n}C_n(\kappa, s)$ is independent of $s$, it is necessary to identify conditions for convergence of solutions $K_t^s, t=1,2, \ldots$ of  generalized DRE  (\ref{kf_m_4_a}) to a unique limit,  $\lim_{n \longrightarrow \infty}K_n^s= K^\infty\geq 0$, that does not depend on the initial data $K_1^s=0$. We address this  problem in Section~\ref{sect:as_an}. We should emphasize that the asymptotic limit of the of Corollary~\ref{cor_ex_1}.(c), i.e., of the AR$(c), c \in (-\infty,\infty)$ is fully analyzed in \cite{charalambous-kourtellaris-loykaIEEEITC2019}. 
\end{remark}

\subsection{Case II) Formulation:   A Degenerate of Case I) Formulation}
\label{sect:case_II}
%
%
%
%
%
%
%

Theorem~\ref{thm_SS} gives the $n-$FTFI capacity for Case I)  formulation. However, since Case II) formulation  is a special case of Case I)  formulation, we expect that  from   Theorem~\ref{thm_SS} we can recover the characterization of the $n-$FTFI capacity for Case II)  formulation, i.e.,  when the code is $(s, 2^{n R},n)$,   $n=1,2, \ldots$, and Conditions 1 and 2 of  Section~\ref{sect:motivation} hold.  We show this in  the next corollary.  \\

\begin{corollary} The degenerate  $n-$FTFI Capacity ${C}_n^{fb}(\kappa)$ of Theorem~\ref{thm_SS} for Case II) formulation   \\ 
\label{cor_issr}
Consider the time-varying AGN channel defined by (\ref{g_cp_1}), driven by a noise with  PO-SS realization  of Definition~\ref{def_nr_2}, and  suppose the following hold.

1) The code is  $(s, 2^{n R},n)$,   $n=1,2, \ldots$, and 

2) Conditions 1 and 2 of  Section~\ref{sect:motivation} hold.

Then the following hold.\\
(a) Corollary~\ref{cor_POSS} holds, i.e.,  all statements of Lemma~\ref{lemma_POSS} hold with  $(\Sigma_t, \hat{S}_t)$ replaced by  $(\Sigma_t^{s}, \hat{S}_t^{s})$  as defined by (\ref{cond_GK_1_n}), (\ref{cond_GK_2_n}). In particular, $(\Sigma_t^{s}, \hat{S}_t^{s})=(0, S_t^s)$ for $t=1,2, \ldots$,   and $H(V^n)=H(V^n|s)$ is given by  (\ref{cond_ent}). \\
(b) All statements of  Theorem~\ref{thm_SS} hold with  $(\Sigma_t, \hat{S}_t)$ replaced by  $(\Sigma_t^{s}, \hat{S}_t^{s})$, as in part (a),   and    $(K_t, \widehat{\hat{S}}_{t})$ defined by  (\ref{ric_ga}), (\ref{ric_gb})  reduce to     
\bea
K_t=K_t^s=cov\Big({S}_t^s,{S}_t^s\Big|Y^{t-1}, S_1^s=s\Big),\;   \widehat{\hat{S}}_{t}=\widehat{S}_t^s={\bf E}\Big\{{S}_t^s\Big|Y^{t-1},S_1^s=s\Big\},\; K_1^s=0, \widehat{S}_1^s=s, \;  t=2, \ldots, n
\eea
In particular,  the optimal input process  $X^n$ of Theorem~\ref{thm_SS}.(c) degenerates to   
\begin{align}
&X_t=  \Lambda_{t}\Big(S_{t}^s- \widehat{S}_{t}^s\Big) + Z_t,\hso X_1=Z_t, \hso  t=2, \ldots, n. \label{cp_13_alt_n_deg}
\end{align}
(c) The characterization of $n-$FTFI capacity,  $C_n^{fb}(\kappa)$ of Theorem~\ref{thm_SS} degenerates to $C_n^{fb,S}(\kappa,s)$ defined by
\begin{align}
{C}_n^{fb}(\kappa)=& {C}_n^{fb,S}(\kappa,s)
  \tri \sup_{  \big(\Lambda_{t}, K_{Z_t}\big), t=1, \ldots, n: \hso \frac{1}{n}  {\bf E}_{s}\big\{\sum_{t=1}^n \big(X_t\big)^2\big\}\leq \kappa }  \sum_{t=1}^n \log \frac{K_{Y_t|Y^{t-1},s}}{K_{V_t|V^{t-1},s}} \label{cp_11_alt_1_deg}\\
=&  \sup_{  \big(\Lambda_{t}, K_{Z_t}\big), t=1, \ldots, n: \hso \frac{1}{n}\sum_{t=1}^n \big(\Lambda_t K_t^s \Lambda_t^T+K_{Z_t}\big)\leq \kappa }\frac{1}{2} 
\sum_{t=1}^n \log\Big( \frac{ \Big(\Lambda_t +C_t\Big)K_t^s \Big(\Lambda_t +C_t\Big)^T + N_t K_{W_t}N_t^T + K_{Z_t}     }{N_t K_{W_t}N_t^T}\Big).
 \label{cp_12_alt_1_deg}
  \end{align}
$K_t=K_t^s={\bf E}_s\big\{E_t^s \big(E_t^s\big)^T\big\}$ satisfies the generalized DRE 
\begin{align}
K_{t+1}^s= &A_t K_{t}A_t^T  + B_tK_{W_t}B_t^T -\Big( B_tK_{W_t}N_t^T  + A_t  K_{t}^s\big(\Lambda_t + C_t \big)^T\Big)  \Big\{N_t K_{W_t} N_t^T+K_{Z_t}\nonumber \\
&+ \big(\Lambda_t + C_t \big) K_{t}^s\big(\Lambda_t + C_t \big) ^T\Big\}^{-1}\Big( B_t K_{W_t}N_t^T + A_t  K_{t}^s\big(\Lambda_t + C_t \big)\Big)^T,\hso   K_t^s\succeq 0, \hso K_1^s=0, \hso t=1, \ldots, n.   \label{rde_ykt_deg}
\end{align}
and the statements of parts (a), (b) hold.
\end{corollary}
\begin{proof} (a) The statements about Lemma~\ref{lemma_POSS} follow  from   Remark~\ref{rem_ss_is}.  (b) The statements about Theorem~\ref{thm_SS} are easily verified by replacing all conditional expectations, distributions, etc, for a fixed   initial state $S_1=S_1^s=s$, and using part (a), i.e.,  $(\Sigma_t^{s}, \hat{S}_t^{s})=(0, S_t^s)$, $t=1,2, \ldots$. 
(c)  Follows from parts (a), (b).
\end{proof}

%
%
%
%
%
%
%
%
%
%
%
%
%
%
%
%

\subsection{Comments  on Past Literature}
\label{sect:p_l}
It is  easily verified that Yang, Kavcic and Tatikonda  \cite{yang-kavcic-tatikonda2007} analyzed    $C_n^{fb}(\kappa,s)$ defined by  (\ref{ftfic_is_g_is_def}), under Case II) formulation, i.e., Conditions 1 and 2 of  Section~\ref{sect:motivation} hold, as  discussed in the next remark.\\

\begin{remark} Prior literature on the time-invariant stationary noise  of PSD (\ref{PSD_G}) \\
\label{rem_ykt}
Yang, Kavcic and Tatikonda \cite{yang-kavcic-tatikonda2007} analyzed the AGN channel driven by a stationary noise with PSD defined by (\ref{PSD_G}) (see    \cite[Theorem~1]{yang-kavcic-tatikonda2007}). The special case of  (\ref{ex_2_1_in})  is found in \cite[Section VI.B, Theorem~7]{yang-kavcic-tatikonda2007}.  \\
The analysis in \cite{yang-kavcic-tatikonda2007} presupposed   the following formulation:\\
(i) the   code is  $(s, 2^{n R},n)$,   $n=1,2, \ldots$, where $S_1=S_1^s=s$ is the initial state of the noise, known to the encoder and the decoder, as discussed in Definition~\ref{def_rem:cp_1}, \\ 
(ii) Conditions 1 and 2  of Section~\ref{sect:motivation}, hold, and \\
(iii) the  $n-$FTFI capacity formula is $C_n^{fb}(\kappa,s)$ defined by (\ref{ftfic_is_g_is_def}).\\
We  emphasize that in \cite[Section~II.C]{yang-kavcic-tatikonda2007} a specific  realization of the PSD is considered to ensure Conditions  1 and 2 hold, i.e., the analysis in  \cite{yang-kavcic-tatikonda2007} presupposed a stationary noise and Case II) formulation. \\  
\end{remark}

Now, we  ask: Given the PSD of the noise defined by (\ref{PSD_G}), and the double-sided realization \cite[eqn(58)]{kim2010}, i.e., the  analog of time-invariant version of the PO-SS realization of Definition~\ref{def_nr_2}, or its analogous one-sided realization, what are the necessary conditions for the feedback capacity of \cite[Theorem~6.1]{kim2010} to be valid?\\
The answer to this  question   is: Conditions 1 and 2  of Section~\ref{sect:motivation} are necessary conditions. We show this in the next proposition.  \\

\begin{proposition} Conditions for validity of  the feedback capacity characterization of \cite[Theorem~6.1]{kim2010}\\
\label{pro_1}
Consider the AGN channel (\ref{g_cp_1}) driven by a stationary noise with PSD defined by (\ref{PSD_G}) with the double-sided or one-sided realization \cite[eqn(58)]{kim2010}, (i.e., analog of time invariant of Definition~\ref{def_nr_2}). \\
   Then a  necessary condition for \cite[Theorem~6.1]{kim2010} to hold is 
\begin{align}
{\bf P}_{X_t|X^{t-1}, Y_{-\infty}^{t-1}}={\bf P}_{X_t|S^{t}, Y_{-\infty}^{t-1}}, \hst t =1, \ldots, \label{eql_kim}
\end{align}
Further, Conditions 1 and 2  of Section~\ref{sect:motivation} are necessary and sufficient  for equality  (\ref{eql_kim}) to hold.  
\end{proposition}
\begin{proof} 
 See Section~\ref{sect:app_pro_1}.
\end{proof} 


The next remark is our final observation on prior literature.

\begin{remark} Comparison of Cover and Pombra Characterization and current literature\\ 
From Corollary~\ref{cor_issr} and  Proposition~\ref{pro_1} we have the following.
\\
The characterization of feedback capacity given in \cite[Theorem~6.1, $C_{FB}$]{kim2010} corresponds to Case II) formulation and not to Case I) formulation. 
Further, the optimization problem of  \cite[Theorem~6.1, $C_{FB}$]{kim2010} is precisely the  optimization problem investigated  in  \cite[Section~VI]{yang-kavcic-tatikonda2007}, with the additional restriction that the innovations part of the channel input is taken to be zero in \cite[Theorem~6.1, $C_{FB}$]{kim2010}, i.e.,  see  \cite[Lemma~6.1 and comments above it]{kim2010}. Recent literature \cite{liu-han2019,gattami2019,ihara2019,li-elia2019} should be read with caution, because the results therein, often build on \cite[Theorem~4.1 and Theorem~6.1]{kim2010}.
\end{remark}

\section{Asymptotic Analysis for Case I) Formulation}
\label{sect:as_an}
In this section we address the asymptotic per unit time limit of the $n-$FTFI capacity. Our analysis includes the following.  

1) Fundamental differences of entropy rates  of jointly Gaussian stable versus  unstable noise processes.

2) Necessary and/or sufficient conditions expressed in terms of detectability and stabilizability conditions of generalized DREs\cite{caines1988,kailath-sayed-hassibi}, for   existence of  entropy rates, and asymptotic stationarity of the input process $X^n, n=1,2,\ldots$ (and output process $Y^n, n=1,2, \ldots,$ if the noise is stable). 


This section also reconfirms that, in general,  the asymptotic analysis of the $n-$FTFI capacity of a  feedback code that depends on the initial state of the channel, i.e., $S_1=S_1^s=s$,  is fundamentally different from a code that does does not depend on the initial state. 
 The analysis of the asymptotic per unit time limit of $C_n^{fb}(\kappa,s)$ of  AGN channels driven by AR$(c), c\in (-\infty, \infty)$ noise, i.e., stable and unstable, is found in  \cite{charalambous-kourtellaris-loykaIEEEITC2019}. 
We consider  the following definition of rate, often used for nonfeedback capacity of stationary processes; however, our formulations does not assume stationarity.  

\begin{definition} Per unit time limit of $C_n^{fb,o}(\kappa)$ and $C_n^{fb,o}(\kappa,s)$ \\
\label{def_limit_1}
Consider the AGN channel defined by (\ref{g_cp_1}), driven by the  time-invariant PO-SS realization  of Definition~\ref{def_nr_2}. \\
(a) For the code of Definition~\ref{def_code}, define the  per unit time limit 
  \begin{align}
 C^{fb,o}(\kappa) \tri  \sup_{\lim_{n \longrightarrow \infty} \frac{1}{n} {\bf E}\big\{\sum_{t=1}^n \big(X_t\big)^2 \big\}\leq \kappa} \lim_{n \longrightarrow \infty} \frac{1}{n} \Big\{ H(Y^n)-H(V^n)\Big\} 
  \leq C^{fb}(\kappa) \tri  \lim_{n \longrightarrow \infty} \frac{1}{n} C_n^{fb}(\kappa)
   \label{CP_F_new}
 \end{align}
 where the supremum is taken over all 
 time-invarinat distributions with feedback ${\bf P}_{X_t|X^{t-1}, Y^{t-1}}^o={\bf P}_{X_t|V^{t-1}, Y^{t-1}}^o, t=1,2,  \ldots$, such that the limits exists and the supremum exists and it is finite. 
  \\
(b) For code  of  Definition~\ref{def_rem:cp_1}, i.e.,  $(s, 2^{n R},n)$,   $n=1,2, \ldots$, with initial state  $S_1=S_1^s=s$, $C^{fb,o}(\kappa)$ is replaced by $C^{fb,o}(\kappa,s)$,  defined by (\ref{CP_F_new}),  with differential entropies, conditional expectations, conditional distributions, defined   for fixed $S_1^s=s$. 
\end{definition} 
 
Our definition of rate is consistent with the definition of rates considered in \cite{kim2010,liu-han2019,gattami2019,ihara2019,li-elia2019}, i.e., the interchange of limit and supremum. However, unlike  \cite{kim2010,liu-han2019,gattami2019,ihara2019,li-elia2019} we treat the general time-invariant stable and unstable,  PO-SS noise realization of Definition~\ref{def_nr_2},   not necessarily  stationary or  asymptotically stationary. \\
We should emphasize  that,  in general, and irrespective of  whether the noise is stable or unstable,  the entropy rates that appear in (\ref{CP_F_new}) may not  exist. To  show  existence of the limits   $C^{fb,o}(\kappa)$ and $C^{fb,o}(\kappa,s)$, we identify necessary and/or sufficient conditions, using the characterization of     Theorem~\ref{thm_SS}, when the   channel input strategies are restricted to the time-invariant strategies $\Lambda_t=\Lambda^\infty, K_{Z_t}=K_Z^\infty, t=1,2, \ldots$. Clearly, by (\ref{CP_F_new}), whether the limit as $n\longrightarrow \infty$ exists, and  supremum over channel input distributions exists,  depend on   the convergence properties of the coupled generalized matrix DREs, $\Sigma_n, K_n^0\equiv K_n (\Lambda^\infty, K_Z^\infty,\Sigma)$, as $n \longrightarrow \infty$. 


\subsection{Entropy Rates of Gaussian Processes}
First, we recall the following definition, which  is standard and it is  found in many textbooks. \\

\begin{definition} Entropy rate of continuous-valued random  processes \\
Let $X_t: \Omega \rar {\mathbb R}^{n_z}, n_x\in Z_+$ a random process defined on some probability space $(\Omega, {\cal F}, {\mathbb P})$. The entropy rate  (differential) is defined by 
\begin{align}
H_R(X^\infty) \tri  \lim_{n \longrightarrow \infty} \frac{1}{n} H(X_1, X_2, \ldots, X_n)
\end{align}
when the limit exists.\\
\end{definition}

 The next theorem quantifies the existence of entropy rates of stationary Gaussian processes \cite{caines1988}.\\

\begin{theorem} The entropy rate of stationary zero mean full rank Gaussian process \cite{caines1988}\\
\label{thm_er_st_G}
Let $X_t: \Omega \rar {\mathbb R}^{n_x}, n_x\in Z_+, \forall t \in {\mathbb Z}_+$ be a stationary Gaussian process, with  zero mean, and  full rank covariance of $ {\bf X}^n$. Let ${\cal  H}_t^X$ denote the Hilbert space of RVs generated by $\{X_t: s\leq t, s, t \in Z_+\}$, and  
define the innovations process by 
\begin{align}
\Sigma_t \tri {\bf E}\Big\{ \Big(X_t - {\bf E}\Big\{X_t\Big|{\cal H}_{t-1}^X \Big\}\Big) \Big(X_t - {\bf E}\Big\{X_t\Big|{\cal H}_{t-1}^X \Big\}\Big)^T\Big\}\succ 0
\end{align}
and its limit 
\begin{align}
\Sigma \tri \lim_{n \longrightarrow \infty} \Sigma_n
\end{align}
Then the entropy rate  is given by 
\begin{align}
H_R(X^\infty) =& \frac{n_x}{2}\log \big(2\pi e\big) +  \frac{1}{2}\lim_{n \longrightarrow \infty} \frac{1}{n} \sum_{t=1}^n \log |\Sigma_t|\\
=&\frac{n_x}{2\pi}\log \big(2\pi e\big) + \frac{1}{2}\log |\Sigma| \label{er_1}
\end{align} 
when it exists. \\
\end{theorem}

An application of Theorem~\ref{thm_er_st_G} is given in  the next proposition \cite{ihara1993}.\\

\begin{proposition} Entropy  rate of Gaussian process described by PSD (\ref{PSD_G}) \\
\label{prop_entr}
Let $V_t, \forall t \in {\mathbb Z}_+$ be a real,  scalar-valued,  stationary Gaussian noise with PSD   (\ref{PSD_G}),  with a corresponding time-invariant stationary realization (similar to Definition~\ref{def_nr_2}). 
Then the entropy rate is given by 
\begin{align}
H_R(V^\infty) =\frac{1}{2}\log \big(2\pi e K_W\big). \label{arma}
\end{align}
\end{proposition}
\begin{proof}  
This is shown in \cite{ihara1993} by using  the Szego formula and Poisson's integral formula.
\end{proof}

The next remark is trivial; it is introduced for subsequent comparison. \\

\begin{remark}
\label{ent_rate_1}
 Let $V_t, \forall t \in {\mathbb Z}_+$ be the nonstationary   ARMA$(a,c), a \in (-\infty,\infty), c\in (-\infty,\infty)$ noise of Example~\ref{ex_1_1_n}. Then   the conditional entropy of $V^n$ for fixed initial state  $S_1=S_1^s=s$, is given by 
\begin{align}
H_R(V^\infty|s)\tri \lim_{n\longrightarrow \infty}  \frac{1}{n} H(V^n|s)= \lim_{n\longrightarrow \infty}  \frac{1}{n}
 \sum_{t=1}^n\frac{1}{2}\log\big(2\pi e K_{W}\big)=\frac{1}{2}\log\big(2\pi e K_{W}\big). \label{arma_tv}
\end{align}
\end{remark}

The next lemma  identifies   fundamental  conditions for the   existence of the entropy rate of the time-varying PO-SS noise realization of Definition~\ref{def_nr_2} (if $S_1=S_1^s=s$ is not fixed),  and includes the entropy rate $H_R(V^\infty)$ of the nonstationary   ARMA$(a,c), a \in (-\infty,\infty), c\in (-\infty,\infty)$ noise of Remark~\ref{ent_rate_1}.   \\

\begin{lemma} Entropy rate of the time-varying PO-SS noise realization of Definition~\ref{def_nr_2}\\
\label{lem_entr_in}
Consider the time-varying PO-SS noise realization of Definition~\ref{def_nr_2}. Then the following hold.\\ 
(a) The joint entropy of $V^n$, when it exists,  is given by 
\begin{align}
H(V^n)=\sum_{t=1}^n H(\hat{I}_t)=\frac{1}{2}\sum_{t=1}^n\log\big(2\pi e K_{\hat{I}_t}\big) \label{entropy_0}
\end{align}
where $\hat{I}_t, t=1, \ldots, n$ is a zero mean covariance $K_{\hat{I}_t}\tri cov(\hat{I}_t, \hat{I}_t)$, Gaussian orthogonal  innovations process of $V^n$,  defined by
\begin{align}
\hat{I}_t \tri & \; V_t-{\bf E}\Big\{V_t\Big|V^{t-1}\Big\},  \hst t=1, \ldots, n \label{entr_n_2}
\end{align}
that is, $\hat{I}_t$ is independent of $\hat{I}_k, \forall k\neq t$. \\
(b) Suppose the sequence $K_{\hat{I}_t}, t=1,2, \ldots, n$, is such that
\begin{align}
\lim_{n \longrightarrow \infty}K_{\hat{I}_n}= K_{\hat{I}}^\infty>0. \label{conv_1}
\end{align}
Then the entropy rate of $V_t, \forall t\in {\mathbb Z}_+$, is given by 
\begin{align}
H_R(V^\infty) =\lim_{n \longrightarrow \infty}\frac{1}{n} \sum_{t=1}^n H(\hat{I}_t)=   \frac{1}{2}\log \big(2\pi e K_{\hat{I}}^\infty \big).\label{in_entr_new}
\end{align} 
\end{lemma}
\begin{proof} 
See Appendix~\ref{sect:app_lem_entr_in}. 
\end{proof}

\ \

%

\begin{remark} Entropy rate of nonstationary Gaussian noise\\
By  Lemma~\ref{lem_entr_in}, a necessary condition for existence of  the entropy rate of nonstationary Gaussian process $V^n$ 
is the convergence of the covariance of the Gaussian orthogonal  innovations process of $V^n$, i.e., of  $K_{\hat{I}_t}\tri cov(\hat{I}_t, \hat{I}_t)$, since $\lim_{n \rightarrow \infty}\frac{1}{n} H(V^n)=\lim_{n \rightarrow \infty}\frac{1}{n}\sum_{t=1}^n H(\hat{I}_t)$.  We can determine such necessary and/or sufficient  conditions from   the convergence properties of the  {\it Generalized Kalman-filter equations} \cite{caines1988,kailath-sayed-hassibi} of   Lemma~\ref{lemma_POSS}.
\end{remark}

%


\subsection{Convergence Properties of Generalized Matrix DREs to AREs}
To address  the  asymptotic properties of estimation errors generated by the  recursions of  Generalized Kalman-filters, such as,  $\widehat{E}_t, t=1,2,\ldots$ of   Theorem~\ref{thm_SS}, generated by  (\ref{kf_m_3a}), we need to  introduce the stabilizing solutions of generalized AREs. The next definition is useful in this respect.\\

\begin{definition} Stabilizing solutions of  generalized matrix AREs \\
\label{def_gmare}
Let $(A, G, Q, S, R, C) \in  {\mathbb R}^{q \times q} \times {\mathbb R}^{q \times k} \times  {\mathbb R}^{k \times k} \times {\mb R}^{k \times p}\times {\mb R}^{p \times p}\times {\mb R}^{p \times q}$.\\
Define the generalized time-invariant matrix DRE 
\begin{align}
&P_{t+1}= A P_t A^T  + G QG^T -\Big(A P_t C^T+GS  \Big) \Big(R+C  P_t C^T\Big)^{-1} . \Big( AP_tC^T+ GS  \Big)^T, \hso P_1=\mbox{given}, \label{gdre_g} \\
&P_t \in {\mathbb S}_+^{q \times q},\: t=1, \ldots, \hso R=R^T \succ 0, \nonumber\\
&F^{CL}(P) \tri A- \Big( A P C^T+GQG^T\Big)\Big(R+ C P C^T \Big)^{-1} C. \nonumber
\end{align}
Define also the corresponding  generalized matrix ARE 
\begin{align}
P= &A PA^T  + G QG^T -\Big(A P C^T+GS  \Big) \Big(R+C  P C^T\Big)^{-1} . \Big( APC^T+ GS  \Big)^T, \hso P \in {\mathbb S}_+^{q \times q}. \label{gare_g}
\end{align}
A solution $P=P^T\succeq  0$ to the generalized matrix ARE (\ref{gare_g}), assuming it exists, is called stabilizing if $spec\big(F^{CL}(P)\big)\in {\mathbb D}_o$. In this case, we  say $F^{CL}(P)$ is asymptotically stable, that is, the eigevalues of $F^{CL}(P)$ are stable.
\end{definition}

With respect to any of the above  generalized matrix  DRE and ARE, we  introduce the important notions of detectability, unit circle controllability, and stabilizability. We use these notions to characterize the convergence properties of solutions of  generalized matrix DREs, $P_n$,  as $n\longrightarrow \infty$,  to a unique symmetric, nonnegative, stabilizing solution $P$ of the generalized matrix ARE. These notions are used to identify necessary and/or sufficient conditions for  the error recursions of  generalized Kalman-filters, such as,  $\widehat{E}_t, t=1,2,\ldots$ of   Theorem~\ref{thm_SS}, generated by  (\ref{kf_m_3a}), to converge in mean-square sense,  to a unique limit.   \\

\begin{definition} Detectability, Stabilizability, Unit Circle controllability \\
\label{def:det-stab}
Consider the generalized matrix ARE of Definition~\ref{def_gmare}, and introduce the matrices
\begin{align}
A^*\tri A- GS R^{-1} C, \hst B^* \tri Q- SR^{-1} S^T, \hst B^*= B^{*, \frac{1}{2}} \big(B^{*, \frac{1}{2}}\big)^T.
\end{align}
(a) The pair  $\big\{A,C\big\}$ is called detectable if there exists a matrix $K \in {\mathbb R}^{q\times p}$ such that  $spec\big(A- K C\big) \in {\mathbb D}_o$, i.e., the eigenvalues $\lambda$ of $A-K C$ lie in ${\mb D}_o$  (stable).\\
(b) The pair $ \big\{A^*, G B^{*,\frac{1}{2}}\big\}$ is called unit circle controllable if  there exists a $K \in {\mathbb R}^{k\times q }$ such that   $spec\big(A^*- G B^{*,\frac{1}{2}}K \big) \notin \{c \in {\mathbb C}: |c|=1\}$, i.e., all eigenvalues $\lambda$ of $A^*- G B^{*,\frac{1}{2}}K$ are  such that $|\lambda|\neq 1$.\\
(c) The pair $ \big\{A^*, G B^{*,\frac{1}{2}}\big\}$ is called stabilizable if  there exists a $K \in {\mathbb R}^{k\times q}$ such that  $spec\big(A^*- G B^{*,\frac{1}{2}}K \big)  \in {\mathbb D}_o$, i.e., all all eigenvalues $\lambda$ of $A^*- G B^{*,\frac{1}{2}}K$ lie in ${\mb D}_o$.\\
(d) The pair  $\big\{A,C\big\}$ is called observable if the rank condition holds, 
\bea
\mbox{rank}\big({\cal O}\big)=q, \hso {\cal O}\tri \left[ \begin{array}{c} C\\ CA \\ \vdots \\ CA^{q-1}\end{array} \right] .
\eea
(e) The pair $ \big\{A^*, G B^{*,\frac{1}{2}}\big\}$ is called controllable if the rank  condition holds, 
\bea
\mbox{rank}\big({\cal C}\big)=q, \hso {\cal O}\tri \left[ \begin{array}{cccc} G B^{*,\frac{1}{2}} & A^* G B^{*,\frac{1}{2}} & \ldots & \big(A^*\big)^{q-1}G B^{*,\frac{1}{2}} \end{array} \right].
\eea
\end{definition}

\begin{remark} The following are 
 well-known \cite{caines1988}. If  the pair $ \big\{A,C\big\}$  is observable  then it is stabilizable, and if the pair $ \big\{A^*, G B^{*,\frac{1}{2}}\big\}$  is controllable then it is stabilizable. 
\end{remark}

The next theorem characterizes, detectability,   unit circle controllability, and stabilizability \cite{kailath-sayed-hassibi,vanschuppen2010}.\\
 
\begin{lemma}\cite{kailath-sayed-hassibi,vanschuppen2010}  Necessary and sufficient conditions for detectability, unit circle controllability, stabilizability\\
(a) The pair $\big\{A,C\big\}$ is detectable  if and only if  there exists no eigenvalue and eigenvector $\{\lambda,x\}$, $Ax=\lambda x$, such that   $|\lambda|\geq 1$ and such that    $Cx =0$. \\
(b) The pair $\big\{A^*, GB^{*,\frac{1}{2}}\big\}$ is unit circle controllable if and only if 
there exists no eigenvalue and eigenvector $\{\lambda,x\}$, $x^T \big(A^*\big)^T=x^T\lambda$, such that   $|\lambda|= 1$ and such that $x^T G B^{*,\frac{1}{2}} =0$. \\
(c) The pair $\big\{A^*, G B^{*,\frac{1}{2}}\big\}$ is stabilizable  if and only there exists no eigenvalue and eigenvector $\{\lambda,x\}$, $x^T \big(A^*\big)^T=x^T\lambda$, such that   $|\lambda|\geq 1$ and such that $x^T G B^{*,\frac{1}{2}} =0$.
\end{lemma}

In the next theorem we summarize known results on sufficient and/or necessary conditions for convergence of solutions $\{P_t, t=1, 2, \ldots, n\}$ of the generalized time-invariant DRE (\ref{gdre_g}), as $n \longrightarrow \infty$,  to a symmetric, nonnegative  $P\succeq 0$,  which is the unique stabilizing solution of a corresponding generalized ARE (\ref{gare_g}).    \\

\begin{theorem}\cite{kailath-sayed-hassibi,caines1988} Convergence of time-invariant generalized DRE\\
\label{thm_ric}
Let  $\{P_t, t=1, 2, \ldots, n\}$ denote a sequence that satisfies the time-invariant  generalized DRE (\ref{gdre_g}) with arbitrary initial condition $P_1\geq 0$. The following hold.\\
(1) Consider the generalized DRE (\ref{gdre_g})  with zero initial condition, i.e., $P_{1}=0$, and assume,  the pair $\big\{A,C\big\}$ is detectable, and  the pair $\big\{A^*, G B^{*,\frac{1}{2}}\big\}$ is unit circle controllable.\\
Then the  sequence $\{P_{t}: t=1, 2, \ldots, n\}$ that satisfies the generalized DRE (\ref{gdre_g}),  with zero initial condition $P_{1}=0$,  converges to $P$, i.e., $\lim_{n \longrightarrow \infty} P_{n} =P$, where  $P$ satisfies the generalized matrix ARE (\ref{gare_g})
 if and only if the pair $\big\{A^*, G B^{*,\frac{1}{2}}\big\}$ is stabilizable.\\
(2) Assume,  the pair $\big\{A,C\big\}$ is detectable, and  the pair $\big\{A^*, B^{*,\frac{1}{2}}\big\}$ is unit circle controllable.  Then there exists a unique stabilizing solution $P\succeq 0$ to the generalized ARE (\ref{gare_g}), i.e.,  such that,  $spec\big(F^{CL}(P)\big) \in {\mathbb D}_o$, if and only if  $\{A^*, G B^{*,\frac{1}{2}}\}$ is stabilizable.\\
(3) If $\{A, C\}$ is detectable and $\{A^*, G B^{*,\frac{1}{2}}\}$ is stabilizable,  then any solution $P_{t}, t=1, 2, \ldots,n$ to the generalized matrix DRE (\ref{gdre_g})  with arbitrary  initial condition, $P_{1}\succeq 0$ is such that $\lim_{n \longrightarrow \infty} P_{n} =P$, where $P\succeq  0$ is the  unique solution of  the generalized matrix ARE (\ref{gare_g}) with $spec\big(F^{CL}(P)\big) \in {\mathbb D}_o$  i.e., it is stabilizing.
\end{theorem}

\subsection{Feedback Rates} 
Now, we return to the feedback rates of Definition~\ref{def_limit_1}. The  next corollary is an application of Theorem~\ref{thm_ric} to the generalized Kalman-filter of Lemma~\ref{lemma_POSS} (for the time-invariant PO-SS realization);  it identifies conditions for existence of the entropy rate $H_R(V^\infty)$, irrespectively of whether the noise is stable or unstable.  \\

\begin{corollary} The entropy rate of PO-SS noise realization based on the  generalized Kalman-filter\\
\label{cor_POSS_IH}
Let  $\Sigma_t^o=\Sigma_t, t=1, 2, \ldots $ denote the solution of the generalized matrix DRE  (\ref{dre_1}) of the generalized Kalman-filter of Lemma~\ref{lemma_POSS}  of the  time-invariant PO-SS realization of $V^n$ of  Definition~\ref{def_nr_2},  i.e.,  $(A_t, B_t, C_t, N_t, K_{W_t})=(A, B, C, N, K_{W}), \forall t$, generated by 
\begin{align}
\Sigma_{t+1}^o= &A \Sigma_{t}^oA^T  + B K_{W}B^T -\Big(A  \Sigma_{t}^oC^T+BK_{W}N^T  \Big) \Big(N K_{W} N^T+C  \Sigma_{t}^o C^T\Big)^{-1}\nonumber \\
& \hst . \Big( A  \Sigma_{t}^oC^T+ B K_{W}N^T  \Big)^T, \hso \Sigma_t^o \succeq 0,    \hso t=1, \ldots, n, \hso \Sigma_{1}^o=K_{S_1}\succeq 0. \label{dre_1_TI}\\
M^{CL}(\Sigma^o)\tri & A- M(\Sigma^o)C, \hso M(\Sigma^o) \tri \Big( A  \Sigma^o C^T+B_{t} K_{W}N^T\Big)\Big(N K_{W}N^T+ C \Sigma^o C^T \Big)^{-1}. \label{dre_1_TI_a}
\end{align}
Let $\Sigma^\infty =\Sigma^{\infty,T} \succeq 0
$ be a solution of the corresponding generalized ARE  
\begin{align}
\Sigma^\infty= &A \Sigma^\infty A^T  + B K_{W}B^T -\Big(A \Sigma^\infty C^T+B K_{W}N^T  \Big) \Big(N K_{W} N^T+C  \Sigma^\infty C^T\Big)^{-1} \Big( A  \Sigma^\infty C^T+ B K_{W}N^T  \Big)^T.  \label{dre_1_SS}
\end{align}
Define the matrices
\begin{align}
&GQG^T \tri B K_W B^T, \hso G S\tri B K_W N^T,  \hso  R\tri N K_W N^T \hso \Longrightarrow  \hso G \tri B, \hso Q\tri K_W, \hso  S\tri K_W N^T, \label{map_1}\\
&A^*\tri A- B K_W N^T \big(N K_W N^T\big)^{-1} C, \hst B^* \tri K_W- K_W N^T \Big(N K_W N^T\Big)^{-1} \Big(K_W N^T\Big)^T. \label{noise_st}
\end{align}
(a) All statements of Theorem~\ref{thm_ric} hold with $(G, Q,S, R)$ as defined by (\ref{map_1}), (\ref{noise_st}).\\
In particular, suppose  \\
(i) $\{A, C\}$ is detectable,  and\\
(ii) $\{A^*, G B^{*,\frac{1}{2}}\}$ is stabilizable.\\
Then any solution $\Sigma_{t}^o, t=1, 2, \ldots,n$ to the generalized matrix DRE (\ref{dre_1_TI})  with arbitrary  initial condition, $\Sigma_{1}^o\succeq 0$ is such that $\lim_{n \longrightarrow \infty} \Sigma_{n}^o =\Sigma^\infty$, where $\Sigma^\infty\succeq  0$ is the  unique solution of  the generalized matrix ARE (\ref{dre_1_SS}) with $spec\big(M^{CL}(\Sigma^\infty)\big) \in {\mathbb D}_o$  i.e., it is stabilizing. \\
(b) The entropy  rate  of $V^n$ is given by 
\begin{align}
H_R(V^\infty)=&\lim_{n\longrightarrow \infty}\frac{1}{2n}\sum_{t=1}^n\log\Big(2\pi e \Big[C \Sigma_t^o C^T + N K_{W} N^T\Big] \Big)\\
=&  H(\hat{I}_t^\infty)\tri \frac{1}{2}\log\Big(2\pi e \Big[C \Sigma^\infty C^T + N K_{W} N^T\Big] \Big), \hso  \forall \Sigma_1^o \succeq 0, \hso    \forall t  \label{ti_inno_11}
\end{align}
where 
\begin{align}
\hat{I}_t^\infty \tri C\big(S_t-\hat{S}_t^\infty\big)+ N W_t \in N(0, C \Sigma^\infty C^T + N K_{W} N^T),\hso  t=1,2, \ldots,\label{ti_inno_1}
\end{align}
 is the stationary Gaussian innovations process, i.e.,  with  $\Sigma_t^o$ replaced by $\Sigma^\infty$, and the entropy rate $H_R(V^\infty)$ is  
   independent of the initial data $\Sigma_1^o\succeq 0$. 
\end{corollary}
\begin{proof}  This is a direct application of Theorem~\ref{thm_ric}. The last part follows from Lemma~\ref{lem_entr_in}.
\end{proof}

Next we apply Corollary~\ref{cor_POSS_IH} to the nonstationary AR$(a,c), a \in (-\infty,\infty),c \in (-\infty,\infty)$ noise.  \\

\begin{lemma}  Properties of solutions of DREs and AREs of AR$(a,c), a \in (-\infty,\infty),c \in (-\infty,\infty)$ noise and entropy rate $H_R(V^\infty)$  \\
\label{lem_pr_are_ar}
Consider the AR$(a,c), a \in (-\infty,\infty),c \in (-\infty,\infty)$ noise of Example~\ref{ex_1_1_n}.(a), and 
the DRE $\Sigma_t^o\tri \Sigma_t, t=1, \ldots, n$,  generated by Corollary~\ref{cor_ex_1}.(a), i.e., 
\begin{align}
&\Sigma_{t+1}^o= \big(c\big)^2 \Sigma_{t}^o  + K_{W} -\Big(c  \Sigma_{t}^o\big(c-a\big)+K_{W}  \Big)^2 \Big(K_{W} +\big(c-a\big)^2  \Sigma_{t}^o \Big)^{-1}, \hso t=1, \ldots, n, \label{ar_ac,DRE} \\
&\Sigma_{1}^o=K_{S_1}=\frac{\big(c_0\big)^2 K_{S_0} +\big(a_0\big)^2 K_{W_0}}{\Big(c_0-a_0\Big)^2}\geq 0. \label{ar_ac,DRE_a}
\end{align}
where $K_W>0, c\neq a$, $K_{S_0}\geq 0, K_{W_0}\geq 0$.  Let $\Sigma^\infty\geq 0$ be a solution of the corresponding  generalized ARE 
\begin{align}
&\Sigma^\infty= \big(c\big)^2 \Sigma^\infty  + K_{W} -\Big(c  \Sigma^\infty\big(c-a\big)+K_{W}  \Big)^2 \Big(K_{W} +\big(c-a\big)^2  \Sigma^\infty \Big)^{-1}. \label{ar_ac_ARE}
\end{align}
Then the  detectability and stabilizability pairs are 
\begin{align}
\{A, C\}=\{c, c-a\}, \hst \{A^*, G B^{*,\frac{1}{2}}\}=\{a, 0\}.   \label{ar_ac,DRE_c}
\end{align}
and the following hold.\\
(1)  The pair $\{A, C\}=\{c, c-a\}$ is  detectable $\forall c \in (-\infty,\infty), a\in (-\infty,\infty)$ (the restriction  $c\neq a$ is always assumed).   \\
(2) The pair $\{A^*, G B^{*,\frac{1}{2}}\}=\{a, 0\}$ is unit circle controllable if and only if $|a|\neq 1$ ($\forall c\in (-\infty,\infty)$). \\
(3) The pair $\{A^*, G B^{*,\frac{1}{2}}\}=\{a, 0\}$ is stabilizable  if and only if $a \in (-1,1)$ ($\forall c\in (-\infty,\infty)$).\\
(4) Suppose $c\in (-\infty,\infty)$ and  $|a|\neq 1$. The sequence    $\{\Sigma_t^o, t=1, 2, \ldots, n\}$ that satisfies  the generalized DRE with zero initial condition, $\Sigma_1^o=0$ 
converges to $\Sigma^\infty$, i.e., $\lim_{n \longrightarrow \infty} \Sigma_n^o=\Sigma^\infty$, where $\Sigma^\infty\geq 0$ satisfies the ARE (\ref{ar_ac_ARE}) 
 if and only if the $\{A^*, GB^{*,\frac{1}{2}}\}=\{a,0\}$ is stabilizable, equivalently, $|a|<1$. Moreover, 
 the  two solutions of the quadratic equation (\ref{ar_ac_ARE}), without imposing $\Sigma^\infty \geq 0$ are 
\begin{align}
\Sigma^\infty= \left\{   \begin{array}{ll}  0 & \mbox{the unique, stabilizing, $\Sigma^\infty \geq 0$  solution of (\ref{ar_ac_ARE}) for $c\in (-\infty,\infty), |a|<1$} \\
\frac{K_W\big(a^2-1\big)}{\big(c-a\big)^2}<0 & \mbox{the non-stabilizing, $\Sigma^\infty<0$ solution of (\ref{ar_ac_ARE}) for $c\in (-\infty,\infty), |a|<1$.}
\end{array} \right.\label{are_2_sol}
\end{align}
That is,  $\lim_{n\longrightarrow \infty}\Sigma_n^0=\Sigma^\infty=0$ is the  unique and stabilizing solution $\Sigma^\infty\geq 0$ of  (\ref{ar_ac_ARE}), i.e., such that $|M^{CL}(\Sigma^\infty)|<1$, if and only if $|a|<1$.\\
(5) Suppose $c\in (-\infty,\infty)$ and $|a|<1$.   Then any solution $\Sigma_{t}^o, t=1, 2, \ldots,n$ to the generalized  DRE (\ref{ar_ac,DRE})  with arbitrary  initial condition, $\Sigma_{1}^o\geq  0$ is such that $\lim_{n \longrightarrow \infty} \Sigma_{n}^o =\Sigma^\infty$, where $\Sigma^\infty \geq  0$ is the  unique solution of  the generalized  ARE (\ref{ar_ac_ARE}) with $M^{CL}(\Sigma^\infty)\in (-1,1)$  i.e., it is stabilizing, and moreover $\Sigma^\infty=0$.    \\
(6) Suppose $c\in (-\infty,\infty)$ and $|a|<1$.  The entropy rate of $V_t, \forall t\in {\mathbb Z}_+$, is given by 
\begin{align}
H_R(V^\infty) =\lim_{n \longrightarrow \infty}\frac{1}{n} \sum_{t=1}^n \frac{1}{2}\log \Big(2\pi e\big[(c-a)^2 \Sigma_t^o+K_W\big] \Big)=   \frac{1}{2}\log \big(2\pi e K_W \big), \hso \forall \Sigma_1^o\geq 0. \label{en_r_ar_ac}
\end{align} 
 \end{lemma}
\begin{proof} 
 See Appendix~\ref{sect:app_lem_pr_are_ar}.
\end{proof}

\begin{remark} Lemma~\ref{lem_pr_are_ar}.(4) emphasizes the fact that in asymptotic analysis of 
$\{\Sigma_t^o, t=1, 2, \ldots,\}$, that satisfies the DRE  (\ref{ar_ac,DRE}),  (\ref{ar_ac,DRE_a}), its limiting value,  $\lim_{n \longrightarrow \infty} \Sigma_n^o=\Sigma^\infty$, where $\Sigma^\infty\geq 0$ satisfies the ARE (\ref{ar_ac_ARE}), with two solutions $\Sigma^\infty=0$  and $\Sigma^\infty=
\frac{K_W(a^2-1)}{(c-a)^2}$. However,  for any $c \in (-\infty, \infty)$, although it is clear that for,  $|a|<1$,  the unique and stabilizing solution is $\Sigma^\infty=0$, since the other solution  $\Sigma^\infty=
\frac{K_W(a^2-1)}{(c-a)^2}<0$, i.e., it is negative, for $|a|\geq 1$, the unique and stabilizing solution is again $\Sigma^\infty=0$.
\end{remark}

To gain additional insight, in the next remark we discuss the application of Lemma~\ref{lem_pr_are_ar} to the AR$(c), c\in (-\infty,\infty)$ noise.\\

\begin{remark} Entropy rate $H_R(V^\infty)$ of the AR$(c),c \in (-\infty,\infty)$ noise \\ 
Consider the  nonstationary  AR$(c), c\in (-\infty,\infty)$ noise defined by (\ref{ex_1_1_AR1}). Then from  Lemma~\ref{lem_pr_are_ar}, 
$\Sigma_t^o, t=1, \ldots, n$ is the solution of   (\ref{ar_ac,DRE}), (\ref{ar_ac,DRE_a}), with $a=0$  (see Corollary~\ref{cor_ex_1}.(b), (\ref{grde_ar1})), and   (\ref{ar_ac_ARE}) degenerates to the ARE, 
\bea
 \Sigma^\infty= \big(c\big)^2 \Sigma^\infty  + K_{W} -\Big(\big(c\big)^2  \Sigma^\infty+K_{W}  \Big)^2 \Big(K_{W} +\big(c\big)^2  \Sigma^\infty \Big)^{-1} \label{are_ar1}
\eea 
For   $a=0$, by  (\ref{ar_ac,DRE_c}) the pair $\{A,C\}=\{c,c\}$ is detectable, and  the pair $\{A^*, GB^{*,\frac{1}{2}}\}=\{0, 0\}$ is stabilizable.  The two solutions of the ARE (\ref{are_ar1}), without imposing $\Sigma^\infty\geq 0$, are 
\begin{align}
\Sigma^\infty= \left\{   \begin{array}{cc}  0 & \mbox{the unique, stabilizing, nonnegative solution of the ARE} \\
-\frac{K_W}{c^2}<0 & \mbox{the non-stabilizing, negative solution) of the ARE}
\end{array} \right.
\end{align}
That is,  $\lim_{n \longrightarrow \infty} \Sigma_n^o=\Sigma^\infty\geq0$, where $\Sigma^\infty=0$ is the unique (stabilizing) solution of the ARE, and  corresponds to the stable eigenvalue of the error equation  (see (\ref{dre_1_a}), i.e.,     $M^{CL}(\Sigma^\infty)=c- \frac{K_W}{K_W}c=0 $. 
\end{remark}

%

Next we compute the entropy rate  
 $H_R(V^\infty)$ of  the time-invariant nonstationary PO-SS$(a,c, b^1, b^2, d^1, d^2)$ noise of Corollary~\ref{entr_poss} to show fundamental differences from the entropy rate $H_R(V^\infty)$ of the AR$(a,c)$ noise of Lemma~\ref{lem_pr_are_ar}.\\
 
 \begin{lemma}  Properties of solutions of DREs and AREs of PO-SS$(a,c, b^1=b, b^2=0, d^1=0, d^2=d)$ noise and entropy rate $H_R(V^\infty)$  \\
\label{lem_pr_are_po-ss}
Consider the the time-invariant nonstationary PO-SS$(a,c, b^1, b^2=0, d^1=0, d^2=d)$ noise   of Example~\ref{ex_1_poss}, i.e., given by 
\begin{align}
&S_{t+1} = a S_{t} + b W_t^1, \hso t=1, 2, \ldots, n-1 \label{ex_1_1_poss_sc}\\
&V_t =c S_t + d W_t^2,\hso t=1, \ldots, n,\label{ex_1_1_poss_a_sc} 
\end{align}
and the sequence  $\Sigma_t^o\tri \Sigma_t, t=1, \ldots, n$,  generated by the DRE of   Lemma~\ref{lemma_POSS} (see  (\ref{poss_ex_1}), i.e., 
\begin{align}
\Sigma_{t+1}^o= \big(a\big)^2 \Sigma_{t}^o  +\big(b\big)^2 K_{W^1}  -\Big(a \Sigma_{t}^o c \Big)^2 \Big(\big(d\big)^2 K_{W^2} + \big(c\big)^2  \Sigma_{t}^o \Big)^{-1},   \hso t=1, \ldots, n, \hso \Sigma_{1}^o=K_{S_1}\geq  0, \hso \Sigma_t^o \geq 0 \label{po-ss_ar1}
\end{align}
where $\big(b\big)^2K_{W^1}\geq 0, \big(d\big)^2 K_W^2>0$.  
Let $\Sigma^\infty\geq 0$ be the corresponding solution of   generalized ARE
 \begin{align}
\Sigma^\infty= \big(a\big)^2 \Sigma^\infty  + \big(b\big)^2 K_{W^1}-\Big(a \Sigma^\infty c  \Big)^2 \Big(\big(d\big)^2 K_{W^2}+ \big(c\big)^2  \Sigma^\infty \Big)^{-1}. \label{po-ss_ar1_ih}
\end{align}
Then the detectability and stabilizability pairs are
\begin{align}
\{A, C\}=\{a, c\}, \hst \{A^*, G B^{*,\frac{1}{2}}\}=\{a, b \big(K_{W^1}\big)^{\frac{1}{2}} \}.   \label{po-ss-ar1_ih1}
\end{align}
and the following hold.\\
(1)  The pair $\{A, C\}=\{a, c\}$ is  detectable $\forall c \in (-\infty,\infty), a\in (-\infty,\infty), c\neq 0$. If $c=0$ the pair $\{A, C\}=\{a, 0\}$ is  detectable if and only if  $|a|<1$. \\
(2) The pair $\{A^*, G B^{*,\frac{1}{2}}\}=\{a, b \big(K_{W^1}\big)^{\frac{1}{2}} \}$ is unit circle controllable if and only if $|b \big(K_{W^1}\big)^{\frac{1}{2}}|\neq 1$, $\forall a\in (-\infty,\infty), c\in (-\infty,\infty)$. \\
(3) The pair $\{A^*, G B^{*,\frac{1}{2}}\}=\{a, b \big(K_{W^1}\big)^{\frac{1}{2}}\}$ is stabilizable  if $b \big(K_{W^1}\big)^{\frac{1}{2}}\neq 0$,  $\forall a \in (\infty,\infty),  c\in (-\infty,\infty)$. If  $b \big(K_{W^1}\big)^{\frac{1}{2}}= 0$ the pair $\{A^*, G B^{*,\frac{1}{2}}\}=\{a, 0\}$ is stabilizable if and only if $|a|<1$.  \\
(4) Define the set 
\begin{align}
{\cal L}^\infty \tri \Big\{ (a,c, \big(b\big)^2K_{W^1})\in & (-\infty,\infty)^2\times [0,\infty) :\hso   \mbox{$(i)$ the pair $\{A, C\}=\{a, c\}$ is detectable, and} \nonumber \\
 &\mbox{$(ii)$ the pair}  \{A^*, G B^{*,\frac{1}{2}}\}=\{a, b \big(K_{W^1}\big)^{\frac{1}{2}} \}\hso  \mbox{is stabilizable} \Big\}.
\end{align}
For any $(a,c,b\big(K_{W^1}\big)^{\frac{1}{2}}) \in {\cal L}^\infty$,   any solution $\Sigma_{t}^o, t=1, 2, \ldots,n$ to the (classical)  DRE (\ref{po-ss_ar1})  with arbitrary  initial condition, $\Sigma_{1}^o\geq  0$ is such that $\lim_{n \longrightarrow \infty} \Sigma_{n}^o =\Sigma^\infty$, where $\Sigma^\infty \geq  0$ is the  unique solution of  the (classical)  ARE (\ref{po-ss_ar1_ih}) with $M^{CL}(\Sigma^\infty)\in (-1,1)$  i.e., it is stabilizing.   \\
(5) For any $(a,c,b^2 K_{W^1}) \in {\cal L}^\infty$ of part (4)  the entropy rate of $V_t, \forall t\in {\mathbb Z}_+$, is given by 
\begin{align}
H_R(V^\infty) =\lim_{n \longrightarrow \infty}\frac{1}{n} \sum_{t=1}^n \frac{1}{2}\log \Big(2\pi e\big[(c)^2 \Sigma_t^o+\big(d\big)^2K_{W^2}\big] \Big)=   \frac{1}{2}\log \big(2\pi e \big[(c)^2 \Sigma^\infty+\big(d\big)^2K_{W^2}\big] \big), \hso \forall \Sigma_1^o\geq 0.\label{en_r_ar_ac_ih_1}
\end{align} 
\end{lemma}
\begin{proof} Follows from Theorem~\ref{thm_ric}. 
\end{proof}

Next, we turn our attention to the convergence properties of the entropy rate $H_R(Y^\infty)$, which is needed for the  characterization of $C^{fb,o}(\kappa)$ of Definition~\ref{def_limit_1}.\\

\begin{theorem} Asymptotic properties of entropy rate $H_R(Y^\infty)$  of Theorem~\ref{thm_SS}\\
\label{thm_fc_IH}
Let  $K_t^o, t=1, \ldots, $ be  the solution of the generalized DRE  (\ref{kf_m_4_a}) of the generalized Kalman-filter of Theorem~\ref{thm_SS}, corresponding to the time-invariant PO-SS realization of $V^n$ of  Definition~\ref{def_nr_2}, $(A_t, B_t, C_t, N_t, K_{W_t})=(A, B, C, N, K_{W}), \forall t$, with time-invariant strategies $(\Lambda_t,K_{Z_t})=(\Lambda^\infty,K_Z^\infty), \forall t$, generated by
\begin{align}
&K_{t+1}^o= A K_{t}^oA^T  + M(\Sigma_t^o)K_{\hat{I}_t^o}\big(M(\Sigma_t^o)\big)^T -\Big(A  K_{t}^o\big(\Lambda^\infty + C \big)^T+ M(\Sigma_t^o)K_{\hat{I}_t^o}   \Big) \Big( K_{\hat{I}_t^o}+ K_{Z}^\infty \nonumber \\
&+ \big(\Lambda^\infty + C \big) K_{t}^o \big(\Lambda^\infty + C \big)^T     \Big)^{-1} \Big(  A  K_{t}^o\big(\Lambda^\infty + C \big)^T+ M(\Sigma_t^o)K_{\hat{I}_t^o}      \Big)^T, \hso K_t^o=K_t^{o,T} \succeq 0, \hso t=1, \ldots, n, \hso K_1^o=0 \label{kf_m_4_a_TI}
\end{align}
where 
\begin{align}
& K_{\hat{I}_t^o}=C \Sigma_t^o C^T + N K_{W} N^T, \hso \Sigma_t^o \hso \mbox{is a solution of (\ref{dre_1_TI})}, \hso \mbox{$M(\Sigma^o)$ is given by (\ref{dre_1_TI_a})}, \\
&F^{CL}(\Sigma^o, K^o)\tri A-  F(\Sigma^o, K^o)\Big(\Lambda^\infty+C\Big), \\
&F(\Sigma^o, K^o) \tri \Big(A  K^o\big(\Lambda^\infty + C \big)^T+   M(\Sigma^o) K_{\hat{I}^o} \Big)\Big\{K_{\hat{I}^o}+   K_{Z^\infty} + \big(\Lambda^\infty + C \big) K^o \big(\Lambda^\infty + C \big)^T \Big\}^{-1}   .
\end{align}
Define  the corresponding generalized ARE by  
\begin{align}
&K^\infty= A K^\infty A^T  + M(\Sigma^\infty)K_{\hat{I}^\infty}\big(M(\Sigma^\infty)\big)^T -\Big(A  K^\infty\big(\Lambda^\infty + C \big)^T+ M(\Sigma^\infty)K_{\hat{I}^\infty}   \Big) \Big( K_{\hat{I}^\infty}+ K_{Z}^\infty \nonumber \\
&+ \big(\Lambda^\infty + C \big) K^\infty \big(\Lambda^\infty + C \big)^T     \Big)^{-1} \Big(  A  K^\infty\big(\Lambda^\infty + C \big)^T+ M(\Sigma^\infty)K_{\hat{I}^\infty}      \Big)^T, \hso K^\infty=K^{\infty,T} \succeq 0. \label{kf_m_4_a_TI_ARE}
\end{align}
where
\begin{align}
& K_{\hat{I}^\infty}=C \Sigma^\infty C^T + N K_{W} N^T, \hso \Sigma^\infty \hso \mbox{is a solution of (\ref{dre_1_SS})}, \hso \mbox{$M(\Sigma^\infty)$ is given by (\ref{dre_1_TI_a})}.
\end{align}
Introduce  the matrices
\begin{align}
&C(\Lambda^\infty)\tri \Lambda^\infty + C, \hso   GQG^T \tri M(\Sigma^\infty)K_{\hat{I}^\infty}\big(M(\Sigma^\infty)\big)^T, \hso G S\tri M(\Sigma^\infty)K_{\hat{I}^\infty} ,  \hso  R(K_Z^\infty)\tri K_{\hat{I}^\infty}+ K_{Z}^\infty. \label{map_1_ih}\\
&\Longrightarrow \hso \hso G \tri M(\Sigma^\infty), \hso Q\tri K_{\hat{I}^\infty}, \hso  S\tri K_{\hat{I}^\infty}, \\
&A^*(\Lambda^\infty,K_Z^\infty)\tri A- M(\Sigma^\infty)K_{\hat{I}^\infty} \Big(  K_{\hat{I}^\infty}+ K_{Z}^\infty \Big)^{-1}\Big(\Lambda^\infty + C\Big), \hso B^*(K_Z^\infty) \tri K_{\hat{I}^\infty}-K_{\hat{I}^\infty}  \Big(K_{\hat{I}^\infty}+ K_{Z}^\infty\Big)^{-1}K_{\hat{I}^\infty}. \label{input_st}
\end{align}
Suppose the detectability and stabilizability conditions of Corollary~\ref{cor_POSS_IH}.(i) and (ii) hold. \\
Then all statements of Theorem~\ref{thm_ric} hold with $(C(\Lambda^\infty), G, Q,S, R(K_Z^\infty))$ as defined by (\ref{map_1_ih}).\\
In particular, suppose  \\
(i) $\{A, C(\Lambda^\infty)\}=\{A,  \Lambda^\infty + C\}$ is detectable,  and\\
(ii) $\{A^*(\Lambda^\infty, K_Z^\infty), G B^{*,\frac{1}{2}}(K_Z^\infty)\}$ is stabilizable.\\ 
Then any solution $K_{t}^o, t=1, 2, \ldots,n$ to the generalized matrix DRE (\ref{kf_m_4_a_TI})  with arbitrary  initial condition, $K_{1}^o\succeq 0$ is such that $\lim_{n \longrightarrow \infty} K_{n}^o =K^\infty$, where $K^\infty\succeq  0$ is the  unique solution of  the generalized matrix ARE (\ref{kf_m_4_a_TI_ARE}) with $spec\big(F^{CL}(K^\infty,\Sigma^\infty)\big) \in {\mathbb D}_o$  i.e., it is stabilizing. \\
Moreover, the entropy rate of $Y^n$ is given by 
\begin{align}
H_R(Y^\infty)=&\lim_{n\longrightarrow \infty}\frac{1}{n}\sum_{t=1}^n H(I_t^o)  =\lim_{n\longrightarrow \infty}\frac{1}{2n}\sum_{t=1}^n   \log\Big(2\pi e \Big[\Big(\Lambda^\infty +C\Big)K_t^o \Big(\Lambda^\infty +C\Big)^T + K_{\hat{I}_t^o} + K_{Z}^\infty\Big] \Big)  \\
=& H({I}_t^\infty)\tri \frac{1}{2}\log\Big(2\pi e \Big[\Big(\Lambda^\infty +C\Big)K^\infty \Big(\Lambda^\infty +C\Big)^T + K_{\hat{I}^\infty} + K_{Z}^\infty\Big] \Big),\hso \forall K_1^o \succeq 0,  \hso \forall t \label{entropy_rate_1}
\end{align}
where $I_t^o, t=1, \ldots, n$ is the innovations process  of Theorem~\ref{thm_SS} (with indicated changes of time-invariant strategies) and where,
\begin{align}
I_t^\infty=\big(\Lambda^\infty+C\big) \big(\hat{S}_t^\infty-\widehat{\hat{S}_t^\infty}\big)+ \hat{I}_t^\infty+ Z_t \in N(0, \big(\Lambda^\infty +C\big)K^\infty \big(\Lambda^\infty +C\big)^T + K_{\hat{I}^\infty} + K_{Z}^\infty), \hso t=1,2, \ldots, \label{ti_inno_2}
\end{align}
 is the stationary Gaussian innovations process, i.e.,  with  $(K_t^o,\Sigma_t^o)$ replaced by $(K^\infty,\Sigma^\infty)$. 
\end{theorem}
\begin{proof} Since the detectability and stabilizability conditions of Corollary~\ref{cor_POSS_IH} hold, then the statements of Corollary~\ref{cor_POSS_IH} hold. By the continuity property of solutions of  generalized difference Riccati equations, with respect to its coefficients (see \cite{caines1988}), and the convergence of the  sequence $\lim_{n \longrightarrow \infty} \Sigma_n^\infty=\Sigma^\infty$, where $\Sigma^\infty\succeq 0$ is the unique stabilizing solution of (\ref{dre_1_SS}), then the statements of   Theorem~\ref{thm_fc_IH} hold, as stated. In particular, under the detectability and stabilizability conditions (i) and (ii), then   $\lim_{n \longrightarrow \infty} K_n^o=K^\infty$, where $K^\infty\succeq 0$  is the  unique and stabilizing solution of   (\ref{kf_m_4_a_TI_ARE}). 
\end{proof}

In the next lemma we apply Theorem~\ref{thm_fc_IH} to the AR$(a,c), a \in (-\infty,\infty),c \in (-\infty,\infty)$ noise of Example~\ref{ex_1_1_n}.(a), using Lemma~\ref{lem_pr_are_ar}. \\

\begin{lemma} 
\label{lem_pr_are_arac}
Consider the AR$(a,c), a \in (-\infty,\infty),c \in (-\infty,\infty)$ noise of Example~\ref{ex_1_1_n}.(a), and 
the DRE $\Sigma_t^o\tri \Sigma_t, t=1, \ldots, n$ and ARE of Lemma~\ref{lem_pr_are_ar}, (\ref{ar_ac,DRE})-(\ref{ar_ac,DRE_c}). \\
Let $K_t^o, t=1, \ldots, n$ denote the solution of the DRE of  Corollary~\ref{cor_ex_1}.(a), when $\Lambda_t=\Lambda^\infty, K_{Z_t}=K_Z^\infty, K_t=K_t^o, \forall t$, i.e., given by 
 \begin{align}
K_{t+1}^o= &\big(c\big)^2 K_{t}^o  +\big(M(\Sigma_{t}^o)\big)^2K_{\hat{I}_t^o} -\Big(c  K_{t}^o\big(\Lambda^\infty + c-a \big)+ M(\Sigma_t^o)K_{\hat{I}_t^o}   \Big)^2 \nonumber \\
&\hst . \Big(  K_{\hat{I}_t^o}+ K_{Z}^\infty + \big(\Lambda^\infty + c-a \big)^2 K_{t}^o      \Big)^{-1}, \hso  K_1^o=0, \hso  t=1, \ldots, n, \label{kf_m_4_s}\\
K_{Z}^\infty \geq & 0, \hso K_t^o \geq0, \hso t=1, \ldots, n  
\end{align} 
and where 
\begin{align}
&M(\Sigma_t^o) \tri \Big( c \Sigma_{t}^o \big(c-a\big)+K_{W}\Big)\Big(K_{W}+ \big(c-a\big)^2 \Sigma_{t}^o \Big)^{-1}, \\
&K_{\hat{I}_t^o}= \big(c-a\big)^2 \Sigma_t^o +K_{W},  \hso t=1, \ldots, n.
\end{align}
Define the set 
\begin{align}
{\cal L}^\infty \tri \Big\{ (a,c)\in (-\infty,\infty)^2, a\neq c :\hso   &\mbox{$(i)$ the pair $\{A, C\}=\{a, c-a\}$ is detectable, and} \nonumber \\
 &\mbox{$(ii)$ the pair}  \{A^*, G B^{*,\frac{1}{2}}\}=\{a, 0 \}\hso  \mbox{is stabilizable} \Big\}.
\end{align}
For any $(a,c) \in {\cal L}^\infty$, let $K^\infty \geq 0$ be a corresponding solution of the ARE (evaluated at $\lim_{n \longrightarrow \infty} \Sigma_n^\infty=\Sigma^\infty=0$), 
\begin{align}
K^\infty= &\big(c\big)^2 K^\infty  +K_{W} -\Big(c  K^\infty\big(\Lambda^\infty + c-a \big)+ K_{W}   \Big)^2  \Big(  K_{W}+ K_{Z}^\infty + \big(\Lambda^\infty + c-a \big)^2 K^\infty      \Big)^{-1}. \label{kf_m_4_s_ss}\\
K_{Z}^\infty \geq & 0, \hso K_W >0.
\end{align} 
and define the pairs
\begin{align}
&\{A, C(\Lambda^\infty\}=\{c, \Lambda^\infty+c-a\}, \\
&\{A^*(\Lambda^\infty,K_Z^\infty),GB^{*,\frac{1}{2}}(K_Z^\infty)\} =\big\{c-K_W \big(K_W+K_Z^\infty\big)^{-1}\big(\Lambda^\infty+c-a\big),  \Big(K_W- \big( K_W\big)^2 \big(K_W+K_Z^\infty\big)^{-1}\Big)^{\frac{1}{2}}\}.
\end{align}
Then the following hold.\\
(1) Suppose $\Lambda^\infty+c-a\neq 0$. Then  $\{A, C(\Lambda^\infty\}=\{c, \Lambda^\infty+c-a\}$ is detectable  $\forall (a,c)\in (-\infty,\infty)^2$. \\
(2) Suppose $\Lambda^\infty+c-a = 0$. Then  $\{A, C(\Lambda^\infty\}=\{c,0\}$ is detectable for if and only if $|c|<1$  $\forall a\in (-\infty,\infty)$.\\
(3)  Suppose $K_Z^\infty=0$.  Then the pair $\{A^*,GB^{*,\frac{1}{2}}\}=\{-\Lambda^\infty+a, 0\}$ is unit circle controllable if and only if $|\Lambda-a| \neq 1$ $\forall a\in (-\infty,\infty)$. \\
(4)  Suppose $K_Z^\infty=0$.  Then the pair $\{A^*,GB^{*,\frac{1}{2}}\}=\{-\Lambda^\infty+a, 0\}$ is stabilizable  if and only if $|\Lambda-a| <1$ $\forall a\in (-\infty,\infty)$.\\
(5) Suppose $\Lambda^\infty+c-a\neq 0, |\Lambda-a| \neq 1$ $\forall (a,c)\in (-\infty,\infty)^2$,  and  $K_Z^\infty=0$, $\Sigma_1^o=0$. The sequence    $K_t^o, t=1, 2, \ldots, n$ that satisfies  the generalized DRE (\ref{kf_m_4_s})  with zero initial condition, $K_1^o=0$,    converges to $K^\infty\geq 0$, i.e., $\lim_{n \longrightarrow \infty} K_n^o=K^\infty$, where  $K^\infty$ satisfies the generalized ARE, 
\begin{align}
K^\infty= &\big(c\big)^2 K^\infty  +K_{W} -\Big(c  K^\infty\big(\Lambda^\infty + c-a \big)+ K_{W}   \Big)^2  \Big(  K_{W} + \big(\Lambda^\infty + c-a \big)^2 K^\infty      \Big)^{-1}, \hso K^\infty\geq 0 \label{kf_m_4_s_ss_kz}
\end{align} 
 if and only if $|a|<1$ (by Lemma~\ref{lem_pr_are_ar}.(4)), and 
 the 
  pair $\{A^*, GB^{*,\frac{1}{2}}\}=\{-\Lambda^\infty+a,0\}$ is stabilizable, equivalently, $|\Lambda^\infty-a|<1$.\\
Moreover, the solutions of the ARE (\ref{kf_m_4_s_ss_kz}),  under the stabilizability condition, i.e., $|\Lambda^\infty-a|<1$, are   
\begin{align}
K^\infty= \left\{   \begin{array}{ll}  0 & \mbox{the unique, stabilizing, $K^\infty\geq 0$  solultion of (\ref{kf_m_4_s_ss_kz}) for $ |\Lambda^\infty-a|<1$} \\
\frac{K_W\Big(\big(\Lambda^\infty-a\big)^2-1\Big)}{\big(\Lambda^\infty+c-a\big)^2}<0 & \mbox{the non-stabilizing, $K^\infty <0$ solution of (\ref{kf_m_4_s_ss_kz}) for $|\Lambda^\infty-a|<1$.}
\end{array} \right.\label{are_2_sol_kz}
\end{align}
That is,  $\lim_{n\longrightarrow \infty}\Sigma_n^0=\Sigma^\infty=0$ is the  unique and stabilizing solution $\Sigma^\infty\geq 0$ of  (\ref{kf_m_4_s_ss_kz}), i.e., such that $|M^{CL}(\Sigma^\infty)|<1$, if and only if $|\Lambda^\infty-a|<1, |a|<1$.
\end{lemma}
\begin{proof} The statements follow from Lemma~\ref{lem_pr_are_ar}, Theorem~\ref{thm_fc_IH} (and general properties of Theorem~\ref{thm_ric}). 
\end{proof}

\begin{remark}
From Lemma~\ref{lem_pr_are_arac}.(5)  follows
that  if $K_Z^\infty=0, \Sigma_1^o=0$ then the unique and stabilizing solution is $K^\infty=0$ and corresponds to $|\Lambda^\infty-a|<1, a \in (-1,1)$.   This is an  application of  Theorem~\ref{thm_ric}.(1). 
\end{remark}

In the next theorem we characterize the asymptotic limit of Definition~\ref{def_limit_1}, by invoking    Theorem~\ref{thm_SS}, Corollary~\ref{cor_POSS_IH},   and Theorem~\ref{thm_fc_IH}.  \\

 \begin{theorem} Feedback capacity $C^{fb,o}(\kappa)$ of Theorem~\ref{thm_SS}\\ 
 \label{lem_cov}
Consider $C^{fb,o}(\kappa)$ of  Definition~\ref{def_limit_1}   corresponding to   Theorem~\ref{thm_SS}, i.e.,  the PO-SS realization of $V^n$ of  Definition~\ref{def_nr_2} is time-invariant, $(A_t, B_t, C_t, N_t, K_{W_t})=(A, B, C, N, K_{W}), \forall t$, and  the  strategies are time-invariant, $(\Lambda_t,K_{Z_t})=(\Lambda^\infty,K_Z^\infty), \forall t$.  \\
Define the set 
\begin{align}
{\cal P}^\infty \tri & \Big\{(\Lambda^\infty, K_Z^\infty)\in (-\infty, \infty)\times [0,\infty): \nonumber \\
& \mbox{(i)   $\{A, C\}$ of Corollary~\ref{cor_POSS_IH} is detectable,} \nonumber \\
& \mbox{(ii) $\{A^*, G B^{*,\frac{1}{2}}\}$ of  Corollary~\ref{cor_POSS_IH}  is stabilizable, $(A^*, B^*)$ defined by (\ref{noise_st})} \nonumber\\ 
& \mbox{(iii) $\{A, C(\Lambda^\infty)\}=\{A,  \Lambda^\infty + C\}$ of Theorem~\ref{thm_fc_IH}   is detectable,} \nonumber \\
& \mbox{(iv) $\{A^*(\Lambda^\infty,K_Z^\infty), G B^{*,\frac{1}{2}}(K_Z^\infty)\}$ of Theorem~\ref{thm_fc_IH}    is detectable, $(A^*(\Lambda^\infty, K_Z^\infty), B^*(K_Z^\infty))$ defined by (\ref{input_st})} 
\Big\}. \label{adm_set}
\end{align}
Then 
\begin{align}
C^{fb,o}(\kappa)
= & \sup_{  \big(\Lambda^\infty, K_{Z}^\infty\big): \hso \lim_{n \longrightarrow \infty}   \frac{1}{n}\sum_{t=1}^n\big( \Lambda^\infty K_t^o (\Lambda^\infty)^T+K_{Z}^\infty \big)    \leq \kappa } \Big\{  \nonumber \\
&  \lim_{n \longrightarrow \infty}    \frac{1}{2n} 
\sum_{t=1}^n \log\Big( \frac{ \Big(\Lambda^\infty +C\Big)K_t^o \Big(\Lambda^\infty +C\Big)^T +C \Sigma_t^o C^T + N K_{W} N^T+ K_{Z}^\infty     }{ C \Sigma_t^o C^T + N K_{W} N^T  }\Big)\Big\}   \label{per_unit_1}  \\
= & \sup_{  \big(\Lambda^\infty, K_{Z}^\infty\big) \in {\cal P}^\infty(\kappa) } \frac{1}{2} 
 \log\Big( \frac{ \Big(\Lambda^\infty +C\Big)K^\infty \Big(\Lambda^\infty +C\Big)^T +C \Sigma^\infty C^T + N K_{W} N^T+ K_{Z}^\infty     }{ C \Sigma^\infty C^T + N K_{W} N^T  }\Big)  \label{ll_3}
\end{align}
 where  
\begin{align}
{\cal P}^\infty(\kappa)\tri& \Big\{(\Lambda^\infty, K_Z^\infty)\in {\cal P}^\infty: K_Z^\infty\geq 0, \hso  \Lambda^\infty K^\infty (\Lambda^\infty)^T+K_{Z}^\infty  \leq \kappa, \nonumber \\
& K^\infty  \hso \mbox{is the unique and stabilizing solution of  (\ref{kf_m_4_a_TI_ARE}), i.e., $|F^{CL}(\Sigma^\infty, K^\infty)|<1$}     \nonumber \\
&\Sigma^\infty  \hso \mbox{is the unique,   stabilizing solution of (\ref{dre_1_SS}), i.e., $|M^{CL}(\Sigma^\infty)|<1$}     \Big\}        \label{ll_4}
\end{align}
provided there exists $\kappa\in [0,\infty)$ such that the set ${\cal P}^\infty(\kappa)$ is non-empty.\\
Moreover, the maximum element $(\Lambda^\infty, K_Z^\infty) \in {\cal P}^\infty(\kappa)$, is such that, \\
(1)  it induces asymptotic stationarity of the corresponding, input and innovations  processes (see Theorem~\ref{thm_SS} for specification), \\
(2) if $V^n, n=1, 2, \ldots$ is asymptotically stationary, then it  induces asymptotic stationarity of the corresponding, input and output processes, and \\
(3) for (i) and (ii), $C^{fb,o}(\kappa)$ is independent of the initial conditions $K_1^o\succeq 0, \Sigma_1^o\succeq 0$.  
\end{theorem} 
 \begin{proof} 
By Definition~\ref{def_limit_1},   Theorem~\ref{thm_SS}, Corollary~\ref{cor_POSS_IH},   and Theorem~\ref{thm_fc_IH}, then follows  (\ref{per_unit_1}). We defined the set ${\cal P}^\infty$ using  the detectability and stabilizability conditions of Corollary~\ref{cor_POSS_IH},   and Theorem~\ref{thm_fc_IH} to ensure convergence of solutions $\{(K_{t}^o, \Sigma_t^o): t=1, 2, \ldots, n\}$ of the  generalized matrix DREs to unique nonnegative, stabilizing solutions of the corresponding  generalized matrix AREs. Then, for any element $(\Lambda^\infty, K_Z^\infty) \in {\cal P}^\infty$ both  summands in (\ref{per_unit_1}) converge.  
 This establishes the characterization of  the right hand side of   (\ref{ll_3}). (1)-(3) follow from the asymptotic properties of the Kalman-filter (due to the stabilizability and detectability conditions). 
 \end{proof}

\ \

\begin{conclusion} Degenerate version of   Theorem~\ref{lem_cov} for feedback code of Definition~\ref{def_rem:cp_1}, i.e.,  $(s, 2^{n R},n)$,   $n=1,2, \ldots$  \\
\label{con_1}
(a) The  characterization of feedback capacity $C^{fb,o}(\kappa,s)$ of the AGN channel (\ref{g_cp_1}) driven by a noise $V^n$   of Definition~\ref{def_nr_2},   for the code of Definition~\ref{def_rem:cp_1}, i.e.,  $(s, 2^{n R},n)$,   $n=1,2, \ldots$, is a degenerate case of Theorem~\ref{lem_cov}, and corresponds to $\Sigma_t=\Sigma_t^s, t=1, \ldots, \Sigma_1=\Sigma_1^s=0$. In particular, since Theorem~\ref{lem_cov} characterizes $C^{fb,o}(\kappa)$ for all initial data $\Sigma_1 \succeq 0$, then it includes $\Sigma_1=\Sigma_1^s=0$, and  it follows that $C^{fb,o}(\kappa)= C^{fb,o}(\kappa,s)$, where $C^{fb,o}(\kappa,s)$  independent of the initial state $S_1^s=s$.\\ 
(b) The   maximal information rate of  \cite[Theorem~7 and Corollary~7.1]{yang-kavcic-tatikonda2007}, i.e., of Case II) formulation,  should be read with caution, because the condition of Theorem~\ref{thm_ric}.(1) are required for convergence. Similarly, the  characterization of feedback capacity of \cite[Theorem~6.1]{kim2010} which correspond to Case II) formulation,  violates Theorem~\ref{thm_ric}.(1), because  is  states that 
   a zero variance of the innovations process is optimal,  i.e., $K_Z^\infty=0$. Consequently, subsequent papers that build on  \cite{kim2010} to derive additional results, such as, \cite{liu-han2019,gattami2019,ihara2019,li-elia2019},  should be read with caution.
%
%
%
\end{conclusion}

We apply Theorem~\ref{lem_cov} to obtain  $C^{fb,o}(\kappa)$ of the  AR$(a,c), a\in (-\infty,\infty),  c \in (-\infty,\infty)$ noise.

\ \

\begin{corollary}
\label{rem_as}
Consider the AR$(a,c), a \in (-\infty,\infty),c \in (-\infty,\infty)$ noise of Example~\ref{ex_1_1_n}.(a).\\
Define the set 
\begin{align}
{\cal P}^\infty \tri & \Big\{ (\Lambda^\infty, K_Z^\infty)\in (-\infty, \infty)\times [0,\infty): \nonumber \\
& \mbox{(i) $c \in (-\infty,\infty), a\in (-1,1), c\neq a$},    \nonumber \\
& \mbox{(ii)   the pair $\{A, C(\Lambda^\infty)\}\tri \{c, \Lambda^\infty+ c-a\}$ is detectable}, \nonumber \\
&\mbox{(ii) the pair $\{A^*(\Lambda^\infty, K_Z^\infty), G B^{*,\frac{1}{2}}(K_Z^\infty)\}$  is stabilizable, where } \nonumber  \\
& A^*(\Lambda^\infty, K_Z^\infty) \tri \ c-K_W \big(K_W+K_Z^\infty\big)^{-1}\big(\Lambda^\infty+c-a\big), \hso G B^{*,\frac{1}{2}}(K_Z^\infty)\tri \Big(K_W- \big( K_W\big)^2 \big(K_W+K_Z^\infty\big)^{-1}\Big)^{\frac{1}{2}}
\Big\}. \nonumber
\end{align}
Then 
\begin{align}
 {C}^{fb,o}(\kappa)
=\sup_{  \big(\Lambda^\infty, K_{Z}^\infty\big)  \in {\cal P}^\infty(\kappa) }  \frac{1}{2} 
 \log\Big( \frac{ \Big(\Lambda^\infty +c-a\Big)^2K^\infty  +K_{W}+ K_{Z}^\infty     }{K_{W}}\Big)= C^{fb,o}(\kappa,s), \hso \forall s \label{cost_ih}
\end{align}
where  
\begin{align}
{\cal P}^\infty&(\kappa)\tri \Big\{(\Lambda^\infty, K_Z^\infty)\in {\cal P}^\infty:  \hso  \big(\Lambda^\infty\big)^2 K^\infty +K_{Z}^\infty  \leq \kappa, \nonumber \\
& K^\infty\geq 0  \hso \mbox{is the unique and stabilizing solution of}\nonumber  \\
&K^\infty= \big(c\big)^2 K^\infty  +K_{W} -\Big(c  K^\infty\big(\Lambda^\infty + c-a \big)+ K_{W}   \Big)^2 \Big(  K_{W}+ K_{Z}^\infty + \big(\Lambda^\infty + c-a \big)^2 K^\infty     \Big)^{-1}      \Big\}  \label{ARE_AR} 
\end{align}
 provided there exists $\kappa\in [0,\infty)$ such that the set ${\cal P}^\infty(\kappa)$ is non-empty.\\
 Moreover, the maximum element $(\Lambda^\infty, K_Z^\infty) \in {\cal P}^\infty(\kappa)$, is such that, \\
(1)  it induces asymptotic stationarity of the corresponding, input and innovations  processes, \\
(2) if $V^n, n=1, 2, \ldots$ is asymptotically stationary, then it  induces asymptotic stationarity of the corresponding, input and output processes, and \\
(3) for (i) and (ii), $C^{fb,o}(\kappa)$ and  $C^{fb,o}(\kappa,s)$ are  independent of $\Sigma_1 \geq 0$ and $s$, respectively, and the following identities hold.
\bea
 C^{fb,o}(\kappa)=C^{fb,o}(\kappa,s)=C^{fb,S,o}(\kappa,s), \hso \forall s \label{eqties}
 \eea
\end{corollary}
\begin{proof} The first part  is an application of Theorem~\ref{lem_cov},  Lemma~\ref{lem_pr_are_ar}, and Lemma~\ref{lem_pr_are_arac}. (1)-(3) are due to the convergence properties of the Kalman-filter (due to the stabilizability and detectability conditions). It remains to show (\ref{eqties}). The equality $C^{fb,o}(\kappa)=C^{fb,o}(\kappa,s), \forall s$ holds by  Conclusion~\ref{con_1}.(a). The last equality holds, because  for the 
AR$(a,c), a \in (-\infty,\infty),c \in (-\infty,\infty)$ noise, if the initial state $S_1=S_1^s=s$ is known to the encoder and the decoder, then Conditions 1 of Section~\ref{sect:motivation} holds, and in addition Condition 2 holds, as easily verified from the equations (\ref{ex_1_5}), (\ref{ex_1_6}).
 \end{proof}

\ \

%
%

\begin{remark}
From Corollary~\ref{rem_as} we obtain the degenerate cases, AR$(c), c \in (-\infty,\infty)$, noise i.e., setting  $a=0$. 
 The various implications of the detectability and stabilizability conditions for the AR$(c), c \in (-\infty,\infty)$ noise are found in  \cite[see Theorem~III.1 and Lemma~III.2]{charalambous-kourtellaris-loykaIEEEITC2019}. The complete analysis of the corresponding $C^{fb,o}(\kappa,s)$  is found in \cite{charalambous-kourtellaris-loykaIEEEITC2019}, and states that for stable AR$(c)$, and time-invariant strategies, then feedback does not increase capacity. 
\end{remark}

\section{Sequential Characterization of $n-$FTFI Capacity for Case II)  Formulation}
\label{sect_POSS}
In this section we consider Case II) formulation, and we derive  the  characterization of feedback capacity, $C_n^{fb,S}(\kappa,s)$,  of the AGN channel (\ref{g_cp_1}) driven by a noise $V^n$   of Definition~\ref{def_nr_2}, i.e.,   for the code of Definition~\ref{def_rem:cp_1},  $(s, 2^{n R},n)$,   $n=1,2, \ldots$,  when Conditions 1 and 2 of  Section~\ref{sect:motivation} hold.  
%
%

\label{sect:iss_r}

\begin{definition} AGN channels driven by noise with invertible  PO-SS  realizations\\
\label{ass_SSUN}
The PO-SS realization of the noise of Definition~\ref{def_nr_2}  is called invertible if it satisfies  the  condition: \\
(A1) Given  the initial state $S_1=S_1^s=s$,   the noise  $V^{t-1}$ uniquely specifies the state $S^t$, for $t=1, \ldots, n$,  and vice versa. 
\end{definition}

\begin{corollary} Characterization of $n-$FTFI Capacity for Case II) formulation     \\ 
\label{cor_kim}
Consider the AGN channel (\ref{g_cp_1}) driven by a noise $V^n$   of Definition~\ref{ass_SSUN}, and  the code of Definition~\ref{def_rem:cp_1},  $(s, 2^{n R},n)$,   $n=1,2, \ldots$, that is, Conditions 1 and 2 of  Section~\ref{sect:motivation} hold.\\
Define the $n-$FTFI Capacity for a fixed initial state $S_1=S_1^s=s$, by 
\begin{align}
C_n^{fb,S}(\kappa,s) =  \sup_{{\cal P}_{[0,n]}^s(\kappa)}H^P(Y^n|s)- H(V^n|s) \label{FTFI_CIS_1}
\end{align}
where the set ${\cal P}_{[0,n]}^s(\kappa)$ is defined by 
\begin{align}
{\cal P}_{[0,n]}^s(\kappa) \tri \Big\{{P}_t(dx_t|x^{t-1}, y^{t-1},s), t=1,\ldots,n: \frac{1}{n} {\bf E}_{s}^{{P}}\Big( \sum_{t=1}^n \big(X_t\big)^2\Big) \leq \kappa    \Big\}\label{FTFI_CIS_2}
\end{align}
and where ${\bf E}_{s}^{{P}}$ means $S_1=S_1^s=s$ is fixed, and the joint distribution depends on the elements of ${\cal P}_{[0,n]}^s(\kappa)$.\\
Then the following hold.\\
(a) The $n-$FTFI capacity, for a fixed $S_1=s$  is characterized  by 
\begin{align}
C_n^{fb,S}(\kappa,s) =&  \sup_{\overline{\cal P}_{[0,n]}^{s,M}(\kappa)}\sum_{t=1}^n\Big\{H^{\overline{P}^M}(Y_t|Y^{t-1},s)-H(V_t|V^{t-1},s)\Big\} \label{is_dis_1_a}\\
 =&\sup_{\overline{\cal P}_{[0,n]}^{s,M}(\kappa)}\sum_{t=1}^n\Big\{H^{\overline{P}^M}(Y_t - {\bf E}\big\{Y_t\Big||Y^{t-1},s\big\}|s)-H(V_t-{\bf E}\big\{V_t\Big|V_t^{t-1},s\big\}|s)\Big\}
\end{align}
where the $\overline{\cal P}_{[0,n]}^{s,M}(\kappa)$ is defined by 
\begin{align}
&\overline{\cal P}_{[0,n]}^{s,M}(\kappa) \tri \Big\{\overline{P}_t^M(dx_t|s_{t}, y^{t-1},s), t=1,\ldots,n: \frac{1}{n+1} {\bf E}_{s}^{\overline{P}^M}\Big( \sum_{t=1}^n \big(X_t\big)^2\Big) \leq \kappa    \Big\} \label{is_dis_1}
\end{align}
and where (\ref{g_cp_3}) is respected,  $\overline{P}_t^M(dx_t|s_{t}, y^{t-1},s)$,  is  conditionally Gaussian, 
with linear conditional mean and nonrandom conditional covariance,  given by\footnote{The notation $S_t=S_t^s, t=2, \ldots, n$ means this sequence is generated from (\ref{real_1a}), when the initial state is fixed, $S_1=S_1^s=s$.}  
\begin{align}
&{\bf E}\Big\{X_t\Big|S_t^s,Y^{t-1}, S_1^s=s\Big\}=\left\{ \begin{array}{lll} \Lambda_t \Big(S_{t}^s - {\bf E}\Big\{S_t^s\Big|Y^{t-1}, S_1^s=s\Big\}\Big) & \mbox{for} & t=2, \ldots, n\\
0, & \mbox{for} & t=1,
\end{array} \right.  \label{mean_1_new}  \\
&K_{X_t|S_{t}^s, Y^{t-1},S_1^s=s}\tri cov\big(X_t, X_t\Big|S_t^s, Y^{t-1},S_1^s=s\big)= K_{Z_t} \succeq 0, \hso t=1, \ldots, n.\label{var_1_neq}
\end{align}
and  $H^{\overline{P}}(Y_t|Y^{t-1},s)$ is evaluated with respect to 
the probability distribution ${\bf P}_t^{\overline{P}^M}(dy_t|y^{t-1},s)$, defined by  
\begin{align}
 {\bf P}_t^{\overline{P}^M}(dy_t|y^{t-1},s) = \int {\bf P}_t(dy_t|x_t,s_t)\; {\bf P}_t^{\overline{P}^M}(dx_t|s_t, y^{t-1},s) \; {\bf P}_t^{\overline{P}^M}(ds_t|y^{t-1},s), \hso t=1, \ldots, n.  \label{eq_a_2_thm_new}
 \end{align}
(b) Define the conditional means and conditional covariance for a fixed $S_1=S_1^s=s$, by 
\begin{align}
K_{t}^{s} \tri & cov\big(S_t^s, S_t^s\Big|Y^{t-1},S_1^s=s) = {\bf E}_s^{\overline{P}^M}\Big\{\Big(S_t^s- \hat{S}_{t}^{s}\Big)\Big(S_t^s-\hat{S}_{t}^{s} \Big)^T\Big\}, \label{cond_GK_1_n_new}\\
 \hat{S}_{t}^{s} \tri & {\bf E}_s^{\overline{P}^M}\Big\{S_t^s\Big|Y^{t-1}, S_1^s=s\Big\},  \hso  t=2, \ldots, n, \hso K_1^{s} \tri  cov\big(S_1^s, S_1^s|S_1^s=s) =0, \hso \hat{S}_1^{s} \tri s. \label{cond_GK_2_n_new}
\end{align} 
The optimal channel input distribution of part (a) is induced by a jointly Gaussian process   $X^n$, with a realization given   by   
\begin{align}
&X_t=  \Lambda_{t}\Big(S_{t}^s- \widehat{S}_{t}^s\Big) + Z_t,\hso X_1=Z_1, \hso  t=2, \ldots, n, \label{cp_13_alt} \\
&  Z_t \in N(0, K_{Z_t}) \hso \mbox{independent of} \hso (S_1,X^{t-1},V^{t-1}, Y^{t-1}), \hso t=1, \ldots, n, \label{cp_8_al} \\
&Y_t= \Lambda_{t}\Big(S_{t}^s- \widehat{S}_{t}^s\Big) + Z_t +V_t, \hso t=1, \ldots, n, \label{cp_15_alt}\\
&\hso \: = \Lambda_{t}\Big(S_{t}^s- \widehat{S}_{t}^s\Big) +C_tS_{t}^s +N_t W_t  + Z_t,\\
&\frac{1}{n}  {\bf E}_{s}^{\overline{P}^M} \Big\{\sum_{t=1}^{n} (X_t)^2\Big\}=\frac{1}{n}\sum_{t=1}^n \Big( \Lambda_{t} K_{t}^s \Lambda_t^T +K_{Z_t}\Big) \leq \kappa  \label{cp_9_al}
\end{align} 
where $\Lambda_t$ is nonrandom. \\
The conditional means and conditional covariance $\widehat{S}_t^s$ and $K_t^s$ are given by  the generalized Kalman-filter, as follows equations.\\
(i)   $\widehat{S}_t^s$ satisfies the Kalman-filter recursion
\begin{align}
&\widehat{S}_{t+1}^s=A_t \widehat{S}_{t}^s+ F_t(K_t^s)I_t^s, \hso \widehat{S}_1=s, \\
& F_t(K_t^s) \tri \Big( A_t  K_{t}^s \big(\Lambda_t + C_t \big)^T +  B_t K_{W_t}N_t^T\Big)
\Big\{N_t K_{W_t}N_t^T+ K_{Z_t} + \big(\Lambda_t + C_t \big) K_{t}^s \big(\Lambda_t + C_t \big)^T \Big\}^{-1}, \\
&I_t^s \tri Y_t - C_t \widehat{S}_{t}^s= \big(\Lambda_t+C_t\big)\big(S_{t}^s- \widehat{S}_{t}^s\big) +N_tW_t + Z_t, \hso t=1, \ldots, n, \\
&I_t^s \in  N(0, K_{I_t^s}), \hso t=1, \ldots, n \hso \mbox{an orthogonal innovations process, i.e., $I_t^s$ is independent of}\nonumber \\
&\hst \hst \mbox{ $I_k^s$, for all $t \neq k$, and ${I}_t^s$ is independent of  $Y^{t-1}$},  \label{inn_po_2_nn_new}\\
&K_{Y_t|Y^{t-1},s}=K_{I_t^s} \tri cov\big(I_t,I_t|S_1^s=s\big)=  \Big(\Lambda_t +C_t\Big)K_t^s \Big(\Lambda_t +C_t\Big)^T + N_t K_{W_t}N_t^T + K_{Z_t}. \label{inno_PO_new} 
\end{align}
(ii) The error ${E}_t^s \tri S_t^s- \widehat{S}_t^s$ satisfies the recursion 
\begin{align}
{E}_{t+1}^s = &F_t^{CL}(K_t^s)E_t^s + \Big(B_t - F_t(K_t^s)N_t\Big)W_t- F_t(K_t^s) Z_t, \hso E_1^s= S_1^s-\widehat{S}_1^s=0, \hso t=1,\ldots, n, \\
F_t^{CL}(& K_t^s)\tri A_t- F_t(K_t^s)\Big(\Lambda_t+C_t\Big).
\end{align}
(iii)  $K_t^s={\bf E}\big\{E_t^s \big(E_t^s\big)^T\big\}$ satisfies the generalized DRE 
\begin{align}
K_{t+1}^s= &A_t K_{t}A_t^T  + B_tK_{W_t}B_t^T -\Big( B_tK_{W_t}N_t^T  + A_t  K_{t}^s\big(\Lambda_t + C_t \big)^T\Big)  \Big\{N_t K_{W_t} N_t^T+K_{Z_t}\nonumber \\
&+ \big(\Lambda_t + C_t \big) K_{t}^s\big(\Lambda_t + C_t \big) ^T\Big\}^{-1}\Big( B_t K_{W_t}N_t^T + A_t  K_{t}^s\big(\Lambda_t + C_t \big)\Big)^T,\hso   K_t^s\succeq 0, \hso K_1^s=0, \hso t=1, \ldots, n.   \label{rde_ykt}
\end{align}
(c) The characterization of the $n-$FTFI capacity of part (a)  is  
\begin{align}
 {C}_n^{fb,S}(\kappa,s)
  =& \sup_{  \big(\Lambda_{t}, K_{Z_t}\big), t=1, \ldots, n: \hso \frac{1}{n}  {\bf E}_{s}\big\{\sum_{t=1}^n \big(X_t\big)^2\big\}\leq \kappa }  \sum_{t=1}^n \log \frac{K_{Y_t|Y^{t-1},s}}{K_{V_t|V^{t-1},s}} \label{cp_11_alt_1}\\
 =& \sup_{  \big(\Lambda_{t}, K_{Z_t}\big), t=1, \ldots, n: \hso \frac{1}{n}  {\bf E}_{s}\big\{\sum_{t=1}^n \big(X_t\big)^2\big\}\leq \kappa }  \sum_{t=1}^n\Big\{H(I_t^s)-H(N_t W_t)\Big\}\\ 
=&  \sup_{  \big(\Lambda_{t}, K_{Z_t}\big), t=1, \ldots, n: \hso \frac{1}{n}\sum_{t=1}^n \big(\Lambda_t K_t^s \Lambda_t^T+K_{Z_t}\big)\leq \kappa }\frac{1}{2} 
\sum_{t=1}^n \log\Big( \frac{ \Big(\Lambda_t +C_t\Big)K_t^s \Big(\Lambda_t +C_t\Big)^T + N_t K_{W_t}N_t^T + K_{Z_t}     }{N_t K_{W_t}N_t^T}\Big).
 \label{cp_12_alt_1}
  \end{align}
  and the statements of part (b) hold.
\end{corollary}
\begin{proof}  See Appendix~\ref{sect:app_cor_kim}.
\end{proof}

\ \

\begin{remark} Comments on the per unit time limit of ${C}_n^{fb,S}(\kappa,s)$\\
(b) The asymptotic analysis of $C^{fb,o}(\kappa)$ and  $C^{fb,o}(\kappa,s)$ of  Section~\ref{sect:as_an}, i.e., based on  Definition~\ref{def_limit_1},    applies naturally to  Corollary~\ref{cor_kim}, by considering   $C^{fb,S,o}(\kappa,s)$.
\end{remark}

In the next remark we clarify the relation of Corollary~\ref{cor_kim} and the analysis of  \cite{yang-kavcic-tatikonda2007} and \cite{kim2010}. \\

\begin{remark} Relations of Corollary~\ref{cor_kim} and \cite{yang-kavcic-tatikonda2007,kim2010}\\
(a) The problem analyzed in \cite{yang-kavcic-tatikonda2007} is precisely $C_n^{fb,S}(\kappa,s)$, when the noise is stationary and Gaussian, i.e., it corresponds to Case II) formulation. 
Corollary~\ref{cor_kim} is derived in \cite{yang-kavcic-tatikonda2007} for the degenerate case of a time-invariant realization of the noise $V^n$, i.e., of  Definition~\ref{ass_SSUN}. However, the asymptotic analysis of  \cite[Section~VI]{yang-kavcic-tatikonda2007} should be read with caution, because it  did not account for the necessary and/or sufficient conditions for convergence of the sequence $K_t^s, t=1,2,\ldots$ generated by the  time-invariant version of the  generalized DRE (\ref{rde_ykt})  i.e.,  $\lim_{n \longrightarrow \infty}K_n^s=K^\infty\succeq 0$, where $K^\infty \succeq 0$ is the unique and stabilizing solution of a corresponding generalized ARE.   \\
(b) The problem analyzed \cite{kim2010} that let to \cite[Theorem~6.1, $C_{FB}$]{kim2010}, is the per unit time limit of  $C_n^{fb,S}(\kappa,s)$, when the noise is stationary, two-sided or one-sided (asymptotically stationary) and Gaussian, i.e., it corresponds to Case II) formulation. The characterization of feedback capacity presented in \cite[Theorem~6.1, $C_{FB}$]{kim2010} presupposed the following hold ((i)-(iii) are also assumed in \cite[Section~VI]{yang-kavcic-tatikonda2007}). \\
(i) The feedback code is Definition~\ref{def_rem:cp_1}, i.e.,  $(s, 2^{n R},n)$.  \\
(ii) The noise is time-invariant and stable, and the PO-SS realization of the noise is  invertible, as presented in  Definition~\ref{ass_SSUN}.  \\
(iii) The definition of rate is $C^{fb,S,o}(\kappa,s)$, with supremum and per unit time limit interchanged, and the supremum taken over using time-invariant channel input distributions.  \\
(iv) the innovations covariance of the channel input process is zero, i.e., $K_{Z_t}=K_Z=0, \forall t$.\\ 
Items (i)-(iv)  are confirmed from \cite[Lemma~6.1]{kim2010} (and comments above), which is used to derive \cite[Theorem~6.1, $C_{FB}$]{kim2010}. \\
However, the characterization of feedback capacity in \cite[Theorem~6.1, $C_{FB}$]{kim2010}  should be read with caution, because the stabilizability condition is violated, due  to the claim by the author that $K_Z=0$ optimal.  By   Theorem~\ref{thm_ric}.(1), for the choice $K_Z= 0$,  the only unique and stabilizing solution of the generalized ARE presented in \cite[Theorem~6.1, $C_{FB}$]{kim2010}, is the zero solution, and hence  $C_{FB}=0$.\\
Prior literature \cite{liu-han2019,gattami2019,li-elia2019} should be read with caution, because many of the results are developed using \cite{kim2010}.  \\
 The above technical matters are discussed extensively  in \cite{charalambous-kourtellaris-loykaIEEEITC2019},  for the case of the AR$(c), c\in (-\infty, \infty)$, where it is also  shown that feedback does not increase capacity for $c\in (-1,1)$, i.e., for the stationary AR$(c)$ noise. 
\end{remark}

\section{Conclusion}
New equivalent sequential characterizations of the  Cover and Pombra \cite{cover-pombra1989}  ``$n-$block'' feedback capacity  formulas are derived using time-domain methods,  for  additive Gaussian noise (AGN) channels driven by nonstationary  Gaussian noise. The new feature of the equivalent characterizations are the representation of  the  optimal channel input process by a {\it sufficient statistic and  Gaussian orthogonal innovations process}. The sequential characterizations of the $n-$block feedback capacity formula are expressed  as a functional of   two generalized matrix difference Riccati equations (DRE) of filtering theory of Gaussian systems. The asymptotic analysis of the per unit time limit of the $n-$block'', called  feedback capacity, with limit and maximization operations interchanged, is also analyzed for time-invariant channel input distributions, using the tools from the theory of generalized matrix Riccati equations. 

From the analysis and derivation of the new sequential characterizations of feedback capacity follows that prior  analysis and characterizations of feedback capacity, such as, \cite{kim2010,liu-han2019,gattami2019,li-elia2019}, do not address  the  Cover and Pombra \cite{cover-pombra1989} feedback capacity problem, because the code definition and noise assumptions in \cite{kim2010,liu-han2019,gattami2019,li-elia2019} (even under the restriction of stationary noise)  are fundamentally different from those in \cite{cover-pombra1989}.   The paper  clarifies the several points of confusion,   which are  recurrent in the literature on these issues.

\section{Appendix}
\label{appendix}

\subsection{Proof of Theorem~\ref{thm_FTFI}}
\label{sect:app_thm_FTFI}

(a) Consider an element of $\overline{\cal E}_{[0,n]}(\kappa)$. Then the conditional entropies $H^{\overline{e}}(Y_t|Y^{t-1}), t=1,\ldots, n$ are defined, provided  the conditional distributions of $Y_t$ conditioned  on $Y^{t-1}$, i.e., ${\bf P}_t^{\overline{e}}(dy_t|y^{t-1})$,  for $t=1, \ldots, n$, are determined. By the reconditioning property of conditional distributions, then 
\begin{align}
 {\bf P}_t^{\overline{e}}(dy_t|y^{t-1}) = &\int {\bf P}_t^{\overline{e}}(dy_t|y^{t-1}, w, v^{t-1})\; {\bf P}_t^{\overline{e}}(dw,dv^{t-1}|y^{t-1}), \hso t=0, \ldots, n \\
=& \int {\bf P}_t(dy_t|\overline{e}_t(w,v^{t-1},y^{t-1}), v^{t-1})\;  {\bf P}_t^{\overline{e}}(dw,v^{t-1}|y^{t-1}),  \hst \mbox{by (\ref{kernel_2}), (\ref{enc_1})}. \label{eq_a_1}
 \end{align}
Hence,  (\ref{eq_a_1_thm}) is shown. Similarly, consider an element of $\overline{\cal P}_{[0,n]}(\kappa)$. Then the conditional entropies $H^{\overline{P}}(Y_t|Y^{t-1}), t=1,\ldots, n$ are defined, provided  the conditional distributions of $Y_t$ conditioned  on $Y^{t-1}$, i.e.,  ${\bf P}_t^{\overline{P}}(dy_t|y^{t-1})$ for $t=1, \ldots, n$, are determined. By (\ref{NCM-A.D_CD_C_CD_n}) and (\ref{NCM-A.D_CD_C_CD2_n}), then   (\ref{eq_a_2_thm}) is obtained.  Since $\overline{\cal E}_{[0,n]}(\kappa)\subseteq \overline{\cal P}_{[0,n]}(\kappa)$ it then follows the inequality  (\ref{dp_1_new}). \\
(b)  This part follows by  the maximum entropy principle of Gaussian distributions. That is, under the restriction (\ref{g_cp_3}), then   a conditional Gaussian  element of $\{\overline{P}(dx_t|v^{t-1},y^{t-1}), t=1, \ldots, n\}\in 
\overline{\cal P}_{[0,n]}(\kappa)$,  with linear conditional mean and nonrandom conditional covariance induces a  jointly Gaussian  distribution of the process   $(X^n,Y^n)$, such that the marginal distribution of $Y^n$ is jointly Gaussian. Below, we provide  alternative proof that   uses  the Cover and Pombra characterization of the $n-$FTFI capacity, given by (\ref{cp_6})-(\ref{cp_12}).  Consider (\ref{cp_6}) and   define the process
\begin{align}
\hso Z_1\tri& \overline{Z}_{1}- {\bf E}\Big\{ \overline{Z}_1 \Big\}, \\ 
Z_t\tri& \overline{Z}_{t}- {\bf E}\Big\{ \overline{Z}_t    \Big|X^{t-1}, V^{t-1}, Y^{t-1}\Big\},\hso t=2, \ldots, n,\\
=&\overline{Z}_t- {\bf E}\Big\{\overline{Z}_{t}\Big| V^{t-1}, Y^{t-1}\Big\}, \hst \mbox{since $X^{t-1}$ is uniquely defined by $(V^{t-1}, Y^{t-1})$.}\label{orthogonal_11_n}
\end{align} 
Then $Z_t$ is a Gaussian orthogonal innovations process, i.e., $Z_t$ is  independent of  $(X^{t-1}, V^{t-1}, Y^{t-1})$, for $t=2, \ldots, n$, and ${\bf E}\big\{Z_t\big\}=0,$ for $t=1, \ldots, n$. 
By (\ref{cp_6}), we re-write $X_t, t=1, \ldots, n$ as,
\begin{align}
X_t =&  \sum_{j=1}^{t-1} B_{t,j} V_j +\overline{Z}_t, \hso t=1, \ldots, n,\\
=& \sum_{j=1}^{t-1} B_{t,j} V_j +{\bf E}\Big\{\overline{Z}_{t}\Big| V^{t-1}, Y^{t-1}\Big\}  + Z_t, \hst \mbox{by (\ref{orthogonal_11_n})}\\
\sr{(a)}{=}&  \sum_{j=1}^{t-1} B_{t,j} V_j+ \overline{\Gamma}_t \left( \begin{array}{c}  {\bf V}^{t-1}\\ {\bf Y}^{t-1}\end{array} \right)+ Z_t, \hst \mbox{for some $\overline{\Gamma}_t$ nonrandom}\\
=&  \sum_{j=1}^{t-1} \Gamma_{t,j}^1 V_j+ \sum_{j=1}^{t-1}\Gamma_{t,j}^2   Y_j+ Z_t, \hst \mbox{for some $(\Gamma_{\cdot,\cdot}^1, \Gamma_{\cdot,\cdot}^2)$}\\
=&\Gamma_{t}^1 {\bf V}^{t-1}+ \Gamma_{t}^2   {\bf Y}^{t-1}+ Z_t, \hst \mbox{by definition} \label{eqn_in_1}
\end{align}
where $(a)$ is due to  the by joint Gaussianity of $({Z}^n, X^n, Y^n)$. From (\ref{eqn_in_1}) and  the  independence of $Z_t$ and  $(X^{t-1}, V^{t-1}, Y^{t-1})$, for $t=2, \ldots, n$,  it then  follows (\ref{mean_1}),  and also  (\ref{var_1}).\\
(c)  The statements follow directly from the representation of part (b), while the independence of $Z^n$ and $V^n$ is due to the code definition, i.e.,  Definition~\ref{def_code}.(iv).\\
(d) The statement follows from (a)-(c), and  (\ref{inn_a_intr})-(\ref{inn_intr}).

\subsection{Proof of Proposition~\ref{pr_ex_2}}
\label{sect:app_pr_ex_2}
(a) The covariances of the realization of the ARMA$(a,c)$ noise of Example~\ref{ex_1_1_n}.(b) satisfy the recursions
\begin{align}
K_{S_{t+1}}= c^2 K_{S_t} + K_{W}, \hst  K_{S_t,V_t}= \Big(c-a\Big) K_{S_t}, \hst K_{V_t}=\Big(c-a\Big)^2K_{S_t}+K_W, \hst    \forall t \in {\mathbb Z}.
\end{align}
If  the  recursion $K_{S_{t+1}}= c^2 K_{S_t} + K_{W}$  is initiated at the stationary value $K_{S_1}=d_{11} = \frac{K_W}{1- c^2}$, then  $K_{S_{t+1}}=d_{11}, \forall t=2,3, \ldots,$, and hence  $S_t, \forall t \in {\mathbb Z}$ is stationary, which then implies stationarity of   $V_t, \forall t \in {\mathbb Z}$. Hence, if (\ref{ex_2_6}) holds then  $(V_t,S_t), \forall t \in {\mathbb Z}$ is   stationary. By simple calculations it then follows (\ref{ex_2_7}). Similarly for the one-sided ARMA$(a,c)$.     (b) By the above covariances, for all $K_{S_1}\geq 0$, then  $\lim_{n\longrightarrow \infty} K_{S_n}= K_{S}^\infty$, where  $K_S^\infty= c^2 K_S^{\infty} +K_W$, which then implies $K_S^\infty= d_{11}$. Similarly, $\lim_{n\longrightarrow \infty} K_{S_n,V_n}= K_{S,V}^\infty=d_{12}$, $\lim_{n\longrightarrow \infty} K_{V_n}= K_{V}^\infty=d_{22}$. (c) (\ref{arma_s1}), (\ref{arma_s2}) follow from Lemma~\ref{lemma_POSS}, by replacing the conditioning information $V^{t-1}$ by $(V_0, V^{t-1})$ in   (\ref{cond_GK_1}) and (\ref{cond_GK_1}) . By mean-square-estimation, the initial data are, 
\begin{align}
\hat{S}_{1}=& {\bf E}\Big\{S_1\Big| V_0\Big\}= {\bf E}\Big\{c S_0 +W_0\Big| V_0\Big\}\nonumber \\
=&{\bf E}\Big\{cS_0+W_0\Big\}+cov(cS_0+W_0,V_0)\Big\{cov(V_0,V_0)\Big\}^{-1}\Big(V_0 -{\bf E}\Big\{V_0\Big\}\Big)
=\Big(c d_{12}+K_W\Big)d_{22}^{-1} V_0, \\
\Sigma_1 =& cov(S_1,S_1\Big|V_0)=cov(S_1,S_1)- \Big\{cov(S_1,S_1)\Big\}^2 \Big\{cov(V_0,V_0)\Big\}^{-1}=d_{11} -d_{11}^2 d_{22}^{-1}
\end{align}
The last part is obvious.

\subsection{Proof of Corollary~\ref{cor_nftfic_s}}
\label{sect:app_cor_nftfic_s}
(a) Since we have assumed $S_1=S_1^s=s$ is fixed, and known to the encoder and the decoder, then  Theorem~\ref{thm_FTFI} still holds, by    replacing all conditional distributions,  expectations and entropies, by corresponding expressions with fixed   $S_1=S_1^s=s$. Hence, (\ref{ftfic_is_g}) is replaced by   (\ref{CP_F_new_nn}), and (\ref{Q_1_3_s1}) is replaced by   (\ref{Q_1_3_s1_is}) (since the code is allowed to depend on $S_1=S_1^s=s$). (b) From the PO-SS realization  of  Definition~\ref{def_nr_2} with $S_1=S_1^s=s$ fixed, it follows that  a necessary condition  for Conditions 1 of  Section~\ref{sect:motivation} to hold is (i).  The expression of entropy   (\ref{cond_ent}) is  easily obtained by invoking condition (i),  and   properties of conditional entropy. That is, $H(V_1|S_1^s=s)=H(C_1 S_1^s+ N_1W_1|S_1^s=s)=H(N_1 W_1|S_1=s)= H(N_1 W_1)$ by independence of $W_1$ and $S_1^s$, and $H(V_2|V_1,s)= H(V_2|V_1,S_1^s=s)= H(C_2 S_2^s +N_2 W_2|C_1 S_1^s+N_1W_1,S_1^s=s)=H(C_2 S_2^s +N_2 W_2|N_1W_1,S_1^s=s)=H(C_2 A_1 S_1^s + C_2 B_1 W_1 +N_2 W_2|N_1W_1,S_1^s=s)=H(N_2 W_2|N_1W_1,S_1^s=s)=H(N_2 W_2)$, etc. This completes the proof.

\subsection{Proof of Theorem~\ref{thm_SS}}
\label{sect:app_thm_SS}
(a) Clearly, (\ref{ftfic_is})-(\ref{PO_tr_3}), follow directly from Theorem~\ref{thm_FTFI}, and  the preliminary calculations, prior to the statement of the theorem. However, (\ref{ftfic_is})-(\ref{PO_tr_3}) can also be shown independently of Theorem~\ref{thm_FTFI}, by invoking the maximum entropy property of Gaussian distributions, as follows.  By Lemma~\ref{lemma_POSS}, then   $H(V^n)=\sum_{t=1}^nH(\hat{I}_t)$.  By the maximum entropy principle, then  $H(Y^n)$ is maximized if ${\bf P}_{Y^n}$ is jointly Gaussian, the average power constraint holds, and (\ref{g_cp_3}) is respected. By (\ref{PO_state_1}),  (\ref{PO_tr_1}), (\ref{PO_tr_2}),   if (\ref{PO_tr_1})-(\ref{PO_tr_3}) hold, then $(X^n,Y^n)$ is jointly Gaussian, and hence $H(Y^n)$ is maximized. This shows  (a). \\
(b) 
Step 1. By (\ref{PO_tr_2_b})  and (\ref{PO_tr_2_c}), an alternative  representation of $X^n$ to the one given in Theorem~\ref{thm_FTFI}, and  induced by (\ref{Q_1_3_s1}),   is 
\begin{align}
&X_t = \Gamma_t^1 \hat{S}_t + \Gamma_t^2 Y^{t-1} +Z_t, \hso  t=1, \ldots, n, \label{inp_PO_1}\\
&Z_t \hso \mbox{satisfies (\ref{cp_8_al_n})}. \label{inp_PO_1_a}
\end{align}
for some nonrandom $(\Gamma_\cdot^1, \Gamma_\cdot^2)$.
Upon substituting (\ref{inp_PO_1}) into the channel output $Y^n$  we have 
\begin{align}
Y_t =& \Gamma_t^1 \hat{S}_t + \Gamma_t^2 Y^{t-1} +Z_t + V_t,\hso  t=1, \ldots, n \label{inp_PO_2}\\
=& \Big(\Gamma_t^1 +C_t\Big)  \hat{S}_t + \Gamma_t^2 Y^{t-1} +Z_t  + \hat{I}_t, \hso \mbox{by  (\ref{inn_po_1}).} \label{inp_PO_2_a}
\end{align}
The  right hand side of (\ref{inp_PO_2_a}) is driven by two independent processes, $Z_t, t=1, \ldots, n$ and $\hat{I}_t, t=1, \ldots, n$, which are also mutually independent. Further, the right hand side of (\ref{inp_PO_2_a}) is a linear function of a state process $\hat{S}_t, t=1, \ldots, n$, which satisfies the recursion (\ref{kal_fil_noise}):
\begin{align}
\hat{S}_{t+1}=A_{t} \hat{S}_{t}+ M_{t}(\Sigma_t)  \hat{I}_t, \hso \hat{S}_{1}=\mu_{S_1}, \label{kal_fil_noise_n} 
\end{align}
Note that the right hand side of (\ref{kal_fil_noise_n}) is driven by the orthogonal process  $\hat{I}_t$, which is  independent of $V^{t-1}$ and hence of $\hat{S}_t$, and also independent of $Y^{t-1}$. By (\ref{cp_8_al_n}), $Z_t$ is independent of $Y^{t-1}$ and of  $\hat{S}_t$. By (\ref{inp_PO_2_a}) and (\ref{kal_fil_noise_n}), it follows that $\widehat{\hat{S}}_t, t=1, \ldots, n$ satisfies a generalized Kalman-filter recursion, similar to that of  Lemma~\ref{lemma_POSS}. Hence, the entropy $H(Y^n)$ can be computed using the innovations process of $Y^n$, as in Lemma~\ref{lemma_POSS}. Define the orthogonal Gaussian innovations process ${I}^n$ of $Y^n$ by 
\begin{align}
{I}_t \tri & Y_t - {\bf E}\Big\{Y_t\Big|Y^{t-1}\Big\}, \hst t=1, \ldots, n\label{inp_PO_3_a}\\
 =&\Big(\Gamma_t^1 + C_t\Big)\Big(\hat{S}_t- \widehat{\hat{S}}_t\Big) + \hat{I}_t- {\bf E}\Big\{\hat{I}_t\Big|Y^{t-1}\Big\} +Z_t, \hso \mbox{by  (\ref{inp_PO_2_a}).}    \label{inp_PO_3}\\
 =&\Big(\Gamma_t^1 + C_t\Big)\Big(\hat{S}_t- \widehat{\hat{S}}_t\Big) + \hat{I}_t +Z_t, \hso \mbox{by $\hat{I}_t$ indep. of $Y^{t-1}$, ${\bf E}\Big\{\hat{I}_t\Big\}=0$  }    \label{inp_PO_3_n}
 \end{align}
The entropy of $Y^n$ is computed as follows.
\begin{align}
H(Y^n)=& \sum_{t=1}^n H(Y_t|Y^{t-1})\\
=& \sum_{t=1}^n H(I_t|Y^{t-1}), \hso \mbox{by (\ref{inp_PO_3_a}) and a property of conditional entropy} \\
=& \sum_{t=1}^n H(I_t), \hso \mbox{by orthogonality of $I_t$ and $Y^{t-1}$.} \label{inp_PO_4}
\end{align}
 By (\ref{inp_PO_3_n}) the Gaussian innovations process $I^n$ does not depend on the strategy $\Gamma^2$, and consequently by (\ref{inp_PO_4}) the entropy $H(Y^n)$ does not depend on the strategy $\Gamma^2$. \\
Step 2.  Let $
g_t(Y^{t-1})\tri  \Gamma_t^2 Y^{t-1},  t=1, \ldots, n.$ By (\ref{inp_PO_1}) and (\ref{inp_PO_1_a}), it then follows,
\begin{align}
\frac{1}{n}{\bf E}&\Big\{\sum_{t=1}^n (X_t)^2\Big\}
=\frac{1}{n}{\bf E}\Big\{\sum_{t=1}^n \Big( \Gamma_t^1 \hat{S}_t+g_t(Y^{t-1})+Z_t\Big)^2\Big\}\\ 
=&\frac{1}{n}{\bf E}\Big\{\sum_{t=1}^n \Big( \Gamma_t^1 \hat{S}_t+g_t(Y^{t-1})\Big)^2\Big\}+ \sum_{t=1}^n K_{Z_t}, \hso \mbox{by indep. of $Z_t$ and $(V^{t-1},\hat{S}_t, Y^{t-1})$}.\label{cost_1_g_n} 
\end{align}
By mean-square estimation theory, then the choice of $g(\cdot)$ that minimizes the right hand of (\ref{cost_1_g_n}) is
\begin{align}
 g_t(Y^{t-1})=g_t^*(Y^{t-1})= -\Gamma_t^1 {\bf E}\Big\{\hat{S}_t\Big|Y^{t-1}\Big\}=-\Gamma_t^1\widehat{\hat{S}}_t, \hso t=1, \ldots, n.
\end{align} 
  Hence, $\Gamma_t^2=-\Gamma_t^1, \forall t$. Let  $ \Lambda_t\tri \Gamma_t^1, \forall t$, and substitute into the recursion (\ref{inp_PO_2_a}), to obtain  (\ref{cp_15_alt_n_1}), and into the average power (\ref{cost_1_g_n}), to obtain  (\ref{cp_9_al_n}). Hence, the derivation of (\ref{cp_13_alt_n})-(\ref{cp_9_al_n}) is completed. \\
It then follows that (\ref{kf_m_1})-(\ref{kf_m_4_a}) are the generalized Kalman-filter recursions, of estimating the  new state process $\hat{S}_t, t=1, \ldots, n$ that  satisfies recursion (\ref{kal_fil_noise_n}),  from the channel output process $Y_t$ that satisfies the recursion (\ref{cp_15_alt_n_1}).\\
  (c) By parts (a), (b), and  the entropy of Gaussian RVs, upon  substituting  (\ref{inp_PO_4}), (\ref{entr_noise}),  into (\ref{ftfic_is}), we obtain part (c).
   This completes the proof.

\subsection{Proof of Proposition~\ref{pro_1}}
\label{sect:app_pro_1}
Since the proof of \cite[Theorem~6.1]{kim2010} is based \cite[Lemma~6.1]{kim2010}, where the   channel input $X_t$ is expressed as $X_t=\Lambda \Big(S_t - {\bf E}\Big\{S_t\Big|Y_{-\infty}^{t-1}\Big), t=1, \ldots$, where $\Lambda$ is a nonrandom vectors, then (\ref{eql_kim}) is a necessary for \cite[Theorem~6.1]{kim2010} to hold. Next, we show Conditions 1 and 2  of Section~\ref{sect:motivation} are necessary and sufficient for equality  (\ref{eql_kim}) to hold. To avoid complex notation, we prove the claim for the realization of  Example~\ref{ex_1_1_n}.(a). Suppose the {\it initial state} $S_1$ of the noise is $S_1=S_1^s=s$, and  known to the encoder and the decoder; without loss of generality take  $s=0$, which  by  (\ref{ex_2_3_in}) implies    $V_0=0, W_0=0$ (as often done in \cite{kim2010}).  Then, the following hold.  
\begin{align}
&S_1=0 \hso \Longrightarrow  \hso V_1= W_1,  \hso S_2= W_1=V_1, \hst \mbox{by (\ref{ex_2_4_in}), (\ref{ex_2_5_in})}, \label{int_kim_1} \\
&\mbox{$(S_1=0, V_1)$ uniquely define $S_2=cS_1+W_1 =W_1=V_1$}, \hst \mbox{by (\ref{ex_2_4_in}), (\ref{ex_2_5_in})}, \\
&V_2=\Big(c-a\Big)S_2 + W_2, \hso S_3=c S_2 +W_2=cV_1+W_2,  \hst \mbox{by (\ref{ex_2_4_in}), (\ref{ex_2_5_in})}, \\
&\mbox{$(S_1=0, V_1, V_2)$ uniquely define $(S_2, S_3)$}, \\
&\mbox{repeating, then  $(S_1=0, V_1, \ldots, V_{t-1})$ uniquely define $(S_2, S_3, \ldots, S_t)$, $\forall t=3, 4, \ldots$}. \label{int_kim_2}
\end{align}
From (\ref{int_kim_1})-(\ref{int_kim_2}) it then follows, that for any $S_1=s$, including, $s=0$,  known the the encoder that the equalities hold: 
\begin{align}
{\bf P}_{X_t|X^{t-1}, Y_{-\infty}^{t-1}, S_1}=&{\bf P}_{X_t|V^{t-1}, Y_{-\infty}^{t-1},S_1}, \hst \mbox{by $Y_t=X_t+V_t$}\\
= & {\bf P}_{X_t|S^{t}, Y_{-\infty}^{t-1},S_1}, \hst t =1, \ldots, \label{eq_kim_2}
\end{align}
We can go one step further to identify the information structure of optimal channel input distributions using (\ref{eq_kim_2}), that is, to show ${\bf P}_{X_t|S^{t}, Y_{-\infty}^{t-1},S_1}={\bf P}_{X_t|S_{t}, Y_{-\infty}^{t-1},S_1}, \hst t =1, \ldots$, by repeating to proof of \cite[Theorem~1]{yang-kavcic-tatikonda2007}. However, for the statement of the proposition  this is not necessary. \\
Suppose either  $S_1=s$ is not known  to the encoder, i.e., $V_0=v_0, W_0=w_0$ are not known to the encoder, and  $S_1\neq 0$, while   the optimal channel input is  expressed as a function of the state of the noise, $S^n$: 
\begin{align}
{\bf P}_{X_t|X^{t-1}, Y_{-\infty}^{t-1}}=&{\bf P}_{X_t|V^{t-1}, Y_{-\infty}^{t-1}}= {\bf P}_{X_t|S^{t}, Y_{-\infty}^{t-1}}, \hst t =1, \ldots,\label{eq_kim_3}
\end{align}
Then by (\ref{ex_2_4_in}) and (\ref{ex_2_5_in}), it follows that $V_1= \big(c-a\big)S_1 +W_1, S_2=a S_1 +W_1$, hence knowledge of $V_1$ does not specify $S_2$, and similarly, $V^{t-1}$ does not specify $S^t$, for $t=2,3, \ldots.$ Hence, we arrive at a  contradiction of the last  equality in (\ref{eq_kim_3}). This competes the proof.

\subsection{Proof of Lemma~\ref{lem_entr_in}}
\label{sect:app_lem_entr_in}
(a) This is due to Lemma~\ref{lemma_POSS}.(v).\\
(b) By taking the per unit time limit (\ref{entropy_0}), and  utilizing the hypothesis (\ref{conv_1}), the continuity of the $\log(\cdot)$ and the fact that, for any convergent sequence $a_n, n=1,2, \ldots$, i.e., $\lim_{n\longrightarrow \infty}a_n=a$,  then $\frac{1}{n}\sum_{t=1}^n a_n \longrightarrow a$, as $n \longrightarrow \infty$, follows  (\ref{in_entr_new}).  

\subsection{Proof of Lemma~\ref{lem_pr_are_ar}}
\label{sect:app_lem_pr_are_ar}
From Corollary~\ref{cor_ex_1}.(a) we deduce that $\Sigma_t^o \tri \Sigma_t, t=1, \ldots, n$ satisfies (\ref{ar_ac,DRE}) with initial condition (\ref{ar_ac,DRE_a}). By Definition~\ref{def:det-stab} the  corresponding  generalized algebraic  Riccati equation is (\ref{ar_ac_ARE}), and   pairs $\{A,C\}$ and $\{A^*, G B^{*,\frac{1}{2}}\}$ are given by   (\ref{ar_ac,DRE_c}).  \\
(1) By Definition~\ref{def:det-stab}, for $c \neq a$ the pair  $\{A,C\}=\{c, c-a\}$ is observable, and  hence detectable. \\
(2) By Definition~\ref{def:det-stab}, the pair $\{A^*, G B^{*,\frac{1}{2}}\}=\{a, 0\}$ is unit circle controllable if and only if $|a|\neq 1$. \\
(3) By Definition~\ref{def:det-stab}, the pair $\{A^*, G B^{*,\frac{1}{2}}\}=\{a, 0\}$ is stabilizable if and only if $a \in (-1,1)$. \\
(4) This follows from Theorem~\ref{thm_ric}.(1) and parts (1), (2) and (3). Since (\ref{ar_ac_ARE}) is a quadratic equation we can  verify the two solutions are $\Sigma^\infty=0$ and $\Sigma^\infty=\frac{K_W\big(a^2-1\big)}{\big(c-a\big)^2}$, and the statement of (\ref{are_2_sol}).  \\
(5) For values $c\in (-\infty,\infty)$ and $|a|<1$,  the pair $\{A,C\}=\{c, c-a\}$ is detectable and the  pair $\{A^*, G B^{*,\frac{1}{2}}\}=\{a, 0\}$ is stabilizable, and the statement follows from  Theorem~\ref{thm_ric}.(3).\\
(6) (\ref{en_r_ar_ac}) follows from   Lemma~\ref{lem_entr_in}.(b), by invoking  Corollary~\ref{cor_ex_1}.(a), i.e.,  $K_{\hat{I}_t}= \big(c-a\big)^2 \Sigma_t^o +K_{W}, t=1, \ldots, n$, where $\Sigma_t^o$ is generated by  (\ref{ar_ac,DRE}),  and part (5).

\subsection{Proof of Corollary~\ref{cor_kim}}
\label{sect:app_cor_kim}
First, note that the analog of Theorem~\ref{thm_FTFI}.(a), for the code $(s, 2^{n R},n)$,   $n=1,2, \ldots$ is  (\ref{FTFI_CIS_1}) and (\ref{FTFI_CIS_2}), because  $P_t(dx_t|x^{t-1}, y^{t-1},s)=\overline{P}_t(dx_t|v^{t-1}, y^{t-1},s), t=1, \ldots, n$. Define  $\overline{\cal P}_{[0,n]}^s(\kappa)$  as in  (\ref{FTFI_CIS_2}) with $x^{t-1}$ replaced by $v^{t-1}, t=1, \ldots, n$.\\
(a) Then 
\begin{align}
{P}_t(dx_t|x^{t-1}, y^{t-1},s)=&{P}_t(dx_t|v^{t-1}, y^{t-1},s_0), \hso t=1,\ldots,n, \hso \mbox{by $Y_t=X_t+V_t$}\\
=&\overline{P}_t(dx_t|s^{t}, y^{t-1},s), \hst \mbox{Definition~\ref{ass_SSUN}}. 
\end{align}
The  PO-SS realization, for fixed $S_1^s=s$ is  then 
\begin{align}
&V_t= C_t S_{t}^s +N_t W_t, \hso S_1^s=s, \hso t=1, \ldots, n,  \label{noise_is_1}\\
&S_{t+1}^s= A_t S_t^s + B_t W_t, \hso S_1^s=s.
\end{align}
 Then 
\begin{align}
{\bf P}_t(dy_t\big|x^t, y^{t-1},s) =&{\bf P}_t(dy_t\big|x^t,v^{t-1}, y^{t-1},s) \\
=&{\bf P}_t(dy_t\big|x^t,s^t, y^{t-1},s), \hst \mbox{by Definition~\ref{ass_SSUN}, (A1)}\\
 =&  {\bf P}_t(dy_t\big|x_t, s^{t}), \hst \mbox{by $Y_t=X_t+ V_t$ and (\ref{noise_is_1})}\\
 \sr{(a)}{=}&{\bf P}_t(dy_t|x_t,s_{t}, s),\hst \mbox{by mutually independence of $(W_1, \ldots, W_n, S_1^s)$}. \label{chan_st} 
\end{align}
%
The probability distribution ${\bf P}_t(dy_t|y^{t-1}, s)$ is then given by 
\begin{align}
 {\bf P}_t^{\overline{P}}(dy_t\big|y^{t-1}, s) = &\int {\bf P}_t(dy_t\big|x_t, s_t) {\bf P}_t(dx_t|s_t, y^{t-1},s) \nonumber \\
 &\otimes {\bf P}_t^{\overline{P}}(ds_t\big|y^{t-1},s), \hso t=1, \ldots, n, \hst \mbox{by reconditioning and (\ref{chan_st}).} \label{out_1}
 \end{align}
The pay-off is the sum of conditional entropies $\sum_{t=1}^n H(Y_t|Y^{t-1},s)$, and the constraint is (\ref{FTFI_CIS_2}). By  Definition~\ref{def_nr_2},  the state $S_t^s, t=2, \ldots, n$ is Markov, ${\bf P}_{S_t^s|S^{s,t-1}}={\bf P}_{S_t^s|S_{t-1}^s}, t=2, \ldots,n$. By (\ref{out_1}) and the Markov property of $S_t^s, t=1, \ldots, n$, follows at each time $t$, the input distribution ${\bf P}_t^{\overline{P}}(ds_t|y^{t-1},s)$ depends on ${\bf P}_j(dx_j|s_j, y^{j-1},s), j=1, \ldots, t-1$ and not on $\overline{P}_j(dx_j|s^{j}, y^{j-1},s), j=1, \ldots, t-1$. By (\ref{is_dis_1_a}) and (\ref{is_dis_1}),  then ${\bf P}_t(dx_t|s_t, y^{t-1},s)=\overline{P}_t^M(dx_t|s_t, y^{t-1},s), t=1, \ldots, n$. It is  noted that (\ref{is_dis_1_a}) and (\ref{is_dis_1}) also follow from a slight variation of the derivation given in   \cite[Theorem~1]{yang-kavcic-tatikonda2007}. By the maximum entropy principle of Gaussian distributions it then follows that the distribution  $\overline{P}_t^M(dx_t|s_{t}, y^{t-1},s)$,  is  conditionally Gaussian, 
with linear conditional mean and nonrandom conditional covariance. 
Then (\ref{mean_1_new})  follows by repeating  step 2 of the derivation of  
Theorem~\ref{thm_SS}. This completes the derivation of all statements of part (a).\\
(b), (c). The statements follow from part (a),  using the generalized Kalman-filter, as in  Theorem~\ref{thm_SS}.

\bibliographystyle{IEEEtran}
\bibliography{Bibliography_capacity}

\end{document}